\documentclass[11pt]{article}
\pdfoutput=1
\usepackage{amssymb}
\usepackage{graphicx}
\usepackage{url}
\usepackage{setspace}
\usepackage[pdftex,bookmarksnumbered,bookmarksopen,hidelinks,breaklinks=true,colorlinks,citecolor=blue,linkcolor=blue]{hyperref}
\usepackage{xcolor}

\usepackage{times}
\usepackage{enumitem}
\usepackage{varwidth}
\usepackage{graphicx}
\usepackage{wrapfig}
\usepackage{enumerate}
\usepackage[noend]{algpseudocode}

\usepackage[utf8]{inputenc} %
\usepackage[T1]{fontenc}    %
\usepackage{booktabs}       %
\usepackage{amsfonts}       %
\usepackage{nicefrac}       %
\usepackage{microtype}      %

\usepackage[tight]{subfigure}
\usepackage{appendix}
\usepackage{amsmath,amsfonts,amsthm}
\usepackage[algo2e,ruled,vlined]{algorithm2e}
\usepackage{mdwlist}
\usepackage{xspace}
\usepackage{color}
\usepackage{mathrsfs}

\usepackage{booktabs}
\usepackage{comment}
\usepackage{geometry}

\usepackage{multirow}

\usepackage{natbib}
\newcommand{\Appendix}[1]{the full version for}

\newcommand{\linelabel}[1]{}

\newtheorem{theorem}{Theorem}[section]
\newtheorem{lemma}[theorem]{Lemma}
\newtheorem{proposition}[theorem]{Proposition}

\newtheorem{definition}{Definition}

\renewcommand{\comment}[1]{}

\algnewcommand\algorithmicinput{\textbf{Input:}}
\algnewcommand\algorithmicoutput{\textbf{Output:}}
\algnewcommand\INPUT{\item[\algorithmicinput]}
\algnewcommand\OUTPUT{\item[\algorithmicoutput]}
\algnewcommand{\LineComment}[1]{\Statex \(\triangleright\) #1}

\title{Recovery from Non-Decomposable Distance Oracles}

\author{Zhuangfei Hu\thanks{University of Waterloo. Email: zhuangfei.hu@uwaterloo.ca} \and Xinda Li\thanks{University of Waterloo. Email: xinda.li@uwaterloo.ca} \and David P. Woodruff\thanks{Carnegie Mellon University. Email: dwoodruf@cs.cmu.edu} \\
\and
 Hongyang Zhang\thanks{University of Waterloo. Email: hongyang.zhang@uwaterloo.ca
}
\and
Shufan Zhang\thanks{University of Waterloo. Email: shufan.zhang@uwaterloo.ca}
\footnote{The preliminary version of this work was accepted and presented at ITCS 2023 \citep{hu_et_al:LIPIcs.ITCS.2023.73} and the full version has been accepted by IEEE Transactions on Information Theory \citep{hu2023recoveryTIT}.}
}

\date{}

\allowdisplaybreaks

\usepackage{amsmath}
\usepackage{amsfonts}
\usepackage{bm}
\usepackage{mathtools}
\usepackage{booktabs} %

\usepackage[utf8]{inputenc}

\usepackage{algorithm}
\usepackage{algorithmicx}

\newtheorem{example}{Example}[section]

\makeatletter
\def\renewtheorem#1{%
  \expandafter\let\csname#1\endcsname\relax
  \expandafter\let\csname c@#1\endcsname\relax
  \gdef\renewtheorem@envname{#1}
  \renewtheorem@secpar
}
\def\renewtheorem@secpar{\@ifnextchar[{\renewtheorem@numberedlike}{\renewtheorem@nonumberedlike}}
\def\renewtheorem@numberedlike[#1]#2{\newtheorem{\renewtheorem@envname}[#1]{#2}}
\def\renewtheorem@nonumberedlike#1{  
\def\renewtheorem@caption{#1}
\edef\renewtheorem@nowithin{\noexpand\newtheorem{\renewtheorem@envname}{\renewtheorem@caption}}
\renewtheorem@thirdpar
}
\def\renewtheorem@thirdpar{\@ifnextchar[{\renewtheorem@within}{\renewtheorem@nowithin}}
\def\renewtheorem@within[#1]{\renewtheorem@nowithin[#1]}
\makeatother

\theoremstyle{definition}
\renewtheorem{definition}{Definition}[section]

\newcommand{\bits}{\{0,1\}}
\newcommand{\mc}[1]{\mathcal{#1}}

\newcommand{\mb}[1]{\mathbf{#1}}
\newcommand{\mbb}[1]{\mathbb{#1}}

\newtheorem{claim}[theorem]{Claim}

\newcommand{\frechet}{Fr\'echet\xspace}

\newcommand{\revise}[1]{{#1}}

\newcommand{\needreview}[1]{{#1}}

\newcommand{\eat}[1]{}

\newcommand{\stitle}[1]{\smallskip \noindent{\bf #1}}

\usepackage{xcolor}
\newcounter{margin} %
\definecolor{xcolor}{rgb}{0.55, 0.71, 0.0}
\definecolor{hcolor}{rgb}{1.0, 0.08, 0.58}

\usepackage{fancyvrb} %

\usepackage{listings}
\usepackage{cprotect}

\DeclareMathOperator\MSS{MSS}

\DeclareMathOperator\DTW{DTW}

\DeclareMathOperator\dist{dist}

\DeclareMathOperator\len{len}

\DeclareMathOperator\Cost{Cost}

\DeclareMathOperator{\poly}{\mathsf{poly}}
\DeclareMathOperator{\polylog}{\mathsf{polylog}}

\newif\ifarxiv
\arxivtrue     %

\newif\ifTIT
\TITfalse   %

\begin{document}

\maketitle

\begin{abstract}
A line of work has looked at the problem of recovering an input from \emph{distance queries}. In this setting, there is an unknown sequence $s \in \{0,1\}^{\leq n}$, and one chooses a set of queries $y \in \{0,1\}^{\mc{O}(n)}$ and receives $d(s,y)$ for a distance function $d$. The goal is to make as few queries as possible to recover $s$. Although this problem is well-studied for \emph{decomposable} distances, i.e., distances of the form $d(s,y) = \sum_{i=1}^n f(s_i, y_i)$ for some function $f$, which includes the important cases of Hamming distance, $\ell_p$-norms, and $M$-estimators, to the best of our knowledge this problem has not been studied for non-decomposable distances, for which there are important special cases such as edit distance, dynamic time warping (DTW), \frechet distance, earth mover's distance, and so on. We initiate the study and develop a general framework for such distances. Interestingly, for some distances such as DTW or \frechet, exact recovery of the sequence $s$ is provably impossible, and so we show by allowing the characters in $y$ to be drawn from a slightly larger alphabet this then becomes possible. In a number of cases we obtain optimal or near-optimal query complexity. We also study the role of adaptivity for a number of different distance functions.  One motivation for understanding non-adaptivity is that the query sequence can be fixed and the distances of the input to the queries provide a non-linear embedding of the input, which can be used in downstream applications involving, e.g., neural networks for natural language processing.
\end{abstract}

\newpage
\tableofcontents

\newpage

\section{Introduction}
\ifarxiv
We study the problem of exact recovery of a sequence from queries to a distance oracle.
\fi
Suppose there is an unknown input sequence $s$ with length at most $n$, defined on a binary alphabet $\{0,1\}$.
Assume we have a distance oracle which returns the distance $d(s,q)$ between a query sequence $q$ and the unknown sequence $s$, where the query sequence $q$ is chosen either adaptively or non-adaptively.
The problem is to determine the sequence $s$ with a minimal number of queries to the distance oracle.
\needreview{This problem has been studied for decomposable distances, that is, the distance function between two sequences can be computed as the sum of distances between pairs of characters at the same entry, but never for non-decomposable distances.}
Among all non-decomposable distances, we are particularly interested in the edit distance, ($p$)-Dynamic Time Warping ($p$-DTW), and \frechet distances.
The edit distance measures the minimum number of edit operations (i.e., insertions, deletions, and substitutions) for transforming one sequence to another.
\revise{\linelabel{intro-dtw}The $p$-DTW distance ($1\le p<\infty$) between two sequences $x, y$ is defined as the minimum $\ell_p$ distance between two equal-length expansions of $x, y$, where the expansion of a sequence means you can duplicate each character of each sequence an arbitrary number of times. 
When $p=1$, the $p$-DTW distance is called the DTW distance.}
If we consider the $\ell_\infty$ norm instead of the $\ell_p$ norm, we obtain the \frechet distance.

The problem of exact recovery for \emph{decomposable distances} is well-studied in the literature, \needreview{under the names of} the coin-weighing problem \citep{shapiro1960e1399,bshouty2009optimal} and the group testing problems \citep{dorfman1943detection,aldridge2019group, coja2020optimal}.
The coin-weighing problem is to identify the weight of each coin from a collection of $n$ coins, each being of weight either $w_0$ or $w_1$ ($w_0$ and $w_1$ are distinct). In this problem, our only access to the coins is via weighing a subset of the coins on a spring scale.
The group testing problem has also been shown to be equivalent to the coin-weighing problem in some settings  \citep{wang2017optimal}.
This line of research has been extensively studied with interesting applications.
For example, the coin-weighing problem can be found in the detection problem \citep{soderberg1963combinatory}, the problem of determining a collection \citep{cantor1966determination}, and the distinguishing family problem \citep{li1991combinatorics}.

The query complexity of the adaptive version of the problem is also related to the original Mastermind game \citep{knuth1976computer}.
The Mastermind problem can be phrased as guessing an input sequence based on Hamming distance queries. 
The non-adaptive version of \needreview{this problem can be shown to be equivalent to the well-studied non-adaptive coin-weighing problem \citep{bshouty2009optimal}.} 
One can then consider other variants of the Mastermind game where the input sequence is guessed based on other distance metrics, such as permutation-based distances \citep{afshani2019query}, $\ell_p$ distances \citep{fernandez2019query} and graph distances \citep{rodriguez2014strong, jiang2019metric}. 
However, general distance metrics that do not decompose into coordinate sums are less understood. 
In this paper, we initiate the study of this exact recovery problem on  \emph{non-decomposable} distances.

One motivation of our exact recovery problem is its application to adversarially robust learning on  discrete domains.
It is well-known that deep neural networks are vulnerable to adversarial examples: test inputs that have been modified slightly in the $\ell_p$ space can lead to problematic machine predictions.
Though there exist various techniques such as Pixel-DP~\citep{lecuyer2019certified} and randomized smoothing~\citep{DBLP:conf/icml/CohenRK19} that achieve certified robustness against $\ell_p$-norm perturbations in continuous domains, in many tasks such as natural language processing, the $\ell_p$ norm is not well-defined for discrete perturbations. 
To resolve this issue, inputs from a discrete domain are usually mapped to vectors in the $\ell_p$ space before being passed to a classifier; this is also known as a word embedding. 
We require two properties of such a mapping: 1) zero information loss; 2) Lipschitzness with respect to the distance metric in the input space. 
We show that the exact recovery problem yields a direct construction of such mappings: suppose the set of query sequences is $\{q_1, \ldots, q_m\}$ and $s$ is the unknown input sequence; the mapping for $s$: $\phi{(s)} = [d(s, q_1), \ldots, d(s, q_m)]$ has Lipschitz constant at most $\sqrt{m}$ (in the $\ell_2$ norm) and maintains complete information about $s$.
Similar to edit distance, which can be used for describing the adversarial capability in changing sequences, the DTW and \frechet distances have received significant attention for their flexibility in handling temporal sequences.
The special instance of our problem on DTW and \frechet distances may be useful for analyzing the robustness of DTW neural networks \citep{cai2019dtwnet}.

A distance embedding further inspires theoretical applications in functional analysis~\citep{vershynin2011lectures}. While the space of input sequence $s$ is a metric space, it may not be a Hilbert space with a definition of norm and inner product. Our result provides us with a tool to define a mapping from a metric space to a Hilbert space without loss of information about the input sequences. One can then use the norm or inner product to analyze input sequences, e.g., when two input sequences are orthogonal and how to normalize an input sequence to have norm $1$.

\comment{
\subsection{Our Assumptions}
\hy{I don't think this section says much. How about deleting it?}
We start from clarifying the assumptions made in this paper. 
Throughout the paper, we assume the alphabet of the unknown input sequence $s$ is \bits. 
To recover this sequence, we submit adaptive or non-adaptive query sequences to the distance oracles where these query sequences may be allowed to utilize alphabets outside \bits. 
We assume the maximum length of $s$ is $n$ while the exact length of $s$ is unknown. 
In the solution to our problems, we introduce one additional character in the query sequence to edit distance and $\mathcal{O}(n)$ additional characters in the query sequence to the DTW distance \hy{This is not true}.
For edit distance, introducing more than 1 extra characters cannot encode more information in the query results, because the edit distance between 0 (or 1) and any other additional character is the same \hy{edit distance between two characters?}.
}

\subsection{Our Contribution and Results}

\revise{To the best of our knowledge, this paper makes the first effort to consider the non-decomposable distance recovery problem.
We first present a general framework to tackle with this problem, and then exhaustively explore representative distances of this class, i.e., edit distance, DTW distance, and \frechet distance.
We also study the role of adaptivity and non-adaptivity and obtain a number of results on lower bounds and upper bounds of query complexity.
Before introducing our technical results, we would like to clarify the assumptions we make in the setting of the problem and justify some of them.
}

\smallskip

\noindent\textbf{Assumptions.}
Throughout the paper, we assume the alphabet of the unknown input sequence $s$ is \bits. We note that under this assumption, all of our results for DTW described below will apply to $p$-DTW.
To recover the sequence $s$, we submit adaptive or non-adaptive query sequences to a distance oracle. As we will show in Section \ref{sec:hardness}, for some distance metrics, there exist input sequences that cannot be distinguished by any sequence on a binary alphabet. Therefore, our query sequences may be allowed to utilize alphabets outside \bits with $\mc{O}(1)$ extra characters to exactly recover the input sequence. 
For edit distance, the extended alphabet can contain \emph{any symbol} outside the binary alphabet, as the edit distance oracle counts the edit operations no matter what symbol is used.
For ($p$-)DTW distance and \frechet distance, the extended alphabet can consist of \emph{any real number}.
We assume the maximum length of $s$ is $n$, while the exact length of $s$ is unknown.

\smallskip

\noindent\textbf{Extension to non-binary inputs.}
\revise{
The binary input sequence setting is not an over-simplified assumption.
All the results we obtain on the binary setting can be naturally extended to any non-binary alphabet $\Sigma$ by encoding the non-binary alphabet in a binary domain.
This will increase the query complexity by a constant factor from $|\Sigma|$ (one-hot encoding) to $\log(|\Sigma|)$ (binary encoding).
Though this may not be the best solution if one considers a large alphabet, this extension works for the results for all distance metrics shown in this paper.
Improvement on this extension to the recovery problem leaves room for future research.
}

\smallskip

\stitle{Optimality.}
\revise{\linelabel{optimality}Throughout the paper, we consider asymptotic optimality, that is, the asymptotic complexity lower and upper bounds match orderwise.
We would like to investigate lower bounds of the problem per distance instance, and develop algorithms that shows upper bounds can match lower bounds up to constant factor or logarithmic factor (under Big-O / Big-Omega tilde notation).}

\smallskip

\revise{To list the results we obtain on this non-decomposable distance recovery problem,}
we begin with a general coordinate descent framework that can help recover sequences from a large class of distance oracles, including but not limited to earth mover's distance (EMD), cascaded norms ($\ell_p$ of $\ell_q$), and $A$ norms (a.k.a. Mahalanobis distance).
We then present improved results on three specific distance metrics: edit distance, DTW distance, and \frechet distance. We first provide several observations on the sequence recovery problem, showing the existence of indistinguishable input sequences despite the fact that we can query their DTW and Fr\'echet distances with all possible binary query sequences. 
We also prove lower bounds on the query complexity in our distance recovery problem w.r.t. DTW, edit, and \frechet distances.
Then we present our main results on recovering sequences from edit, DTW, and \frechet distance oracles, with adaptive and non-adaptive strategies.

\medskip

\subsubsection{Existence of Indistinguishable Sequences}
\label{sec:hardness}

We observe that, for some distances, there exist sequences that cannot be distinguished by any query sequence over a binary alphabet.
This can be proved by showing concrete examples, i.e., a pair of sequences that cannot be distinguished, which we show is true for the DTW and the \frechet distances, as stated in the following theorem.

\begin{theorem}[Informal, existence of indistinguishable sequences]
There exists a pair of sequences $(s, s')$ such that $s$ and $s'$ cannot be distinguished by any query sequence on a binary alphabet, for the DTW distance and the \frechet distance.
\end{theorem}

The formal proof of this theorem for the DTW distance is deferred to Theorem \ref{theorem:dtw_hardness}.
The analogous discussion for the \frechet distance can be found in Section \ref{sec:frechet}.
Due to the existence of indistinguishable sequences, we define the concept of an \emph{equivalence class} of sequences, which is a set of input sequences which are indistinguishable from all queries by a given distance oracle.

This observation suggests the scope of the distance recovery problem we study.
We further \emph{categorize the recovery guarantee into the following three levels}, from strong to weak:
1) recover the \textbf{exact input sequence};
2) recover any sequence in the \textbf{same equivalent class of the input sequence}, where the equivalence class is defined to be the set for which any two input sequences in the equivalence class cannot be distinguished by calling the distance oracle to all query sequences;
3) recover any sequence which has \textbf{zero distance to the input sequence}.
While the third level is the weakest one, in certain cases it can be reduced to the first two levels---for norm-induced distance functions, the recovered sequence is exactly the input sequence; for semi-norm-induced distance functions, the recovered sequence is in the same equivalence class.
\revise{For other distance functions which are not \emph{metric}, recovering a sequence with zero distance to input does not necessarily imply any one of the first two levels.}
We will show that our general coordinate descent framework can recover sequences with the third-level guarantee.

\bigskip

\subsubsection{General \emph{Coordinate Descent} Framework for Adaptively Querying Distance Oracles}

We develop a general framework for recovering an input sequence from adaptive queries, which models the problem as a \emph{zero-th order optimization} and utilizes a coordinate-descent-based algorithm to give a solution.
The \emph{coordinate descent} framework defines the distance between the input sequence and the query sequence as the loss function. The objective of the optimization is to reduce the loss function to $0$, which guarantees what we call the \emph{third level of recovery}.
We define a \emph{step operation} to modify the query sequence. For example, in the context of edit distance, a step operation is defined as adding/removing/substituting a character of the query sequence.
To perform coordinate descent, our algorithm performs one step operation each time and queries the oracle to find a direction for which the loss decreases by \revise{at least a pre-determined constant scalar}.
By iteratively performing this method, the loss can be reduced to $0$ and we show that the overall complexity of this method is $\poly(n)$, given that the maximum length of the sequence is $n$.
For a large class of non-decomposable distance functions, such as the earth mover's distance (EMD), the cascaded norm ($\ell_p$ of $\ell_q$), and the $A$ norm, we can use this framework to yield a solution, as stated in the following theorem.

\begin{theorem}[Coordinate Descent for Adaptive Distance Queries]
For an arbitrary distance oracle, a binary alphabet $\bits$ and any input sequence $s \in \bits^i$ where $0 \leq i \leq n$, using \emph{coordinate descent} can \emph{reduce the distance to} the input sequence $s$ \emph{to 0}, by adaptively querying the distance oracle between $s$ and a set of query sequences with query complexity at most $\poly(n)$.
\end{theorem}

Sufficient conditions for using this framework and further details can be found in Theorem \ref{theorem:adaptive_general}.

\bigskip

\subsubsection{Lower Bounds on the Recovery Problem}

If we study the problem of exact recovery (the first level of recovery), we can obtain an information-theoretic lower bound of  $\Tilde{\Omega}(n)$ for various distance oracles, given by the following theorem. 
Here $f(n) = \Tilde{\Omega}(g(n))$ if \revise{\linelabel{poly-1}$f(n) = \Omega(g(n)/\polylog(n))$}. 

\begin{theorem}[Lower Bounds for Exact Recovery]
\label{theorem:lower_bound_exact}
For any input sequence $s \in \bits^i$ where $0 \leq i \leq n$, if for any input sequence and query the distance oracle has $\poly(n)$ possible values, any algorithm which \emph{exactly recovers} $s$ by querying the distance oracle between $s$ and a set of query sequences requires query complexity at least $\Tilde{\Omega}(n)$. 
\end{theorem}

The idea behind this bound is that, there are exponentially many possible input sequences with length at most $n$, while \revise{for the distance oracles given in our setting,} the output of each query is a distance between two sequences which only has $\poly(n)$ possibilities. Hence, we need at least 
\revise{$\log_{\poly(n)}(2^{n+1})=\Tilde{\Omega}(n)$}
queries.
We instantiate this theorem on the edit distance and DTW distance in Theorem \ref{thm:edit_lower_bound} and Theorem \ref{thm:dtw_lower_bound}, for recovery to the exact input distance.

We note for the DTW distance and \frechet distance, there exist indistinguishable sequences, which lead to the recovery problem for equivalence class.
Since \needreview{the total number of equivalence classes is less than the number of input sequences, the previous counting technique (based on simple facts from information theory) no longer works}. So we need a different argument, as we give in the following theorem:

\begin{table*}[tbp]
\centering
\caption{Summary of our results for recovering arbitrary input sequences of length $n$ under the constraint that the query length is of $\mc{O}(n)$. LB: Lower Bound. \#EC: Number of Extra Characters.}
\label{table:summary}
{
\centering
\resizebox{1.0\textwidth}{!}{%
\begin{tabular}{c|c|c|c|l|c|c}
\hline
Oracle & Query Complexity  & LB & \eat{Query Length &} Adaptive? & \#EC &  Level of Recovery & Positions\\
\hline
\hline
Edit     &    $\revise{2k \log (n/k) + k + \log n + c}$ or $n+2$   & $\Tilde{\Omega}(n)$ & \eat{$\mc{O}(n)$ &}  Adaptive &  0 & Exact sequence & Theorems \ref{theorem:adaptive_edit}\&\ref{theorem:adaptive_edit_N+1} \\ 
Edit     &    $n+1$   & $\Tilde{\Omega}(n)$ & \eat{$\mc{O}(n)$ &}  Non-adaptive &  1 &  Exact sequence & Theorem \ref{theorem:edit} \\ 
Edit     &    \revise{$\frac{1}{2}(n^2 + 3n)$}   & $\Tilde{\Omega}(n)$ & \eat{$\mc{O}(n)$ &}  Non-adaptive &  0 &  Exact sequence & Theorem \ref{theorem:edit_non_adaptive_n2} \\ 
($p$-)DTW      &    $n+1$   & $\Tilde{\Omega}(n)$  & \eat{$\mc{O}(n)$ &}  Adaptive  & 1  & Exact sequence & Theorem \ref{thm:dtw_adaptive}  \\ 
($p$-)DTW      &    \revise{$2n$}   & $\Omega(n)$  & \eat{$\mc{O}(n^2)$ &}  Non-adaptive  & 0 & Equivalent class  & Theorem \ref{theorem:oz_equivalence}  \\ 
($p$-)DTW      &    $n^2+n$   & $\Tilde{\Omega}(n)$  & \eat{$\mc{O}(n)$ &}  Non-adaptive  & 1  &  Exact sequence & Theorem \ref{theorem:dtw_1e}  \\ 
($p$-)DTW      &    $n+2$   & $\Tilde{\Omega}(n)$  & \eat{$\mc{O}(n)$ &}  Non-adaptive  & 2\textsuperscript{*} &  Exact sequence & Theorem \ref{theorem:dtw_o1}   \\ 
\frechet  &    $2n - 1$  & $2n - 1$ & \eat{$\mc{O}(n)$ &}  N/A\textsuperscript{\textdagger}  & 0\textsuperscript{**}  & Equivalent class & Theorem \ref{theorem:frechet} \\ 
Any distance  &   $\poly(n)$  & - & \eat{$\mc{O}(n)$ &}  Adaptive  & 0 & Zero distance to input & Theorem \ref{theorem:adaptive_general}  \\ 
\hline
\end{tabular}
}
}
\flushleft{
\rule{0in}{1.2em}{\footnotesize \textsuperscript{\textdagger} For both adaptively and non-adaptively querying the \frechet distance oracle, the optimal bound on the query complexity is $2n-1$.
}

\rule{0in}{1.2em}{\footnotesize \textsuperscript{*} Increasing \#EC from $2$ to an arbitrary constant cannot improve the query complexity to be better than $\mc{\tilde{O}}(n)$.
}

\rule{0in}{1.2em}{\footnotesize \textsuperscript{**} Involving extra characters not only cannot improve the level of recovery from ``equivalence class'' to ``exact sequence'', but also cannot improve the query complexity (see Theorem \ref{thm:frechet_extra_not_useful}).
}
}
\end{table*}

\begin{theorem}[Lower Bounds for Equivalence Class Recovery]
\label{theorem:lower_bound_equivalence}
For a binary alphabet $\bits$ and any input sequence $s \in \bits^i$ where $0 \leq i \leq n$, any algorithm which \needreview{recovers the sequence $s$ up to equivalence} by querying the DTW or \frechet distance oracle between $s$ and a set of query sequences requires query complexity at least ${\Omega}(n)$. 
\end{theorem}

We highlight our techniques used in proving this lower bound in Section \ref{sec:our_techniques}, while the formal proof can be seen in Theorem \ref{thm:dtw_no_extra_lowerbound} and Theorem \ref{thm:frechet_lower_bound}.

\bigskip

\subsubsection{Adaptively Querying Distance Oracles, Optimally}

\noindent
We first answer the distance recovery problem with adaptive query strategies.
Our solutions are summarized in the theorem below.

\begin{theorem}[Upper Bounds for Adaptive Exact Recovery]
\label{theorem:upper_bound}
For a binary alphabet $\bits$ and any input sequence $s \in \bits^i$ where $0 \leq i \leq n$, there exists an algorithm which can \emph{exactly recover} the input sequence $s$, by \emph{adaptively} querying the distance oracle (for the \emph{edit} and \emph{DTW} distances) between $s$ and a set of query sequences with query complexity at most $\mc{O}(n)$.
\end{theorem}

All results in Theorem \ref{theorem:upper_bound} match our lower bounds on the query complexity.
Without extra character(s), using the DTW distance oracle we can only recover a sequence in the same equivalence class.
Our result in Theorem \ref{theorem:upper_bound} for the DTW distance is achieved with the assistance of $1$ extra character outside the alphabet $\{0,1\}$, and the proof and algorithm can be found in Theorem \ref{thm:dtw_adaptive}.

For the edit distance, we have two different adaptive algorithms that can achieve the $\mc{O}(n)$ bound.
The first algorithm makes use of the property that, for two sequences, the edit distance is equal to the difference in their lengths, if and only if one sequence is a subsequence of the other.
We construct an $\mc{O}(n)$ adaptive query set and a binary search algorithm utilizing this property to recover the input sequence.
Our second algorithm instead queries the length of the input sequence by an empty sequence and then finds a set of $\mc{O}(n)$ bases as the query set, from which we can reconstruct the input sequence.
These are further detailed in Theorem \ref{theorem:adaptive_edit} and Theorem \ref{theorem:adaptive_edit_N+1}.

For the \frechet distance, adaptive and non-adaptive strategies are essentially the same, because we prove that $2n-1$ queries are necessary and sufficient for recovering from a \frechet distance oracle.
However, we can only recover a sequence in the equivalence class in this setting. This result is described as a non-adaptive query strategy in Theorem \ref{theorem:frechet}.

\bigskip

\subsubsection{Non-adaptively Querying Distance Oracles, Optimally}

Next we describe our non-adaptive query strategies for the distance recovery problem.
Theorem \ref{thm:upper_bound_non_adaptive_exact} shows upper bounds for exact sequence recovery, while Theorem \ref{thm:upper_bound_non_adaptive_equivalent_class} summarizes our results on the recovery problem of finding a sequence in the same equivalence class as the input sequence.

\begin{theorem}[Upper Bounds for Non-adaptive Exact Recovery]
\label{thm:upper_bound_non_adaptive_exact}
For a binary alphabet $\bits$ and any input sequence $s \in \bits^i$ where $0 \leq i \leq n$, there exists an algorithm which can \emph{exactly recover} the input sequence $s$, by querying the distance oracle (for the \emph{edit} and \emph{DTW} distances) between $s$ and a \emph{non-adaptive} set of query sequences with query complexity at most $\mc{O}(n)$, with the assistance of $\mc{O}(1)$ extra characters in the query sequences. 
\end{theorem}

With $1$ extra character, we show the construction of a set of non-adaptive queries that can exactly recover sequences from the edit distance (Theorem \ref{theorem:edit}), while with $2$ extra characters, we can exactly recover input sequences from the DTW distance (Theorem \ref{theorem:dtw_o1}).
Both results match our lower bound on the query complexity, while we complement our results with an $\mc{O}(n^2)$ query complexity algorithm for the DTW distance with $1$ extra character (Theorem \ref{theorem:dtw_1e}). 
We note that non-adaptive strategies have limited power compared to adaptive strategies. Hence, we \needreview{consider adding extra characters to construct query strategies that are comparable to the lower bounds.} 
For the edit distance, introducing more than $1$ extra characters cannot encode more information in the query results, because the cost between $0$ (or $1$) and any other additional character is always the same.

\begin{theorem}[Upper Bounds for Non-adaptive Equivalence Class Recovery]
\label{thm:upper_bound_non_adaptive_equivalent_class}
For a binary alphabet $\bits$ and any input sequence $s \in \bits^i$ where $0 \leq i \leq n$, there exists an algorithm which can \emph{recover the sequence in the same equivalence class as} the input sequence $s$, by querying the distance oracle (for the \emph{DTW} and \emph{\frechet} distances) between $s$ and a \emph{non-adaptive} set of query sequences with query complexity at most $\mc{O}(n)$, without extra characters in the query sequence.
\end{theorem}

By Theorem \ref{thm:upper_bound_non_adaptive_equivalent_class}, if we are not allowed to use extra characters, we can only recover the sequence in the same equivalence class as the input sequence for the DTW distance.
Our query construction and proof are shown in Theorem \ref{theorem:oz_equivalence}.
We also remark that for \frechet distance,
using extra characters cannot help to improve the results of Theorem \ref{theorem:frechet}, as shown in Theorem \ref{thm:frechet_extra_not_useful}.

\stitle{Summary.}
The main technical results of this paper are summarized in Table \ref{table:summary}.

\subsection{Paper Roadmap}
The remainder of the paper is organized as follows.
Section \ref{sec:background} introduces the notations and essential background definitions (regarding sequence, distances, and matching properties) used in this paper.
Section \ref{sec:our_techniques} highlights the techniques and insights behind our proofs of non-adaptively querying the DTW distance oracle, which helps the understanding of the most non-trivial and interesting parts of this paper.
Section \ref{sec:adaptive} consists of our results on the recovery problem with adaptive queries, which begin with a general framework for all non-decomposable distances and follow by instantiations as per distance using specific properties.
We present and discuss our results on the lower bounds and upper bounds of query complexity for recovery with non-adaptive queries on edit distance, DTW distance, and \frechet distance, with different recovery guarantees, in Section \ref{sec:edit_non}, Section \ref{sect:dtw} and Section \ref{sec:frechet}, respectively.
Section \ref{sec:related_work} summarizes the related papers to our problem.
As an initiation of this line of study in the recovery of non-decomposable distances, we finally describe the yet-open problems in Section \ref{sec:open_problems}.

\section{Preliminaries}
\label{sec:background}

\revise{\phantomsection\linelabel{notation-0}We would like to briefly introduce the fundamental concepts, definitions and notations that are involved in this paper. An alphabet is a finite set of characters. A binary alphabet contains two elements, $\Sigma_b \coloneqq \bits$. A sequence is either empty $\phi$, or an enumerated collection of characters selected from a given alphabet.
We denote the length of a sequence $s$ by \phantomsection\linelabel{len-op-0}$\revise{\len}(s)$.
Throughout the paper, we use $[n]$ to denote the set $\{ 1, 2, \dots, n \}$.
Then for sequence $s$, $[\len(s)]$ represents its indices set.
Note our indices set starts from 1.

A distance function between a pair of sequences measures the similarity and the structural relationship between them.
A distance function $\dist(\cdot, \cdot)$, as a \emph{metric}, satisfies the following properties:
\begin{itemize}
    \item \emph{Identity}: $\dist(s, s') = 0$ iff $s = s'$;
    
    \item \emph{Commutativity}: $\dist(s, s') = \dist(s', s)$;

    \item \emph{Triangle inequality}: for any sequence $x$, 
    $\dist(s, s') \leq \dist(s, x) +  \dist(x, s')$;
    
    \item \emph{Non-negativity}: $\dist(s, s') \geq 0$.
\end{itemize}

Different distance functions can capture the similarity information from different perspectives.
While we use $\dist(\cdot, \cdot)$ to denote the distance metric in general, in this paper we are in particular interested in the edit distance (denoted by $d_L(\cdot, \cdot)$, $L$ for Levenshtein), ($p$)-Dynamic Time Warping ($p$-DTW) distance (denoted by $d_{\DTW}(\cdot, \cdot)$), and \frechet distance (denoted by $d_F(\cdot, \cdot)$), which are non-decomposable to a sum of coordinate-wise contributions.
We note that the widely used DTW distance is \emph{not a metric} because identity and triangle inequality properties do not hold for it.
It has been shown a generalization to $p$-th power of DTW (i.e., $p$-DTW) distance satisfies weak triangle inequality up to a factor parameterized by $p$ and the sequence length \citep{DBLP:conf/waoa/BuchinDGPR22}.
We discuss in this paper how the missing triangle inequality affects our recovery problem (especially for DTW).

There are several other definitions related to sequences that are useful in our paper.
}

\begin{table*}[tbp]
\centering
\caption{\revise{Summary of Main Notations}}
\label{table:notations}
{
\centering
\resizebox{\textwidth}{!}{%
\begin{tabular}{r|l||r|l}
\hline
Notation & Meaning of Notation  & Notation & Meaning of Notation\\
\hline
$s$ & The input sequence & $\phi$ & The empty sequence \\
$s[i]$ & The $i$-th character of sequence $s$ & $s[i, j]$ & A substring of $s$ (from the $i$-th to the $j$-th character) \\
$c^m$ & Repeating character $c$ for $m$ times  & $\len(s)$ & The length of $s$ \\
$[n]$ & $\{ 1, 2, \ldots, n \}$  & $[\len(s)]$ & The index set of $s$ \\
\textsc{lor}($s, i$) & The length of the $i$-th run of $s$ & \#\textsc{runs}($s$) & The number of runs in $s$\\
$\mc{Q}$ & Query set & $q^{(i)}$ & The $i$-th query in the query set \\
$\dist(\cdot, \cdot)$ & The general distance oracle & $d_L(\cdot, \cdot)$ & The edit distance oracle \\
$d_{\DTW}(\cdot, \cdot)$ & The DTW distance oracle & $d_F(\cdot, \cdot)$ & The \frechet distance oracle \\
$\MSS(seq, r)$ & A MSS instance & $\| \cdot \|_p$ & $\ell_p$ norm \\

\hline
\end{tabular}
}
}
\end{table*}

\begin{definition}[Runs and Expansion, \citep{DBLP:conf/compgeom/BravermanCKWY19}]
The runs of a sequence $x$ are the maximal substrings consisting of a single
repeated character. Any sequence obtained from $x$ by extending $x$’s runs is an expansion of $x$.
\revise{For a given character $c$, we use $c^{m}$ to represent the sequence obtained by repeating $c$ for $m$ times.}
We denote \revise{the length of the $i$-th run of $x$ by \textsc{lor}($x, i$), where \textsc{lor} means \emph{Length of Run} function, and} the number of runs of a sequence $x$ by \#\textsc{runs}($x$).
\end{definition}

The following definitions of a condensed expression and subsequence are useful in developing our algorithms.

\begin{definition}[Condensed Expression]
\label{def:condensed}
We say $y$ is a condensed expression of $x$ if (i) $y$ has the same number of runs as $x$, \revise{(ii) the first and last character of $y$ and $x$ are the same,} (iii) each run of $y$ only has 1 character.
\end{definition}

\begin{definition}[Subsequence and Substring]
\label{def:subsequence_substring}
Given a sequence $y$, its subsequence $x$ is derived by deleting zero or more characters from $y$ without changing the order of the remaining characters.
The substring $x'$ is a \emph{contiguous subsequence} of $y$.
We use $x[a]$ to denote the $a$-th character of the sequence $x$, and $x[a, b]$ to denote a substring of $x$ which starts from the $a$-th character and ends at the $b$-th character.
\end{definition}

As an example, consider the sequence $0010111$.
The number of runs in this sequence is $4$.
The runs of sequence $0010111$ are $00$ (the 1st run), $1$ (the 2nd run), $0$ (the 3rd run), and $111$ (the last run), with length of $2, 1, 1, 3$, respectively. 
By duplicating the characters, we can extend a run in a sequence and then obtain another sequence which is an expansion of the original one.
For instance, by extending the second run in $0010111$, we get $0011110111$ which is the expansion of sequence $0010111$.
The condensed expression of $0010111$ is the sequence $0101$.
Sequences $010$, $101$, $0111$ are subsequences (or substrings) of $0010111$, while $01111$, $000$, $1111$ are only subsequences (not substrings).

\revise{The definitions of these three distances (Edit, DTW, and \frechet) are listed as follows.}

\begin{definition}[Edit Distance, or Levenshtein Distance \citep{levenshtein1966binary}]
Given two sequences $x$ and $y$, the edit distance $d_{L}(x, y)$ equals the \textit{minimal number} of \textit{edit operations} required for a sequence $x$ to be transformed to sequence $y$.
Specifically, we consider the Levenshtein distance  \citep{levenshtein1966binary} which captures the addition, deletion, and substitution of single symbols.
\end{definition}

\revise{\linelabel{l1_abs}We use $\| \cdot \|_1$ or simply $\| \cdot \|$ to denote the $\ell_1$ norm distance between two equi-length sequences whose symbols are real numbers. The notation for absolute value $|\cdot|$ is used to calculate the cost or difference between two characters.}

\begin{definition}[DTW Distance, \citep{DBLP:conf/compgeom/BravermanCKWY19}]
\label{def:dtw}
\revise{\phantomsection\linelabel{DTW-definition}Consider two sequences $x, y$ of length $m_1$ and $m_2$, respectively.} A correspondence $(\overline{x}, \overline{y})$ 
between $x$ and $y$ is a pair of equal-length expansions of $x$ and $y$. The cost of a correspondence is calculated as the $\ell_1$ distance between $\overline{x}, \overline{y}$: $\| \overline{x} - \overline{y}\|_1$. A correspondence between $x$ and $y$ is said to be optimal if it has the minimum attainable cost, and the resulting cost is called the dynamic time warping distance $d_{\revise{\DTW}}(x, y)$, that is
$
d_{\revise{\DTW}}(x, y) = \min_{(\overline{x}, \overline{y}) \in \mc{W}_{x, y}} \| \overline{x} - \overline{y}\|_1,
$
where $\mc{W}_{x, y}$ denotes the set of all correspondences $(\overline{x}, \overline{y})$.
\end{definition}

\begin{definition}[$p$-DTW Distance, \citep{DBLP:conf/waoa/BuchinDGPR22}]
By replacing the $\ell_1$ norm in Definition \ref{def:dtw} with the $\ell_p$ norm ($1\le p<\infty$), we obtain the definition for the $p$-DTW distance.
\end{definition}

\comment{

\begin{definition}[Weak Triangle Inequality, \citep{buchin2021approximating}]
For any $m_1, m_2 \in \mbb{N}$, let $x, z \in X^{\leq m_1}, y \in X^{\leq m_2}$, and $1 < p < \infty$. Then, we have the following inequality,
$$
d_{\revise{\DTW}_p}(x, z) \leq m_1^{1/p} \cdot \big( d_{\revise{\DTW}_p}(x, y) + d_{\revise{\DTW}_p}(y, z) \big).
$$
\end{definition}
}

\revise{
In addition to the existing definitions, we need to introduce some new concepts essential to our proofs for ($p$)-DTW distance.

\begin{definition}[Monotonic Sequence]
\revise{\linelabel{ind-1}Recall that the indices set of sequence $x$ is denoted by $x.\text{indices} \coloneqq [\len(x)]$.}
We say a sequence $x$ is monotonic, if for every $i, j \in \revise{[\len(x)]}$, $i < j \Rightarrow x_i \leq x_j$, or for every $i, j \in \revise{[\len(x)]}$, $i < j \Rightarrow x_i \geq x_j$, where $x_i$ denotes the $i$-th character in $x$.
\end{definition}

\begin{definition}[Matching]
\label{def:matching}
Consider the query sequence $q$ and the input sequence $s$ as two \emph{vertex sets} $(U =\{ u_1, \dots, u_\ell \}, V = \{ v_1, \dots, v_n \})$ where the vertex set $U$ denotes the characters in sequence $q$ and the vertex set $V$ denotes the characters in sequence $s$.
Let $M$ be an \emph{edge set} that for each $m \coloneqq (u, v) \in M$, we have $u \in U$ and $v \in V$.
We say \revise{\linelabel{M_s_q_0}$M(q, s)$} is a \emph{matching} \revise{(or simply $M$ when the context is clear)} between $q$ and $s$ (or $U$ and $V$) if $M$ satisfies the following properties:\\
\indent 1) every vertex in $U$ and $V$ corresponds to at least one edge in $M$; \\
\indent 2) the first character in $U$ is matched to the first character in $V$ and the last character in $U$ is matched to the last character in $V$;\\
\indent 3) the indices of matched character pairs are monotonic, i.e., for any two edges $(u_i, v_j), (u_k, v_l) \in M$, $i > k \Rightarrow j\geq l$ and $j > l \Rightarrow i \geq k$.\\
We define the degree of a vertex, $\revise{\deg}(u_i)$ or $\revise{\deg}(v_j)$, as the number of associated edges in a matching $M$.
\end{definition}

\begin{definition}[DTW Matching]
\label{def:dtw-matching}
The cost of an edge $m \coloneqq (u_i, v_j) \in M$ is defined to be the $\ell_1$ norm distance $\revise{\Cost}(m) \coloneqq \| u_i - v_j \|$.
The cost of a matching is defined as $\revise{\Cost}(M) \coloneqq \sum_{m \in M} \revise{\Cost}(m)$.
Let $\mc{M}$ consist of all possible matchings between $q$ and $s$ (or $U$ and $V$).
If a matching $M \in \mc{M}$ has \emph{minimal cost} on the edges, that is $\revise{\Cost}(M) = \min_{M_i \in \mc{M}} \revise{\Cost}(M_i)$, we call this matching a \emph{DTW matching}. 
A \emph{DTW matching} yields a \emph{DTW distance} between $q$ and $s$.
\end{definition}

\linelabel{easy_matching}Based on our definitions, the concepts of matching provide a different perspective of the non-decomposable distance.
A matching between two vertex sets defines a possible alignment between two sequences with different lengths.
The notion of DTW matching better captures the graph-theoretical properties of the implicit optimal alignment in computing DTW distance than the conventional definition.
The cost of a DTW matching is equal to the DTW distance between two sequences which are constituted by the vertex sets respectively.
We note that there might exist multiple DTW matchings (of equal cost) between a pair of sequences.

\begin{definition}[Isomorphic Matching]
Given input sequence $s$ of length $\ell$, two query sequences $q$ and $q'$ of length $n$ and two corresponding matchings $M$ (between $q$ and $s$) and $M'$ (between $q$ and $s'$). 
We say $M$ and $M'$ are \emph{isomorphic} if, 
$\forall 1 \leq i \leq \ell$ and $\forall 1\leq j \leq n$, edge $(s_i, q_j)\in M$ $\iff$ edge $(s_i, q_j') \in M'$.  
\end{definition}
}

\begin{definition}[\frechet Distance]
By replacing the $\ell_1$ norm in Definition \ref{def:dtw} with the  $\ell_\infty$ norm, we obtain the definition of the \frechet distance.

\end{definition}

The \frechet distance in our paper is equivalent to the discrete \frechet distance in the prior works of \citep{eiter1994computing,aronov2006frechet}.

\stitle{Extended alphabet.}
\revise{\linelabel{extended_alphabet}Since in this paper we discuss recovery sequence based on distance queries from \emph{binary or extended alphabet}, we would like to note that the distance definitions are independent of the alphabets. That being said, while we study the problem by restricting the input sequence as drawn from the binary alphabet (which generalizes to any constant-sized alphabet by applying coding methods), we do not change the distance definitions in a skewed way of embedding special symbols on the extended alphabet or backdoors to the oracle.
To ensure that the distance output makes sense, we specify the extended alphabets for queries to the different distance oracles.
For edit distance, the extended alphabet can contain \emph{any symbol} outside the binary alphabet, as the edit distance oracle counts the edit operations no matter what symbol is used.
For ($p$-)DTW distance and \frechet distance, the extended alphabet can consist of \emph{any real number}. This makes sense because the DTW and \frechet distances are defined based on $\ell_p$ or $\ell_\infty$ cost.
}

The main notations used in this paper are summarized in Table \ref{table:notations}.

\section{Our Techniques}
\label{sec:our_techniques}

\revise{
In this section, we summarize and highlight the main technical insights behind our results on non-adaptive recovery from the DTW distance oracle, which are the most non-trivial and interesting parts of this paper.
We hope to convey our intuitive ideas in a less formal manner before diving into the full proofs in the later sections.
Reader may skip this section if they are looking for the complete statements and proofs of these results.
In particular, we will cover the intuitions behind the following four theorems.
}

\begin{theorem}[Hardness, Refers to Theorem \ref{theorem:dtw_hardness}]
\label{thm:informal_dtw_hardness}
There exists a pair of input sequences $s$ and $s'$ such that for any query sequence $q$, $d_{\revise{\DTW}}(s, q) = d_{\revise{\DTW}}(s', q)$. That is, $s$ and $s'$ cannot be distinguished by DTW Distance Oracle queries without using extra characters. 
\end{theorem}

\revise{
Theorem \ref{thm:informal_dtw_hardness} shows the impossibility of only using binary sequences to recover the input sequence from the DTW distance oracle.}
If two input sequences cannot be distinguished, we say that they are in the same \emph{equivalence class}.
\revise{The following two informal theorems state the upper bound and lower bound on DTW distance recovery up to the equivalence class.}

\begin{theorem}[Informal, Upper Bound, Refers to Theorem \ref{theorem:oz_equivalence}]
\label{thm:informal_dtw_no_extra}
There exists a query set $\mc{Q}$ consisting of $\mc{O}(n)$ queries of length $\mc{O}(n)$, such that any two distinguishable input sequences can be distinguished by $\mc{Q}$.
\end{theorem}

\revise{\linelabel{query-construct-0}$\mc{Q}$ is designed to contain all queries with $i$ runs, for any $i \in [1, n]$.}

\begin{theorem}[Informal, Lower Bound, Refers to Theorem \ref{thm:dtw_no_extra_lowerbound}]
\label{thm:informal_dtw_no_extra_lowerbound}
    For the binary alphabet $\bits$, any algorithm to recover an arbitrary input sequence $s \in \bits^\ell$, where $0 \leq \ell \leq n$, up to its equivalence class, by querying the DTW distance to a set of sequences, has query complexity  $\Omega(n)$.
\end{theorem}

\revise{Note that our upper bound matches the lower bound for DTW equivalence class recovery. The next exciting finding is that, using queries that contain a small number of extra characters, we can exactly recover the input sequence.}

\begin{theorem}[Informal, Upper Bound with Extra Chars, Refers to Theorem \ref{theorem:dtw_o1}]
\label{thm:informal_dtw_extra}
By introducing $\mc{O}(1)$ extra characters to the query sequence alphabet, we can recover any input sequence of length $\leq n$ with $\mc{O}(n)$ DTW queries.
\end{theorem}

We aim to recover the given input sequence (of length $\leq n$) with the minimum number of queries for different distance metrics.
Theorems \ref{thm:informal_dtw_hardness}, \ref{thm:informal_dtw_no_extra} and \ref{thm:informal_dtw_no_extra_lowerbound} summarize the best results one can hope to obtain for recovering sequences from a DTW oracle without extra characters, i.e., identifying the equivalence class that the input sequence belongs to.
If we are allowed to use extra characters in the query construction, we can distinguish and recover all the sequences with $\mc{O}(n)$ queries, as informally stated in Theorem \ref{thm:informal_dtw_extra}.
We summarize and highlight the techniques used in proving these theorems in the rest of this section, in which the informal proofs are grouped as follows.
In Section \ref{subsec:tech_binary}, we show proof sketches on recovery of sequences using binary queries, which include results from Theorems \ref{thm:informal_dtw_hardness}, \ref{thm:informal_dtw_no_extra} and \ref{thm:informal_dtw_no_extra_lowerbound}.
In Section \ref{subsec:tech_extra}, we give a bird's-eye view over the key ideas of the query construction and proof of Theorem \ref{thm:informal_dtw_extra}.

\begin{figure}
    \centering
    \includegraphics[width=0.65\textwidth]{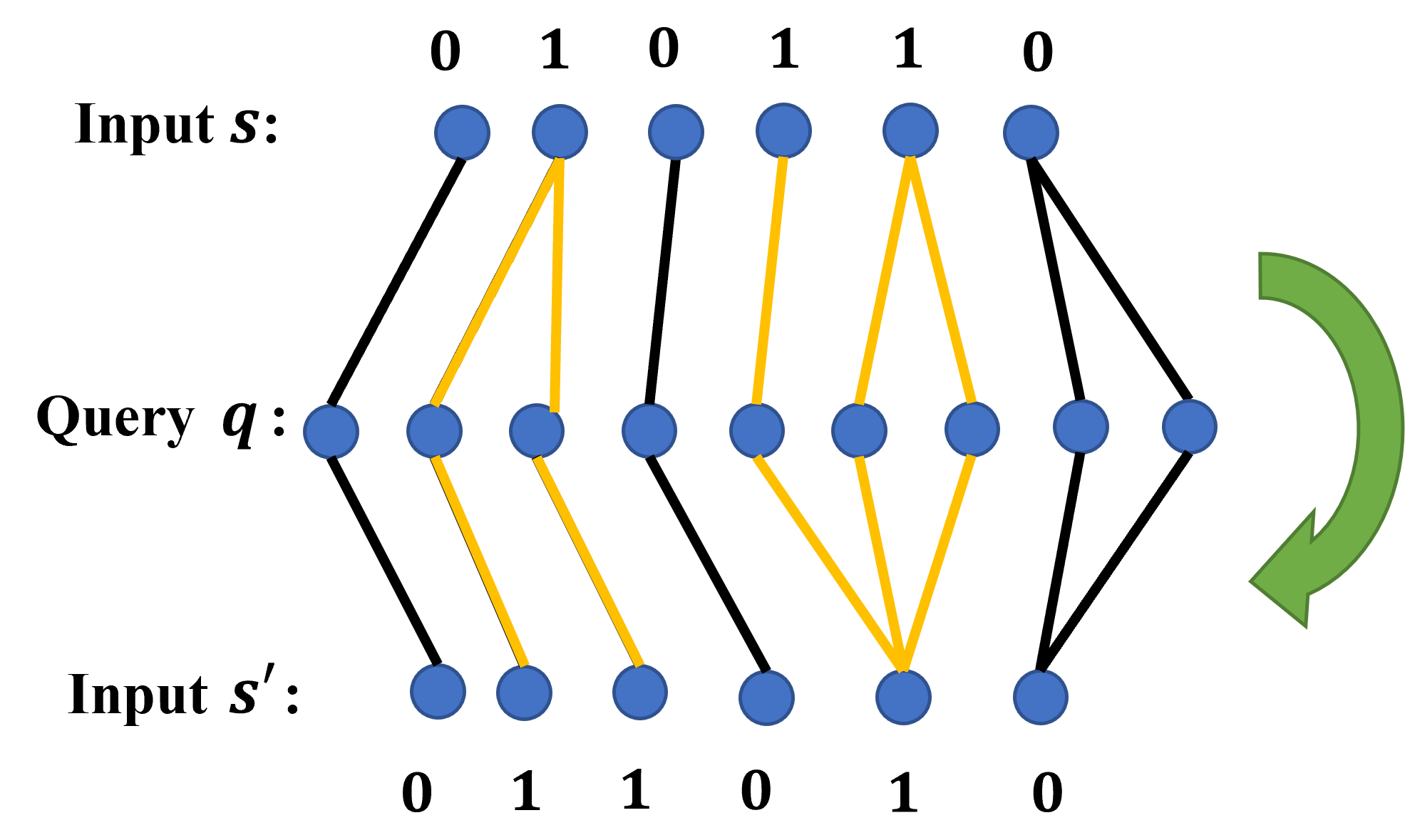}
    \caption{Constructing a matching between $q$ and $s'$ based on the matching between $q$ and $s$.}
    \label{fig:hardness}
\end{figure}

\subsection{Optimal Non-adaptive Strategy using DTW Queries over Binary Alphabet} 
\label{subsec:tech_binary}

\revise{
The hardness result (Theorem \ref{thm:informal_dtw_hardness}) is shown by finding evidence of such a pair of indistinguishable input sequences.
}
\smallskip

\noindent\textit{Informal proof for Theorem \ref{thm:informal_dtw_hardness}.} In the case of DTW Distance, we discover that it is actually impossible to recover any given input with an arbitrary number of queries.
For example, the input sequences $s = 010110$ and $s' = 011010$ cannot be exactly recovered, since they cannot be distinguished by any query sequence. To see this, the idea is that $d_{\revise{\DTW}}(1, r) = d_{\revise{\DTW}}(11, r)$ for any non-empty sequence $r$, unless $r = 0$. Therefore, a DTW matching between $s$ and any query sequence $q$ would yield a corresponding matching between $s'$ and $q$ with the same cost, (see Figure \ref{fig:hardness} as an example) and vice versa. %
(Refer to Theorem \ref{theorem:dtw_hardness} for detailed proof). This implies that $d_{\revise{\DTW}}(s, q) = d_{\revise{\DTW}}(s', q)$, and thus $s$ and $s'$ cannot be distinguished by $q$.

\bigskip

Before giving the intuition for the proof of Theorem \ref{thm:informal_dtw_no_extra} and \ref{thm:informal_dtw_no_extra_lowerbound}, we first introduce the notion of a Min 1-Seperated Sum (MSS) problem \citep{abboud2015tight,DBLP:conf/cpm/SchaarFN20}, where each instance of the DTW distance computation can be reduced to solving a corresponding instance of MSS problem.
\revise{The reduction plays the role of an important primitive in our proofs.}

\stitle{MSS Problem}. The min 1-separated sum (MSS) problem takes as input a sequence $seq$ of $m$ positive integers and an integer $r \geq 0$. The problem is to select $r$ integers from $seq$ and minimize their sum, under the constraint that any two adjacent integers cannot be selected simultaneously. We say $\MSS(seq, r)$ is an MSS instance.

\begin{theorem}[DTW-to-MSS Reduction, \citep{DBLP:conf/cpm/SchaarFN20}, Theorem 2]
\label{theorem:our_techs_dtw_mss}
Let $x \in\{0,1\}^{m}$ and $y \in\{0,1\}^{n}$ be two binary strings such that $x[1]=y[1], x[m]=y[n]$, and $\#\textsc{runs}(x) \geq\#\textsc{runs}(y)$. Then, the DTW distance between $x$ and $y$, i.e., $d_{\revise{\DTW}}(x, y)$, equals the sum of a solution for the MSS instance $\MSS\biggl(\Bigl(\textsc{lor}(x,2), \ldots,\textsc{lor}(x, \#\textsc{runs}(x)-1)\Bigl), \frac{(\#\textsc{runs}(x)-\#\textsc{runs}(y))}{2} \biggl)$. 
\end{theorem}
\noindent
To give an example of the reduction, let $s = 010110$ and $q = 010$.
By \revise{\linelabel{mis-ref}Theorem \ref{theorem:our_techs_dtw_mss}}, we obtain $d_{\revise{\DTW}}(s, q) = \revise{\MSS}((1, 1, 2), 1)$.
For ease of presentation, we will use $\revise{\MSS}(x ,(\#\textsc{runs}(x)-\#\textsc{runs}(y)) / 2)$ to represent the same MSS instance.

\begin{figure*}[t]
    \centering
    \includegraphics[width=0.8\textwidth]{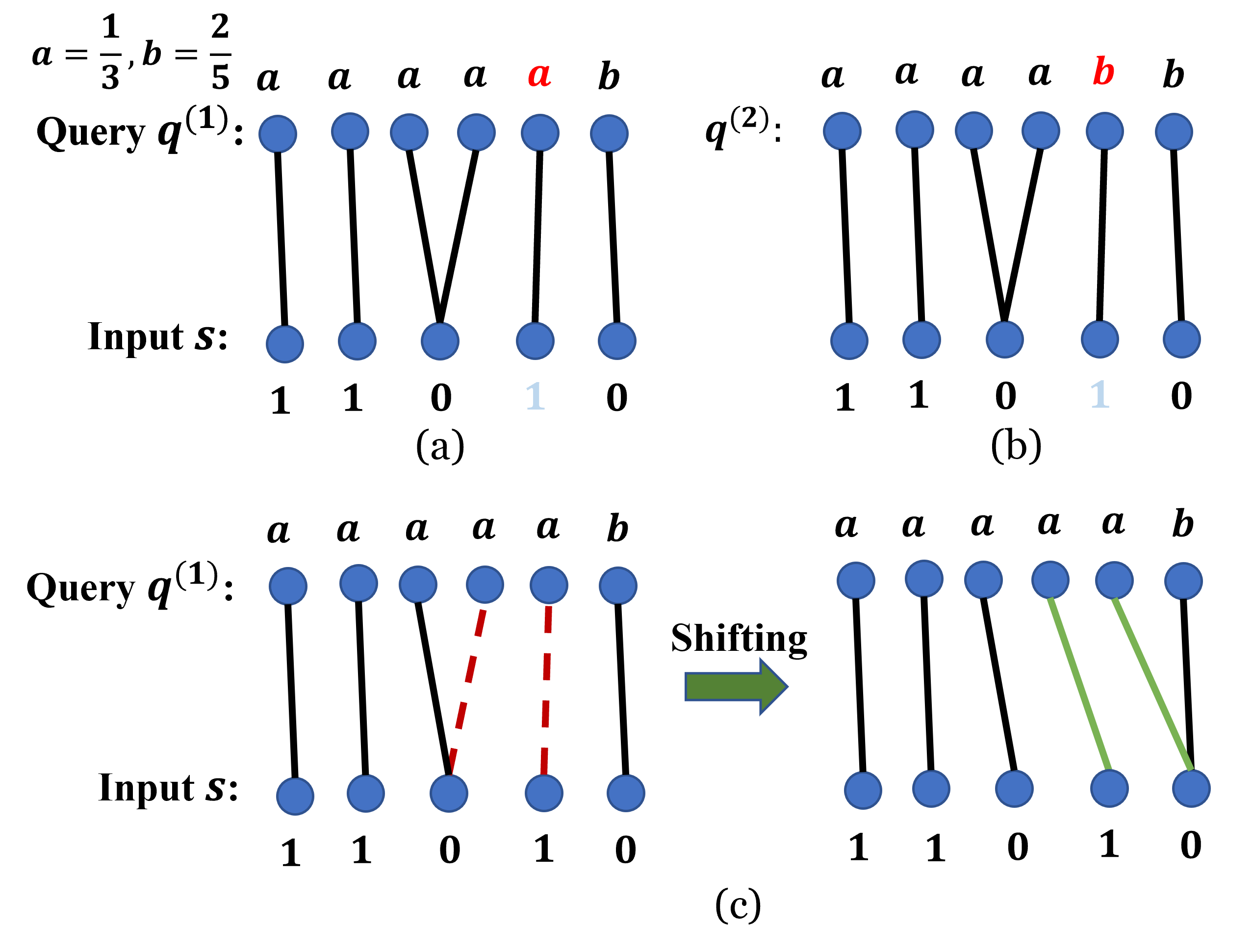}
    \caption{(a) Illustration of input-uniqueness and 0/1-uniqueness; (b) Illustration of isomorphism and \revise{performing a difference operation}, compared to Fig (a); (c) Illustration of shifting operation.}
    \label{fig:our_techniques}
\end{figure*}

\stitle{Remark.} For binary strings $x \in \{0, 1\}^m, y \in \{0, 1\}^n$ where $x[1]\neq y[1]$ or $x[m] \neq y[n]$, we can still reduce $d_{\revise{\DTW}}(x, y)$ to an MSS instance \revise{(which will be presented later in the paper \needreview{using} another technique from \citep{DBLP:conf/cpm/SchaarFN20})}. In this section, where we only illustrate the main idea of the proofs, we will only consider the case where the input sequence and query sequence each have the same starting character and the same ending character (so Theorem \ref{theorem:our_techs_dtw_mss} can be directly applied), and other cases can be resolved similarly. For full details, we defer to later sections. 

\medskip

\noindent\textit{Intuition for Theorem \ref{thm:informal_dtw_no_extra}.}
\revise{We would like to skip the proof sketch for Theorem \ref{thm:informal_dtw_no_extra}, but just to mention the insights of the query construction to obtain such an \needreview{orderwise} optimal query complexity upper bound.
The set of queries $\mc{Q}$ contains queries of all possible combinations of runs in the input sequence. That is, for the maximum length $n$ of the input sequence, the set of the possible number of runs is $[n]$. 
This gives us $n$ queries. Since we have 0 runs and 1 runs, there are $2n$ queries in the query set $\mc{Q}$ in total.
Then the remainder of the proof is to perform case analysis --  we first eliminate obvious cases and then build a mapping to the corresponding MSS instances such that if any pair of sequences cannot be distinguished by $\mc{Q}$, they cannot be distinguished by any binary queries.
}

\bigskip

\noindent\textit{Informal proof of Theorem \ref{thm:informal_dtw_no_extra_lowerbound}.}
\revise{\linelabel{query-construct-1}Recall our query set $\mc{Q}$ contains queries of all numbers of runs.}
The intuition for the proof of Theorem \ref{thm:informal_dtw_no_extra_lowerbound} is that, for each given constant-length interval of the number of runs, we can construct a certain pair of input sequences which can only be distinguished by queries with a number of runs within this interval. For instance, it can be proved that $s_1 = 01^301^30^31^30^31^30$ and $s_2 = 01^30^21^30^21^30^31^30$ can only be distinguished with queries with a number of runs within $[4, 10]$. Thus, an  $\Omega(n)$ number of such constructed pairs of input sequences can correspond to $\Omega(n)$ disjoint intervals, yielding an $\Omega(n)$ lower bound for this problem. 

We now construct a class of pairs of input sequences $(s, s')$ \revise{where $s$ and $s'$ share the same starting and ending character,} such that $s$ and $s'$ can only be distinguished by queries $q$ with a number of runs within $[\#\textsc{runs}$(s)$+c_1, \#\textsc{runs}$(s)$+c_2]$ for two constants $c_1<c_2$. According to Theorem \ref{theorem:our_techs_dtw_mss}, as long as the constructed pair of input sequences $(s, s')$ have the same number of runs, for a query $q$ with more than \#\textsc{runs}$(s)$ number of runs, $d_{\revise{\DTW}}(q, s)$ and $d_{\revise{\DTW}}(q, s')$ are only determined by the query $q$ and \#\textsc{runs}$(s)$, and thus $q$ cannot distinguish $s$ and $s'$. For a query $q$ with fewer than \#\textsc{runs}$(s)$ number of runs, $d_{\revise{\DTW}}(q, s)$ and $d_{\revise{\DTW}}(q, s')$ are reduced to two MSS instances. Note that for different queries $q$, the sequences (i.e., the first parameter) of MSS instances remain the same, while \#\textsc{runs}($q$) determines the number of elements selected in the sequences of MSS instances (i.e., $(\#\textsc{runs}(s) - \#\textsc{runs}(q))/2$). We hope to construct a pair of sequences $seq$ and $seq'$ such that $\revise{\MSS}(seq, 1)\neq \revise{\MSS}(seq', 1)$ and $\revise{\MSS}(seq, x) = \revise{\MSS}(seq', x)$ for all $x>1$: let $seq$ and $seq'$ be the sequences corresponding to MSS instances of $s$ and $s'$; in this way, $s$ and $s'$ would still be distinguishable because $\revise{\MSS}(seq, x)\neq \revise{\MSS}(seq', x)$ for $x=1$, but any query $q$ with fewer than \#\textsc{runs}$(s)-4$ runs cannot distinguish $s$ and $s'$ because $\revise{\MSS}(seq, x) = \revise{\MSS}(seq', x)$ for all $x\ge 2$, where $x=(\#\textsc{runs}(s) - \#\textsc{runs}(q))/2$.

\subsection{Optimal Non-adaptive Strategy using DTW Queries with Extra Characters}
\label{subsec:tech_extra}

We show that, if we augment the ability of our oracles by introducing extra characters, we can solve the DTW distance oracle recovery problem with optimal query complexity up to polylogarithmic factors.

\smallskip
\noindent\textit{Informal proof of Theorem \ref{thm:informal_dtw_extra}.}
We would like to construct a query set of size $\mc{O}(n)$ that can recover the input sequence using a DTW distance oracle.
A natural idea is to retrieve information about the input sequence by \revise{\linelabel{differentiating-1}taking the difference between} the query results of neighbouring queries (i.e., queries only differing by $1$ character).
To achieve this, we construct a query set satisfying the following three properties:

1) \emph{Isomorphism}:
The matchings corresponding to neighboring queries should be isomorphic. \needreview{Fig.~\ref{fig:our_techniques} (a) and Fig.~\ref{fig:our_techniques} (b) show} an example of isomorphism, where only one character of the input sequence is changed, while the structure of both optimal matchings remains identical.
With this property, we know that the difference between the query results of neighboring queries only reflects the effects of the different characters in neighboring queries.
This property is the essence of guaranteeing the correctness of \revise{\linelabel{differentiating-2}the difference operation}.

2) \emph{Input-uniqueness}:
Each character in the query sequence should be matched to exactly $1$ character in the input sequence.
\revise{Another way to think of this property is to imagine a \emph{total function} that maps the entire query sequence to the input sequence.
Each matching between the query and input defines such a function so that we can extract information about the input by knowing something about the function.}
With this property, we can \revise{\linelabel{differentiating-3}take the difference} to get the information of a single character in the input sequence with a pair of neighboring queries. Note that if the differing character in the neighboring queries is matched to multiple characters in the input sequence, the difference in the query results can only reflect the sum of the costs over these characters, which makes exact recovery hard. Take \needreview{Fig.~\ref{fig:our_techniques} (a) and Fig.~\ref{fig:our_techniques} (b)} as an example. 
Input-uniqueness is satisfied for both \needreview{Fig.~\ref{fig:our_techniques} (a) and Fig.~\ref{fig:our_techniques} (b)}, since all characters in the query sequences of both figures have degree $1$. 
\needreview{Denote the matchings from Fig.~\ref{fig:our_techniques} (a) and Fig.~\ref{fig:our_techniques} (b) by $M_a$ and $M_b$ respectively.}
Since \needreview{$M_a$} has cost $3(1-a) + 2a + b$ while \needreview{$M_b$} has cost $2(1-a) + 2a + (1-b) + b$, we know that $\revise{\Cost}(\needreview{M_a})-\revise{\Cost}(\needreview{M_b}) = b-a$. 
By \revise{\linelabel{differentiating-4}taking the difference}, we can infer that $s[4] = 1$; otherwise, if $s[4]= 0$, we would have $\revise{\Cost}(\needreview{M_a})-\revise{\Cost}(\needreview{M_b}) = (a-0) - (b-0) = a-b$.

Combining properties 1) and 2), we note that each character in the input sequence can match to $1$ or more characters in the query sequence, so we can obtain an expansion of the input sequence.
Based on the example, \needreview{Fig.~\ref{fig:our_techniques} (a) and Fig.~\ref{fig:our_techniques} (b)}, we can obtain an expansion, $110010$, of the input sequence.
We can then infer that the input sequence is of the form $1^x0^y10$, where $x, y \in [1, 2]$.
To recover the exact input sequence, we require more information given by the following third property.

3) \emph{0/1-uniqueness}:
In an optimal matching w.r.t. our constructed queries, either all 0's or all 1's in the input sequence have degree 1.
Using this property, we can locate the exact position of either all 0's or all 1's in the input sequence, and exactly recover the input sequence by combining the two cases. 
In the example of Fig \ref{fig:our_techniques}, 1-uniqueness is satisfied in \needreview{Fig.~\ref{fig:our_techniques} (a) and Fig.~\ref{fig:our_techniques} (b)}, while 0-uniqueness is not, since $s[3]$ in both figures has degree 2. 
According to 1-uniqueness, we can reduce the form of the input sequence from $1^x0^y10$ to $110^y10$.
Similarly, we can construct another set of queries that satisfies 0-uniqueness to locate the positions of 0's in the input sequence, which determines $y$ in this example.

\stitle{Sequence Monotonicity $\rightarrow$ Input-uniqueness.}
We observe that property 2) can be obtained from a \emph{monotonic} design of the query sequences.

\begin{lemma}[Refers to Lemma \ref{lemma:single-direction_o1}]
\label{lemma:single-direction_o1_intro}
Given a monotonic sequence $q$ of length $n$ where 
\begin{equation}
    \min_{i\in [n]}\max\{|q[i]-0|, |q[i]-1|\} > \max_{i,j\in [n]}|q[i]-q[j]|,\label{equ:single_direction_o1_intro}
\end{equation}
for any input sequence $s$ with length $\ell \leq n$, given a DTW matching $M$ for $(q,s)$, we have $\revise{\deg}(q[i]) = 1$ for all characters $q[i]$ in $q$. 
\end{lemma}

The intuition for Lemma \ref{lemma:single-direction_o1_intro} is that, with the monotonic property and equation (\ref{equ:single_direction_o1_intro}) guaranteed in our query construction, we can ensure that there do not exist characters $s[i]\in s$ and $q[j] \in q$ where $\revise{\deg}(s[i])>1$ and $\revise{\deg}(q[j])>1$ are satisfied at the same time. Fig \ref{fig:rematching_o1} in a later section illustrates that, for such a pair of $s[i]$ and $q[j]$, we can always construct a matching with lower cost where one of their degrees is decreased to $1$. Therefore, either all characters in $s$ or all characters in $q$ would have degree $1$. Since $\revise{\len}(q) = n \geq \revise{\len}(s)$, we know that $\revise{\deg}(q[i]) = 1$ for all characters $q[i]$ in $q$.

\needreview{Fig.~\ref{fig:our_techniques} (a) and Fig.~\ref{fig:our_techniques} (b)} satisfy sequence monotonicity, since the query sequences in both figures are monotonic sequences of length $n$ and \revise{for $x$ in $\{1, 2\}$, $\min_{i\in [n]}\max\{|q^{(x)}[i]-0|, |q^{(x)}[i]-1|\} = \frac{3}{5}> (\frac{2}{5} - \frac{1}{3}) = \max_{i,j\in [n]}|q^{(x)}[i]-q^{(x)}[j]|$.}

\stitle{Sequence 0/1-preference $\rightarrow$ 0/1-uniqueness.}
We observe that property 3) can be guaranteed by the \emph{0/1-preferred} design of the query sequences.
If all characters in the query sequence are less than (or greater than) $\frac{1}{2}$, then we can guarantee 1-uniqueness (or 0-uniqueness) of the input sequence.
Intuitively, this would hold because, if all characters in the query sequence are less than (or greater than) $\frac{1}{2}$, matching them to 0's (or 1's) in the input sequence yields lower cost than matching to 1's (or 0's). 
\needreview{Fig.~\ref{fig:our_techniques} (a) and Fig.~\ref{fig:our_techniques} (b)} satisfy 0-preference, since all characters in query sequences (either $a = \frac{1}{3}$ or $b= \frac{2}{5}$) are less than $\frac{1}{2}$.

\stitle{Query Construction.} We now propose the following design of the query sequence.
We first need a single 0 query and a single 1 query to obtain the number of 1's and 0's in the input sequence.
Let $a, b$ be two fractional characters that satisfy $0 < b-a < a < b < \frac{1}{2}$ and the denominators of $a, b$ are co-prime.
Without loss of generality, we can assume $a = \frac{1}{3}$ and $b = \frac{2}{5}$.
We will use $a, b$ as the extra characters to construct the query sequences.
In particular, the rest of the query sequences (other than the $0$ query and the $1$ query) consist of queries $\mc{Q}$ in the form of $q^{(i)} = a^{n-i}b^{i}$, where $i = 1, \dots, n$.
This query construction satisfies \emph{sequence monotonicity} and \emph{sequence 0/1-preference} properties. Now we need to prove it also satisfies \emph{isomorphism}.

\begin{lemma}[Refers to Lemma \ref{lemma:4_o1}]
\label{lemma:isomorphism_intro}
For any input sequence $s$, there exists an \revise{ set of isomorphic matchings} $\mc{M}$ where $M_i\in \mc{M}$ is optimal for query $q^{(i)} \in \mc{Q}$.
\end{lemma}

Lemma \ref{lemma:isomorphism_intro} guarantees the isomorphism property of the constructed query set $\mc{Q}$. 
Here we construct an isomorphic set of matchings $M_i \in \mc{M}$ such that only the first $0$ in the input sequence has degree greater than $1$, while all other characters in the matching are of degree $1$. 
\needreview{Fig.~\ref{fig:our_techniques} (a) and Fig.~\ref{fig:our_techniques} (b)} are instances of $M_1$ and $M_2$, where the matchings in both figures are isomorphic to each other. 
\ifarxiv
\begin{wrapfigure}{r}{5.5cm}
\includegraphics[width=5.5cm]{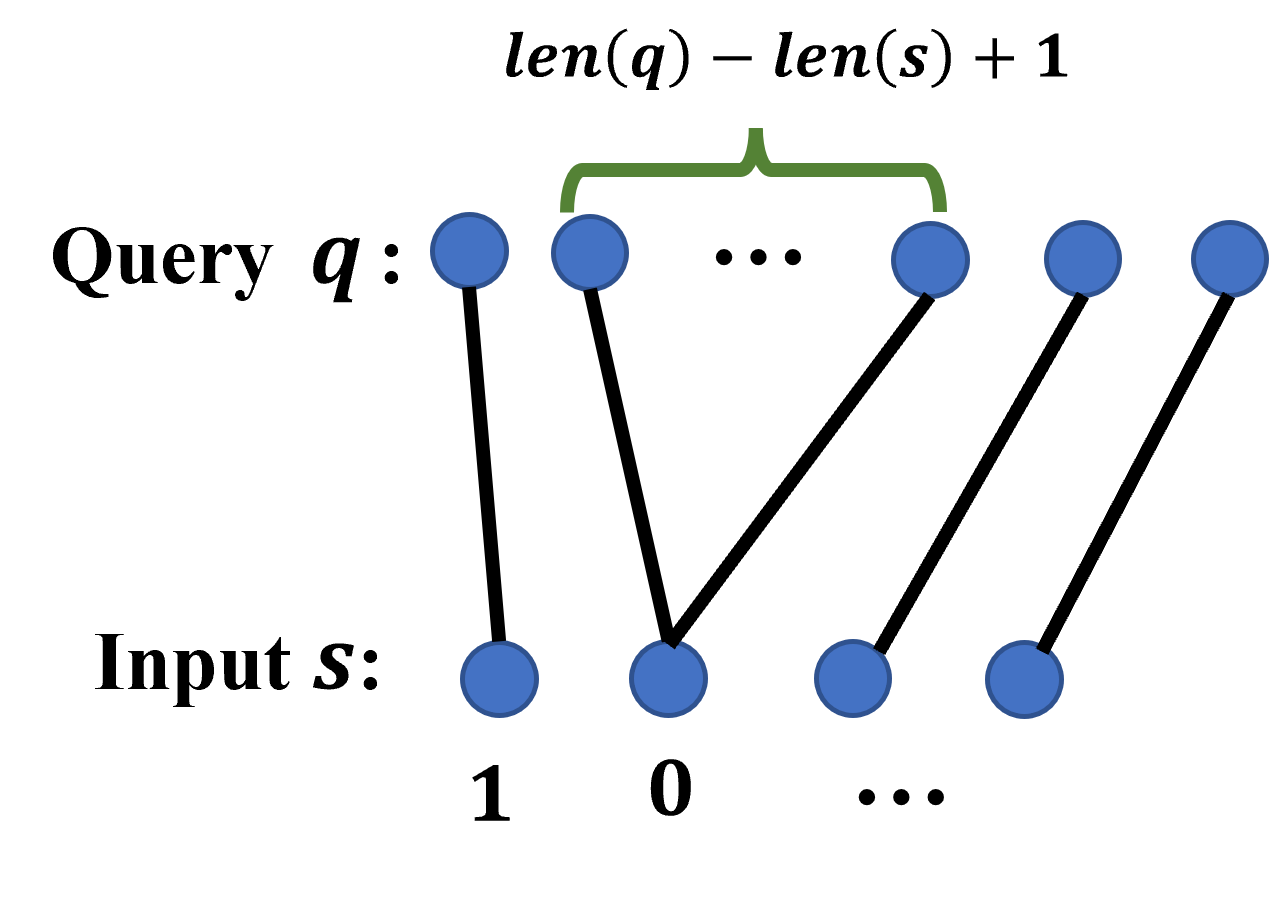}
\vspace{-1cm}
\caption{The position of the first $0$ in $s$ and the length of both sequences can determine the structure of $M_i$.}\label{fig:first_0}
\end{wrapfigure}
\fi
Note that in this construction, the structure of the matchings is only determined by the position of the first $0$ in the input sequence and the length of both sequences (see Fig \ref{fig:first_0}).
Since all query sequences in $\mc{Q}$ have the same length, an isomorphism of constructed matchings is naturally guaranteed.

To prove the optimality of the $M_i$, 
we introduce the notion of a  \emph{``shifting'' operation}.
Consider two 0's in the input sequence.
If any character between them has degree 1 and the first 0 has degree greater than 1, by running the shifting operation we decrease the degree of the first 0 by 1 and increase the degree of the last 0 by 1, while preserving the degree of all other characters.
\needreview{Fig.~\ref{fig:our_techniques} (c)} illustrates an example of the shifting operation.

\begin{claim}[Refers to Claim \ref{claim:3_o1}]
For our constructed query set $\mc{Q}$, a shifting operation would not reduce the total cost of the matching.
\end{claim}

\begin{claim}[Informal, Refers to Claim \ref{claim:4_o1}]
Given input sequence $s$, query $q^{(i)}\in \mc{Q}$ and any optimal matching $M_i^*$ between $s$ and $q^{(i)}$, we can obtain $M_i^*$ by applying a series of shifting operations to $M_i$. 
\end{claim}

Combining the above two claims, we can show that the $M_i$ are always optimal, which proves Lemma \ref{lemma:isomorphism_intro}. 
So far, the constructed query set satisfies three properties -- isomorphism, input-uniqueness, and 0/1-uniqueness.
Further details of our algorithm to recover the input sequence are given in later sections (see Algorithm \ref{alg:dtw_on_recover}).

\section{Recovery with Adaptive Queries}
\label{sec:adaptive}

\subsection{General Framework}
\label{sec:coordinate_descent}

\begin{theorem}[Coordinate Descent Framework]
\label{theorem:adaptive_general}
\revise{\linelabel{coordinate-descent}For a given distance oracle $\revise{\dist}(\cdot, \cdot)$, a constant-sized alphabet $\Sigma$ and any input sequence $s \in \Sigma^i$ where $0 \leq i \leq n$,} there exists an adaptive algorithm which returns a sequence $s'$ such that its distance to the input sequence $s$ satisfies $\revise{\dist}(s, s')=0$ using $\poly(n)$ queries, given that the following two conditions are true:
\begin{itemize}
    \item \revise{There exists a positive constant $c$ (independent of $n$)}, $\forall s \in \Sigma^i, q \in  \Sigma^{\mc{O}(n)}, \text{where}~ \revise{ 0 \leq i \leq n ~\text{and}~ \dist(s, q)>0}$, we can find a sequence $q'$ within $\poly(n)$ queries such that $\revise{\dist}(s, q) \geq \revise{\dist}(s, q') + c$;
    \item $\forall s \in \Sigma^i, q \in  \Sigma^{\mc{O}(n)}$, $\revise{\dist}(s, q) \leq \poly(n)$.
\end{itemize}
\end{theorem}

\noindent
\textit{Proof sketch}: The two above conditions naturally imply a local search algorithm. To recover the sequence $q$, we perform the following steps: 1) randomly initialize $q$. 2) find $q'$ such that $\revise{\dist}(s, q) > \revise{\dist}(s, q')$. 3) set $q$ to $q'$ and repeat 2) to 3).
The algorithm terminates if $\revise{\dist}(s, q) = 0$, and outputs the final $q$ as the sequence $s'$.

Since we reduce $\revise{\dist}(s, q)$ by at least \revise{a positive constant $c$} in each iteration, and \revise{\linelabel{poly-2}$\revise{\dist}(s, q) \leq \poly(n)$},  the algorithm terminates in at most \revise{\linelabel{poly-3}$\poly(n) / c$} iterations.
Therefore, the total number of queries is \revise{\linelabel{poly-4}$\mc{O}(\poly(n))$}.
\qed

\bigskip

The above local search algorithm can be applied to all aforementioned distances.
Specifically, the complexity for the edit distance, DTW distance and \frechet distance is $\mc{O}(n^2)$, $\mc{O}(n^2)$ and $\mc{O}(n)$, respectively.
A detailed instantiation of the algorithm on these distances can be found in Appendix \ref{sec: Local Search Algorithm Instantiation}.

\revise{
\noindent \textbf{Remark.}
As stated in the theorem, the objective of this coordinate descent framework is to reduce $\dist(s, s')$ to 0, which reflects our ``zero distance to input'' recovery guarantee.
We remark that, for distance function which is a \emph{metric}, this guarantee implies ``recover to equivalence class'', while for distances such as DTW where the triangle inequality does not apply, there exist sequences that can be distinguished whereas the distance is 0.
Such an example includes sequence 101 and 1011.
}

\subsection{Edit Distance}
We show that a binary input sequence with maximum length $n$ can be adaptively recovered using at most $\revise{n + \log n + c  \in \mc{O}(n)}$ queries to the edit distance oracle (where $c$ is a constant), by the following theorem.

\begin{theorem}[Adaptive Strategy for Edit Distance]
\label{theorem:adaptive_edit}
For a binary alphabet $\bits$, and any input sequence $s \in \bits^{\ell}$ with $k$ runs where $0 \leq \ell \leq n$, there exists an adaptive algorithm to recover the input sequence $s$ using at most \revise{$2k \log (n/k) + \log n + k + c$} queries $\revise{\mc{Q}}$ of length $\leq n$ and the exact Levenshtein distance of $s$ to each query sequence $q_i \in \revise{\mc{Q}}$, where the query sequences use no extra characters.
\end{theorem}

\begin{proof} 
The proof makes use of the following claim.
\begin{claim}
\label{claim:edit_subseq}
Given two sequences $x$ and $y$, the edit distance $d_{L}(x, y) = |\revise{\len(x)} - \len(y)| $ if and only if $x$ is a subsequence of $y$ or $y$ is a subsequence of $x$.
\end{claim}

\noindent
\textit{Proof of claim.}
\revise{Without loss of generality}, we can assume that $\revise{\len(x)} \geq \len(y)$. Since each insertion, deletion or substitution operation can change the sequence length by at most $1$, we have $d_{L}(x, y) \geq t$ where $(t=\len(x) - \len(y))$. 
If $y$ is a subsequence of $x$, we can obtain $y$ by performing $t$ deletions on $x$. Since $d_{L}(x, y) \geq t$, we have $d_{L}(x, y) = t$.
If $y$ is not a subsequence of $x$, we show that $d_{L}(x, y) > t$. To transform $x$ to $y$ we would need at least $t$ deletions. Since $y$ is not a subsequence of $x$, we cannot obtain $y$ by merely performing $t$ deletions on $x$, implying that $d_{L}(x, y) > t$.
\bigskip

\noindent
Next, to prove Theorem \ref{theorem:adaptive_edit}, we observe that for any sequence on a binary alphabet, the first run starts with either 0 or 1.
That is, the \revise{\linelabel{dense-1}condensed expression} of a binary sequence is in the form of $1010...$ or $0101...$.
Let the number of runs be $k$.
The first part of our adaptive recovery algorithm is determining the input sequence's \revise{\linelabel{dense-2}condensed expression}.
To do so, we need the following set of $2n+1$ queries, $\{ \phi, 0, 1, 01, 10, 010, \dots \}$, where the maximum length of the query in this set is $n$.
The length of the input sequence $\ell$ can be determined by querying the empty sequence $\phi$.
The \revise{\linelabel{dense-3}condensed expression} of the input sequence is equal to the query in the query set of maximum length such that $k = \ell - r$, where $k$ is the length of this query and $r$ is the query result from the oracle. 
Since we are adaptively querying the oracle, we do not require all $2n+1$ queries.
By using our querying strategy, the query complexity of this part can be reduced to $\log n + c$.
To see this, we take out all $n$ queries beginning with $0$ from the query set and adaptively query the oracle using binary search\revise{\linelabel{revise_theorem:adaptive_edit_log_n} to find the longest query sequence such that the edit distance between this query and the input sequence equals the length difference between two sequences. 
Next, we add an $1$ to the left side (or the most significant bit) of the longest query we just selected, then query the oracle to see if the distance is smaller.
The query sequence with the smaller edit distance is therefore the \revise{\linelabel{dense-4}condensed expression} of the input sequence.
Since a $\phi$ query is required at the beginning,
the entire process requires at most $\log n + 2 \leq \log n + c$ queries.}

The second part of our algorithm is to recover the sequence from the \revise{\linelabel{dense-5}condensed expression} \revise{via expanding each run} by inserting $1$'s (or $0$'s) into the corresponding location.
We have obtained the number of runs, which is $k$.
\revise{According to the claim, if any one of the runs of the query sequence contains more characters than that of the input sequence, meaning the query sequence is no longer the subsequence of the input sequence, then we can observe from the query results.
Therefore, we can recover the input sequence run by run.
The na\"ive way of achieving this is to iterate over runs and insert one character per time to a run until scanning and fulfilling the entire sequence, which results in at most $n$ queries.
Combining the first part of recovering the condensed expression, this approach gives us the overall query complexity of $n + \log n + c  \in \mc{O}(n)$.
}

\revise{An alternative approach to recover the runs is to determine the number of characters in each run using line search and binary search.
That is, we increase the number of characters in a run exponentially (by a factor of 2) and then look back to find the exact number by binary search.
Compared to directly using binary search to find the length of the run within the range of $[2, n-k]$, the complexity analysis of our approach can avoid a potential $n\log n$ term.
To give an example of this approach, suppose we have a run of length 13.
To recover this run, instead of using 13 queries by the na\"ive approach, we can make 7 queries with the following numbers of 1's: 2, 4, 8, 16, 12, 14, and 13, respectively. 
For a run with length $m$, the worst-case query complexity of this approach is $\lceil2\log m\rceil$.
}
Let $t_i$ be the number of characters in each $1$ run or $0$ run.
Then we have $\sum_{i=1}^k t_i = \ell \leq n$.
Since the number of characters in each block can be determined adaptively using line and binary search, we can derive the query complexity of the second part as $\sum_{i=1}^k \lceil 2\log t_i \rceil \leq 2\log (\Pi_{i=1}^k t_i) + k \leq 2\log (\sum_{i=1}^k t_i / k)^k + k \leq 2k \log (n / k) + k$.
The first inequality holds due to the AM-GM inequality.
Combining the two parts of the algorithm, we know that the overall adaptive query complexity for the exact recovery of the sequence is \revise{$2k\log (n/k) + k + \log n + c \in \mc{O}(n)$.}
\end{proof}

In many cases our alternative approach, as shown in the complexity, saves queries.
There are also edge cases that the na\"ive approach wins the game -- for runs with two characters, using binary search requires 3 queries (i.e., queries with 2, 4, 3, characters in this run respectively), while the na\"ive approach finishes the task with only 2 queries.

\begin{theorem}[Yet Another Adaptive Strategy for Edit Distance]
\label{theorem:adaptive_edit_N+1}
For a binary alphabet $\bits$, and any input sequence $s \in \bits^{\ell}$ where $0 \leq \ell \leq n$, there exists an adaptive algorithm to recover the input sequence $s$ using at most $\revise{n + 2  \in \mc{O}(n)}$ queries $\revise{\mc{Q}}$ of length $\leq n$ and the exact Levenshtein distance of $s$ to each query sequence $q_i \in \revise{\mc{Q}}$, where the query sequences use no extra characters.
\end{theorem}

\begin{proof}
The adaptive query strategy is the following.
We first use an empty sequence to query the length $\ell \in [n] = \{1, 2, \ldots, n\}$ of the input sequence.
Then we use $\ell + 1 \leq n + 1$ queries: an $e_0 = 0^\ell$ query and a set of $e_i = 0^{i-1}10^{\ell-i}, i \in [\ell]$ queries (all with length $\ell$).

\begin{claim}
\label{claim:edit_adaptive_n_1}
    $s[i] = \begin{cases}
        0, & \text{if}~~ d_{L}(s, e_0) - d_{L}(s, e_i) \leq 0;\\
        1, & \text{if}~~ d_{L}(s, e_0) - d_{L}(s, e_i) = 1.
    \end{cases}$
\end{claim}

\noindent
\textit{Proof of claim.}
If $s[i] = 1$, $d_{L}(s, e_0) - d_{L}(s, e_i) = $ (\#1's in $s$) $-$ (\#1's in $s - 1$) $ = 1$.
If $s[i] = 0$, $d_{L}(s, e_0) = $ (\#1's in $s$).
We show that $d_{L}(s, e_i) \geq $ (\#1's in $s$).
First, $d_{L}(s, e_i) \geq $ (\#1's in $s$) $-$ (\#1's in $e_i$) = \#1's in $s$ $-$ 1.
Consider the series of transformations from $s$ to $e_i$:
1) If we only perform substitution on $s$, we need at least \#1's in $s$ + 1 operations.
2) Otherwise we show that we have at least one insertion. If we perform at least one deletion operation(s) on $s$, since $s$ and $e_i$
 are of the same length, we would need at least one insertion(s) on $s$.
 Note that insertions on $s$ cannot reduce the difference of the number of 1's between $s$ and $e_i$.
 Thus, we need at least (\#1's in $s$ $-$1) extra operations to reduce the difference to $0$ and we have $d_{L}(s, e_i) \geq $ (\#1's in $s$ $-$1) $+$1 $=$ \#1's in $s$.
Combining these cases, we obtain $d_{L}(s, e_i) \geq $ (\#1's in $s$).

\bigskip
\noindent
By claim \ref{claim:edit_adaptive_n_1}, we can recover the sequence $s$ character by character.
\end{proof}

\stitle{Remark.}
We remark that both results of Theorem \ref{theorem:adaptive_edit} and Theorem \ref{theorem:adaptive_edit_N+1} are useful.
\revise{Clearly, using $n+2$ queries in the second algorithm is a better strategy than the na\"ive approach in the first algorithm, which requires $n + \log n + c$ queries.
However, the first algorithm with the binary search approach yields query complexity of $2k \log (n /k) + \log n + 2k + c$.
When the number of runs (i.e., $k$) is small,} this result is better than the $n+2$ queries in the second algorithm.

\subsection{DTW Distance}

\begin{theorem}[Adaptive Strategy for DTW Distance]
\label{thm:dtw_adaptive}
For a binary alphabet $\bits$, and any input sequence $s \in \bits^{\ell}$ where $0 \leq \ell \leq n$, there exists an adaptive algorithm to recover the input sequence $s$ using at most $\revise{n + 1  \in \mc{O}(n)}$ queries $\revise{\mc{Q}}$ of length $\leq n$ and the exact DTW distance of $s$ to each query sequence $q^{(i)} \in \revise{\mc{Q}}$, where the query sequences use \revise{\linelabel{revise:theorem_3.4}1} extra character.
\end{theorem}

\begin{proof}
Using an adaptive method, for a binary alphabet $\bits$, an input sequence $s \in \bits^i$ where $0 \leq i \leq n$ can be exactly recovered with at most $\revise{n + 1  \in \mc{O}(n)}$ queries to the DTW distance oracle.
We need 1 additional character, which is the fractional character $\frac{1}{2}$, to construct the set of query sequences.
The details are presented as follows. 

First, with a single-character query sequence $q^{(1)} = \frac{1}{2}$, we can obtain the length of the input sequence $s$, which is $\ell = 2d_{\revise{\DTW}}(s, q^{(1)})$.

Consider the query $q^{(2)}=0~ (\frac{1}{2})^{\ell-1}$. Note that each $\frac{1}{2}$ in the $q^{(2)}$ corresponds to at least $\frac{1}{2}$ cost in the query result, and we have $d_{\revise{\DTW}}(s, q^{(2)})\geq (\ell-1)/2$. If $s[1] = 0$, then $s[1]$ and $q^{(2)}[1]$ are perfectly matched, so $d_{\revise{\DTW}}(s, q^{(2)}) = (\ell-1)/2$. Otherwise, the first character $0$ in $q^{(2)}$ would correspond to cost $>0$ in the query result, so $d_{\revise{\DTW}}(s, q^{(2)}) > (\ell-1)/2$.
In this way, we can recover $s[1]$.

Now we recover the whole sequence by induction. Suppose we have recovered \revise{$s[1,k]$}, we show that we can recover $s[k+1]$ with the query sequence \revise{$q^{(k+2)} = s[1,k]s[k](\frac{1}{2})^{\ell-k-1}$}. Noting that each $\frac{1}{2}$ in $q^{(k+2)}$ corresponds to at least a $\frac{1}{2}$ cost in the query result, we have $d_{\revise{\DTW}}(s, q^{(k+2)})\geq (\ell-k-1)/2$. If $s[k+1] = s[k]$, then \revise{$s[1,k+1]$} and \revise{$q^{(k+2)}[1,k+1]$} can be perfectly matched, so $d_{\revise{\DTW}}(s, q^{(k+2)}) = (\ell-k-1)/2$. Otherwise, we claim that $d_{\revise{\DTW}}(s, q^{(k+2)}) > (\ell-k-1)/2$. If the cost corresponding to $q^{(k+2)}[k+1] > 0$, we would already have $d_{\revise{\DTW}}(s, q^{(k+2)}) > (\ell-k-1)/2$, so we can assume that the cost corresponding to $q^{(k+2)}[k+1]$ is $0$. Since $q^{(k+2)}[k+1] = s[k] \neq s[k+1]$, we know that $q^{(k+2)}[k+1]$ cannot be matched with $s[k+1]$. Suppose $q^{(k+2)}[k+1]$ is matched with substring \revise{$s[u,u+t]$} in the optimal DTW matching, where $t\geq 0$ and $s[u] = s[u+1] = \cdots  = s[u+t] = q^{(k+2)}[k+1] = s[k]$. Since $k+1 \notin [u, u+t]$, we either have $k+1 > u+t$ or $k+1 < u$. If $k+1 > u+t$, since $\forall u+t < j \leq \ell$ we have $s[j]$ matched to a $\frac{1}{2}$, the total cost would be at least $(\ell-(u+t))/2 > (\ell-k-1)/2$. Otherwise if $k+1 < u$, note that \revise{$s[1,u]$} are matched to \revise{$q^{(k+2)}[1,k+1]$} in the optimal DTW matching. Since $s[k+1]\neq s[k]$, the number of runs in \revise{$s[1,k+1]$} would be greater than the number of runs in \revise{$q^{(k+2)}[1,k+1]$} by $1$. Thus,  the number of runs in \revise{$s[1,u]$} would be greater than the number of runs in \revise{$q^{(k+2)}[1,k+1]$} by at least 1, and they cannot be perfectly matched. Therefore, the cost corresponding to \revise{$q^{(k+2)}[1,k+1]$} would be greater than $0$, yielding a total cost of greater than $(\ell-k-1)/2$.

By induction, we can recover the input sequence of maximum length $n$ with $n+1$ queries. %
\end{proof}

\section{Recovery with Non-Adaptive Edit Distance Oracle Queries}
\label{sec:edit_non}

We begin with a lower bound for edit distance.
\begin{theorem}
\label{thm:edit_lower_bound}
For a binary alphabet $\bits$, any algorithm to recover an arbitrary input sequence $s \in \bits^\ell$ where $0 \leq \ell \leq n$ by querying the  Levenshtein distance to a set of sequences of length $\mc{O}(n)$  requires a query complexity of $\Omega(n/\log n)$.

\end{theorem}
\begin{proof}
For each query \revise{of length $\mc{O}(n)$}, the result would be an integer $d = \mc{O}(n)$.
\revise{\linelabel{edit_lower_range}Without loss of generality, assume the query is of length $an+b$ where $a, b$ are non-negative constant integers and $b<n$.
For an arbitrary input sequence with length $\leq n$, the query result falls into the range of $[(a-1)n+b, an+b]$ if $a>0$ (or $[0, n]$ if $a=0$), yielding $n+1$ possibilities.}
For $0\leq k \leq n$, the number of different sequences of length $k$ is $2^k$, and the total number of sequences of length no greater than $n$ would be $\sum_{k=0}^{n}2^{k} = 2^{n+1}-1$. Thus, to distinguish all possible sequences, one would need at least $\log_{\revise{n+1}}(2^{n+1}) = (n+1)/\log(n+1) \in \Omega(n/\log n)$ queries.

\revise{
Note that this information theoretical proof only applies to deterministic algorithms.
Next we give a proof if one is allowed to use a randomized algorithm.
To show this, we introduce a one-way two-party communication game called \emph{INDEX}.

\begin{definition}[INDEX Game \citep{indexGame}]
    Consider two players Alice and Bob. Alice and Bob have access to a common public coin and their computation can depend on this. Alice holds an $n$-bit string $x\in \bits^n$ and is allowed to send a single message $M$ to Bob (i.e., this is a one-way protocol). Bob has an index $i \in [n]$ and his goal is to learn $x[i]$, i.e., $\Pr_{r}[\mathcal{O}ut(M)=x[i]]\geq \frac{2}{3}$). 
\end{definition}

It is shown in \citep{indexGame} that the above problem requires $|M| = \Omega(n)$. To reduce our recovery problem from the INDEX game, let $R$ be an adaptive randomized recovery algorithm which works as follows.
First, Alice randomly selects a query $q_1$ based on the first part $r_1$ of the shared public coin, and computes $d(q_1, x)$. Alice then adaptively selects a set of queries $q_i$, where each $q_i$ is chosen based on disjoint parts $r_1, \ldots, r_i$ of the public coin, as well as the responses $d(q_1, x), \ldots, d(q_{i-1}, x)$ to previous queries. Alice then sends all query results $d(q_i, x)$ to Bob as the message $M$.

We now show that, if the algorithm $R$ is correct w.p. $2/3$, then $M$ contains $\Omega(n/\log n)$ query results.
Given the success probability $2/3$ of $R$, from message $M$, Bob can reconstruct the string $x$ w.p. at least $2/3$, so Bob can learn each bit of $x$ w.p. at least $2/3$.
According to \citep{indexGame}, $|M| = \Omega(n)$ bits. Since each distance query result contains at most $\mathcal{O}(\log n)$ bits, it follows that $\Omega(n/\log n)$ queries are required.
}
\end{proof}

\subsection{Exact Recovery with Extra Character(s)}

We now move on to the analysis of the upper bound for edit distance with the assistance of extra character(s).
\revise{The following theorem uses 1 extra character in the extended alphabet to construct query sequences. We note that for edit distance, using more than 1 extra character in the extended alphabet does not help recover the input sequence, because the edit distance oracle only counts the edit operations made from transforming one sequence to another. Different characters result in the same edit cost.}

\begin{theorem}[Non-adaptive Strategy for Edit Exact Recovery with 1 Extra Character]
\label{theorem:edit}
For a binary alphabet $\bits$ and an input sequence $s \in \bits^\ell$ where $0 \leq \ell \leq n$, there exists an algorithm to recover the input sequence $s$, given a set of \revise{\linelabel{On_N_1_2}$n+1 \in \mc{O}(n)$} query sequences $\revise{\mc{Q}}$ of length $\leq n$ and the exact Levenshtein distance of $s$ to each query sequence $q \in \revise{\mc{Q}}$, where an extra character $2$ is allowed in the query sequences.
\end{theorem}

\stitle{Proof of Theorem \ref{theorem:edit}.} 
\revise{The intuition of our proof is to build an oracle that returns the number of 1's in the first $j$ characters of the input sequence $s$.
Then querying the oracle with all possible $j$'s (where $j \in [n]$) implies a recovery of the input sequence.
Note that this oracle calls the edit distance oracle as a subroutine.
The following lemma shows the existence of such an oracle.}

\begin{lemma}
\label{lem:edit_distance_mapping}
\revise{Let $s \in \bits^{n'}$ be a non-empty sequence with length $\len(s) = n' \leq n$.
Consider a sequence $s' = 1^{j}2^{n-j}$, where $j \in [n']$, and $2$ denotes a random character not in the binary alphabet $\bits$.
Let $k$ denote the number of $1$'s in the substring $s[1, j]$ (i.e., the first $j$ characters of $s$).
Then the edit distance between $s$ and $s'$ is equal to $n-k$.}
\end{lemma}

\begin{proof}
\revise{We prove the lemma in two steps. 
First, we prove that the number of operations required in the transformation from $s$ to $s'$ is greater than or equal to $n-k$.
Second, we show the existence of a sequence of operations that transforms $s$ to $s'$ in exactly $n-k$ steps.

To formally prove the first step, we perform a case analysis on the $j$-th character of $s$, i.e., $s[j]$, being 0 or 1.
When $s[j]=0$, the following claim shows $d_L(s, s') \geq n-k$.}

\begin{claim}
\label{claim:edit_distance_mapping}
\revise{If $s[j]=0$, then $d_L(s, s') \geq n-k$.}
\end{claim}

\noindent
\emph{Proof of claim.}
Let $s = [\text{prefix}]0[\text{suffix}]$, where the length of the prefix is $\len([\text{prefix}]) = j-1$.
Recall sequence $s' = 1^{j}2^{n-j}$.
The edit distance between $s$ and $s'$ can be regarded as the number of operations required in the transformation from $s$ to $s'$. 
This transformation from $s$ to $s'$ leads to a \textit{sequence} of operations of insertion, deletion, and substitution.
For an optimal transformation sequence, swapping two adjacent operations in this sequence generates another valid sequence of operations of the same length.%
We can therefore assume all the deletion operations are performed in the beginning, and we denote the number of deletions by $d$.
Let $t$ be the number of $0$'s in $[\text{prefix}]$, and 
suppose $d \geq t + 1$.
We can assume that all entries in $[\text{suffix}]$ are $1$ and any $0$ is deleted before a $1$ is deleted; indeed, these assumptions will not increase the edit distance.
After the deletion operations, $s$ is a sequence of $n'-d$ 1's. 
If $n'-d \leq j$, then we need an additional $n-(n'-d)$ insertions to recover $s'$. 
Thus the total cost is $d + n-(n'-d) \geq d + n - j \geq (t + 1) + n - j = n - k $. 
It remains to consider the case that $d < t + 1$. 
At this point $s$ is a sequence of length $n'-d$ containing $(t + 1)-d$ 0's among its first $j$ entries, and remaining $1$'s.  
Since there are no more deletions, any 1's occurring after the $j$-th entry must be substituted to a 2. 
There are $(n'-d)-j$ such 1's that each cost $1$. 
Also, each of the $(t + 1)-d$  $0$'s among the first $j$ entries costs one for a substitution. Finally, we need at least $n-n'+d$ insertions to obtain equal-length sequences. So the total cost is at least $\underbrace{(d)}_{\text{\#deletions}} + \underbrace{(n-n'+d)}_{\text{\#insertions}} + \underbrace{[(n'-d) - j]}_{\text{\#substitutions of 1's}}  + \underbrace{[(t + 1)-d]}_{\text{\#substitutions of 0's}} = n -j + t + 1 = n - k$. 
This completes all cases.

\bigskip

\revise{To finish the case analysis, now we consider the case that $s[j]=1$.
We define $s'' = (s[1,j-1]) 0 (s[j+1, n'])$.
Note $s''$ is obtained by substituting the $j$-th character of $s$ from 1 to 0, hence $d_L(s, s'') = 1$.
We have $k$ 1's in $s[1, j]$, so we have $(k-1)$ 1's in $s''[1, j]$.
By Claim \ref{claim:edit_distance_mapping}, we have $d_L(s', s'') \geq n-k +1$.
Since edit distance is \emph{a metric}, by triangle inequality, $d_L(s', s) + d_L(s, s'') \geq d_L(s', s'') \geq n-k + 1$.
Therefore, $d_L(s, s') \geq n - k + 1 -1 = n-k$.

Now, for the second step, we give a valid sequence of operations to transform $s$ to $s'$ in exact $n-k$ steps.
1) insert $2$'s to the end of $s$ such that $s$ and $s'$ have the same length. 
This results in $n-n'$ insertions.
2) for every index $i \in [n]$, substitute $s[i]$ to $s'[i]$ if they are different in the first place.
The number of operations is counted as follows. 
For $i \in [1, j]$, it requires $j - k$ substitutions since there are $(j-k)$ $0$'s.
For $i \in [j+1, n']$, it requires $n' - j$ substitutions since we need to substitute every character to $2$.
For $i \in [n'+1, n]$, it requires no substitutions.
Therefore, we have $n-n'+j-k+n'-j = n-k$ operations in total.}
\end{proof}

\smallskip

\noindent\textit{Query Sequence Construction.}
\revise{
We introduce an additional wildcard character which is not in the input sequence alphabet.
Using this newly introduced character, the $n+1$ query sequences are constructed as follows.
We use an empty sequence together with $n$ sequences of the form of $1^j2^{n-j}$ for $j=1, \ldots, n$ where $2$ denotes a ``not-in-the-alphabet'' character.}

\medskip

\noindent\emph{Algorithm to recover input sequence $s$.}
\revise{We now give an algorithm to recover $s$ using the query sequence set to complete the proof of Theorem \ref{theorem:edit}.
From the query result of the empty sequence, we know the length of the input sequence ($n'$).
If $n'=0$, we know the input sequence is empty as well.
Otherwise, consider the query results of the sequences $1^{j}2^{n-j}$, for $j \in [n']$. 
By Lemma \ref{lem:edit_distance_mapping}, we know the number of $1$'s ($k$) in the first $j$ characters of $s$ ($j \in [n']$), which implies a complete recovery of $s$.}
The exact recovery algorithm is presented in Algorithm \ref{alg:edit_recover}. 
We note that the order of the query sequences matters.
\hfill$\blacksquare$

\stitle{Remark.} Note that the exact length of sequence $s$ is unknown. This algorithm works non-adaptively for any sequence $s$ with length $\leq n$.

\subsection{Exact Recovery without Extra Characters}

\begin{theorem}
\label{theorem:edit_non_adaptive_n2}
For a binary alphabet $\bits$ and an input sequence $s \in \bits^\ell$ where $0 \leq \ell \leq n$, there exists an algorithm to recover the input sequence $s$, given \revise{$\frac{1}{2}(n^2 + 3n) \in \mc{O}(n^2)$} query sequences $\revise{\mc{Q}}$ of length $\leq n$ and the exact Levenshtein distance of $s$ to each query sequence $q_i \in \revise{\mc{Q}}$, without extra characters.
\end{theorem}

\begin{proof}
    The construction can be obtained by naturally extending the query set in the proof of Theorem \ref{theorem:adaptive_edit_N+1} to all lengths $\ell \in [n]$.
    This gives us \revise{\linelabel{sum-n-n}$\sum_{\ell=1}^n (\ell +1) = \frac{1}{2}(n^2 + 3n) \in \mc{O}(n^2)$} queries.
\end{proof}

\begin{algorithm2e*}[t]
\LinesNumbered
\caption{Exact Recovery Algorithm via Queries to an Edit Distance Oracle}
\label{alg:edit_recover}
\SetKwInput{KwInput}{Input}                %
\SetKwInput{KwOutput}{Output}              %
\DontPrintSemicolon
  \KwInput{Non-adaptive query sequences $\revise{\mc{Q}} = \{ q^{(1)}, q^{(2)}, \dots, q^{(n+1)} \}$;
      The edit distance \emph{query result} from the sequence for recovery to each query sequence $\revise{\mc{R}} = \{ d^{(1)}, d^{(2)}, \dots, d^{(n+1)} \}$.\;}
\KwOutput{The sequence for recovery $s$.}
\SetKwFunction{FMain}{{\sc RecoveryEdit}}
  \SetKwProg{Fn}{Function}{:}{}
    \Fn{\FMain{$\revise{\mc{Q}}, \revise{\mc{R}}$}}{
    sequence = [] {\color{blue}\Comment{Initialize the sequence for recovery.}} \\
    sequence.append($n-d^{(2)}$) \\
    \For{$i \in [2, n+1]$} {
        sequence.append($d^{(i)} - d^{(i+1)}$)
    }
    sequence = $\phi$ if $d^{(1)}=0$, else sequence[1, $d^{(1)}$] {\color{blue}\Comment{$d^{(1)}$ is the distance to the empty string.}} \\
    \Return s $\coloneqq$ sequence
    }
\end{algorithm2e*}

\section{Recovery with Non-Adaptive DTW Distance Oracle Queries}
\label{sect:dtw}

\subsection{Hardness Result without Extra Characters}
\label{sec:hardness_dtw}

\begin{theorem}[Indistinguishable Sequences by Binary Queries with DTW Oracle]
\label{theorem:dtw_hardness}
There exists a pair of input sequences $s$ and $s'$ such that for any query sequence $q$, $d_{\revise{\DTW}}(s, q)$ = $d_{\revise{\DTW}}(s', q)$. That is, $s$ and $s'$ cannot be distinguished by DTW Distance Oracle queries without extra characters. 
\end{theorem}

\begin{proof}
\revise{\linelabel{witness}We can prove this theorem by constructing a witness pair of input sequences.}
Consider the following pair of input sequences: $s = 010110$ and $s'=011010$.
We argue that this pair of input sequences cannot be distinguished by any binary sequence query $q$.

First, for query sequences that only consist of $0$, it is obvious $d_{\revise{\DTW}}(s, q) = d_{\revise{\DTW}}(s', q) = 3$.
Then we only need to consider query sequences containing 1('s).
To see $d_{\revise{\DTW}}(s, q) = d_{\revise{\DTW}}(s', q)$ in this case, we will show (1) $d_{\revise{\DTW}}(s, q) \leq d_{\revise{\DTW}}(s', q)$ and (2) $d_{\revise{\DTW}}(s, q) \geq d_{\revise{\DTW}}(s', q)$ hold simultaneously.

Note that the sequence $s$ contains three 1's.
To prove case (1), we show that there exists an optimal DTW matching satisfying the following properties:

\begin{itemize}
    \item[a)] The first 1 in $s$ is matched to a \emph{substring (c.f. Definition \ref{def:subsequence_substring})} of $q$ that begins and ends with both 1's;
    
    \item[b)] The second 1 in $s$ is matched to a substring of $q$ that begins with 1;
    
    \item[c)] The third 1 in $s$ is matched to a substring of $q$ that ends with 1.
\end{itemize}

To see the existence of such an optimal matching, we would like to show that, if any one of these properties is violated, we can find another matching with at most the same cost that does not violate these properties.
We take property a) as an example to illustrate this.
If a) is violated, then the substring in $q$ that the first 1 in $s$ gets matched to contains at least a $0$ in the beginning or the end, or both.
If this substring contains both a $0$ and a $1$, then we can map the $0$ at the beginning (or in the end) to the $0$ on the left (or right) side to the first $1$ to obtain a matching with lower cost.
We consider the substring that contains only $0$.
In the optimal matching, the first $1$ in $s$ cannot get matched to more than one $0$ because this will yield more cost than necessary.
Then it reduces to the case where $1$ is matched to a single $0$.
In this case, if the left $0$ in $s$ is matched to a substring that contains at least a $1$ in $q$, then matching the first 1 (in $s$) to this (these) 1('s) leads to a matching with lower cost, since the right 0('s) in the substring can be matched to the 0 on the right to the first 1 in $s$.
Then this leaves the discussion for the case that the first ``01'' in $s$ is matched to a substring with only 0('s) in $q$.
For ease of presentation, we denote this substring by ``$ss$-0''.
The second 0 in $s$ can always be matched to $ss$-0 because this will not yield cost and therefore we know in a potential optimal matching the first ``010'' can be matched to $ss$-0.
Since we know $q$ contains at least a single 1, this(these) 1('s) will be matched to character after the first ``010'' (i.e., the second 1)  in $s$.
We then argue that we can change this matching to obtain an equally optimal matching without violating the properties:
i) the first 0 in $s$ is matched to the substring before the first 1 in $q$; ii) the first 1 and the second 0 in $s$ are simultaneously matched to the first 1 in $q$; iii) after the first ``010'' in $s$, the matching does not change.
In this new matching, there is a cost of 1 saved and generated due to the matching changing on the first 1 and the second 0 in $s$, and therefore the overall DTW cost does not change and the matching remains optimal.
For the rest of the properties, the cases and proofs are similar.
We therefore omit the detailed analyses.

Next, we will show that, given the matching (between $s$ and $q$) with these three properties, we can find a matching between $s'$ and $q$ that will generate DTW cost at most $c$ (that is, $d_{\revise{\DTW}}(s', q) \leq c$).
In particular, we give the following reduction in two matchings.

\begin{itemize}
    \item[a)] All 0's in $s'$ get matched to the same substring in $q$ as all 0's in $s$;
    
    \item[b)] The first 1 and the second 1 in $s'$ get matched to the substring in $q$ that matches the first 1 in $s$;
    
    \item[c)] The third 1 in $s'$ gets matched to the two substrings in $q$ that match the second 1 and the third 1 in $s$.
\end{itemize}

By this matching, the cost between $s'$ and $q$ is exactly $c$.
We do not need to know if this matching is optimal for $d_{\revise{\DTW}}(s', q)$ but this shows $d_{\revise{\DTW}}(s', q) \leq c$.
We note that these three properties hold for any query sequence that contains at least a single 1, because the single 1 can be a substring of this query sequence to satisfy the properties.
Thus, this analysis covers all possible cases of a binary query sequence.
By symmetry, a similar construction can be shown for the opposite side and the conclusion is $d_{\revise{\DTW}}(s, q) \geq d_{\revise{\DTW}}(s', q)$.
Combining the two parts of the proof, we obtain that $d_{\revise{\DTW}}(s, q) = d_{\revise{\DTW}}(s', q)$ for any binary query $q$.
\end{proof}

\subsection{Recovery without Extra Characters w.r.t. Equivalence Classes}

As indicated by Theorem \ref{theorem:dtw_hardness}, there exist input sequences that cannot be distinguished by DTW distance oracle queries.
For ease of presentation, we say that any two different input sequences $s$ and $s'$ are \emph{distinguishable} if $s$ and $s'$ can be distinguished by DTW Distance Oracle queries.
We categorize mutually indistinguishable sequences into \emph{equivalence classes}.
In this context, using binary queries, the best solution we can provide in this problem setting is to recover those input sequences up to their equivalence class.

\linelabel{characterization_is_hard}The characterization of the set of \emph{indistinguishable} binary sequences, given a parameterized sequence length $n$, is not so simple to describe (which can be seen from Observation 4 of \citep{DBLP:conf/cpm/SchaarFN20}). 
However, we can propose an optimal query strategy in this setting to distinguish all distinguishable sequences and \emph{prove optimality} by making use of the reduction between the calculation of DTW distance and the min 1-separated sum problem \citep{abboud2015tight,DBLP:conf/cpm/SchaarFN20}.
We introduce the necessary results from \citep{DBLP:conf/cpm/SchaarFN20} below and interpret them in our setting.

\begin{definition}[Min 1-Separated Sum (MSS), \citep{DBLP:conf/cpm/SchaarFN20}]
The min 1-separated sum (MSS) problem takes the inputs of a sequence $\left(b_{1}, \ldots, b_{m}\right)$ of $m$ positive integers and an integer $r \geq 0$.
The problem is to select $r$ integers $b_{i_{1}}, \ldots, b_{i_{r}}$ with $1 \leq i_{1}<i_{2}<\cdots<i_{r} \leq m$ and $i_{j}<i_{j+1}-1$ for all $1 \leq j<r$ such that $\sum_{j=1}^{r} b_{i_{j}}$ is minimized.
We say $((b_{i_{1}}, \ldots, b_{i_{r}}), r)$ is an MSS instance.
\end{definition}

\begin{theorem}[DTW-to-MSS Reduction, \citep{DBLP:conf/cpm/SchaarFN20}, Theorem 2]
\label{theorem:dtw_mss}
Let $x \in\{0,1\}^{m}$ and $y \in\{0,1\}^{n}$ be two binary strings such that $x[1]=y[1], x[m]=y[n]$, and $\#\textsc{runs}(x) \geq\#\textsc{runs}(y)$. Then, the DTW distance between $x$ and $y$, i.e., $d_{\revise{\DTW}}(x, y)$, equals the sum of a solution for 
$\MSS\biggl(\Bigl(\textsc{lor}(x,2), \ldots,\textsc{lor}(x, \#\textsc{runs}(x)-1)\Bigl), (\#\textsc{runs}(x)-\#\textsc{runs}(y))/2 \biggl)$. 
\end{theorem}
\noindent
For ease of presentation, we will use $\revise{\MSS}(x, (\#\textsc{runs}(x)-\#\textsc{runs}(y))/2)$ to represent the same MSS instance.
\begin{theorem}[\cite{DBLP:conf/cpm/SchaarFN20}, Observation 4]
\label{theorem:obv_4}
Let $x \in\{0,1\}^{m}, y \in\{0,1\}^{n}$ with $m^{\prime}:=\#\textsc{runs}(x) \geq n^{\prime}:=\#\textsc{runs}(y)$. Further, let $a:=\textsc{lor}(x,1), a^{\prime}:=\textsc{lor}(x,m^{\prime}), b:=\textsc{lor}(y,1)$, and $b^{\prime}:=\textsc{lor}(y,n^{\prime})$. The following holds:

\noindent
If $x[1] \neq y[1]$, then:
\begin{align*}
    & d_{\revise{\DTW}}(x, y)  \nonumber \\    
& = \begin{cases}\max (a, b), & m^{\prime}=n^{\prime}=1; \\ 
a+d_{\revise{\DTW}}(x[a+1, m], y), & m^{\prime}>n^{\prime}=1;  \nonumber\\ 
\parbox{5.5cm}{$\min \left(a+d_{\revise{\DTW}}(x[a+1, m], y), \right. \\
    \left. \quad \quad ~~ b+d_{\revise{\DTW}}(x, y[b+1, n])\right)$} & n^{\prime}>1.\end{cases} 
\end{align*}

\noindent
If $x[1]=y[1]$ and $x[m] \neq y[n]$, then:
\begin{align*}
& d_{\revise{\DTW}}(x, y)  \nonumber \\    
& = \begin{cases}a^{\prime}+d_{\revise{\DTW}}\left(x\left[1, m-a^{\prime}\right], y\right), & n^{\prime}=1; \\ 
\parbox{6.3cm}{$\min \left(a^{\prime}+d_{\revise{\DTW}}\left(x\left[1, m-a^{\prime}\right], y\right), \right. \\
    \left. \quad \quad ~~
    b^{\prime}+d_{\revise{\DTW}}\left(x, y\left[1, n-b^{\prime}\right]\right)\right)$} & n^{\prime}>1.\end{cases}
\end{align*}
\end{theorem}

In Theorem \ref{theorem:obv_4}, we call $a, b, a'$ and $b'$ (which are the length of first/last blocks of $x$ or $y$) \textit{offsets}. Theorem \ref{theorem:obv_4} actually states that, for two sequences $x, y$ with different starting and ending characters, by removing the first/last run of $x$ or $y$, calculating $d_{\revise{\DTW}}(x, y)$ can be reduced to calculating the offset and solving a DTW sub-problem where the sub-sequences start and end with the same character.

To illustrate how we can transfer a DTW problem to an MSS instance, we give a concrete example here.
Let $s = 010110$, $q^{(1)} = 010$, $q^{(2)} = 011$.
We first consider the calculation of the DTW distance between $s$ and $q^{(1)}$.
Since the first and the last blocks of $s$ and $q^{(1)}$ are the same and the number of runs of $s$ is more than that of $q^{(1)}$, we can directly apply Theorem \ref{theorem:dtw_mss}, where we have the MSS instance $\revise{\MSS}((1, 1, 2), 1)$ and the DTW distance is equal to the solution to this MSS instance.
As for the computation of the DTW distance between $s$ and $q^{(2)}$, we need to first apply Theorem \ref{theorem:obv_4} since the last blocks of $s$ and $q^{(2)}$ are different.
By Theorem \ref{theorem:obv_4}, $d_{\revise{\DTW}}(s, q^{(2)}) = \min (1 + d_{\revise{\DTW}}(\text{``}01011\text{''}, q^{(2)}), 2 + d_{\revise{\DTW}}(s, \text{``}0\text{''}))$.
Then by Theorem \ref{theorem:dtw_mss}, the calculation of $d_{\revise{\DTW}}(\text{``}01011\text{''}, q^{(2)})$ yields the MSS instance $\revise{\MSS}((1, 1), 1)$ and computing $d_{\revise{\DTW}}(s, \text{``}0\text{''}))$ is equivalent to
$\revise{\MSS}((1, 1, 2), 2)$.

\noindent
We now show the lower bound on the query complexity using binary queries.

\begin{theorem}[\revise{Lower Bound for DTW Equivalence Class Recovery}]
\label{thm:dtw_no_extra_lowerbound}
    For binary alphabet $\bits$, any algorithm to recover an arbitrary input sequence $s \in \bits^\ell$, where $0 \leq \ell \leq n$, up to equivalence class, by querying the DTW distance to a set of sequences, requires a query complexity of $\Omega(n)$.
\end{theorem}

\begin{proof}
We will assume the input sequence is of length $\leq n$ and all the query sequences are of length $\mc{O}(n)$, when the context is clear in the proof.
\begin{claim}
\label{claim:MSS_133232}
    Given $\revise{\MSS}_1((\underbrace{3, \ldots, 3}_\text{a}, 1, 3, 3, \underbrace{3, \ldots, 3}_\text{b} ), x)$ and $\revise{\MSS}_2((\underbrace{3, \ldots, 3}_\text{a}, 2, 3, 2, \underbrace{3, \ldots, 3}_\text{b} ), x)$, we claim that when $x=1$, $\revise{\MSS}_1 \neq \revise{\MSS}_2$ and when $2 \leq x\leq (a+b+4)/2$, $\revise{\MSS}_1 = \revise{\MSS}_2$.
\end{claim}

\bigskip
\noindent\textit{Proof of claim \ref{claim:MSS_133232}.}
By the definition of MSS, when $x=1$, $\revise{\MSS}_1 = \min (3, \ldots, 3, 1, 3, 3, 3, \ldots, 3) = 1$ and $\revise{\MSS}_2 = \min (3, \ldots, 3, 2, 3, 2, 3, \ldots, 3) = 2$.
When $2 \leq x < (a+b+4)/2$, $\revise{\MSS}_1 = 3x-2 = \revise{\MSS}_2 $.
When $x = (a+b+4)/2$, if $a$ is odd, $\revise{\MSS}_1 = 3x = \revise{\MSS}_2$; otherwise, $\revise{\MSS}_1 = 3x-2 = \revise{\MSS}_2 $.

\bigskip
\begin{claim}
\label{claim:mss_unique_diff}
Let $\mc{Q}$ be a query set which can distinguish any pair of binary input sequences \revise{\linelabel{that-are-distinguishable}that are distinguishable}.
For $\forall c \in \mathbb{N}_{+}$ such that $6c+9 \leq n$, $\exists q \in \mc{Q}$ such that \#\textsc{runs}($q$) $\in [2c, 2c+6]$.
\end{claim}

\bigskip
\noindent\textit{Proof of claim \ref{claim:mss_unique_diff}.} 
Consider two input sequences, $s = 01^301^3(0^31^3)^c0$ and $s' = 01^30^21^30^21^3(0^31^3)^{c-1}0$, where $\revise{\len}(s) = \revise{\len}(s') = 6c+9 \leq n$.
We know that \#\textsc{runs}($s$) $=$ \#\textsc{runs}($s'$) $= 2c+5$.
First we show that $s$ and $s'$ are distinguishable.
Let $q^\dagger = 0(10)^c10$, where \#\textsc{runs}($q^\dagger$) $= 2c+3$.
According to Theorem \ref{theorem:dtw_mss}, $d_{\revise{\DTW}}(s, q^\dagger) = \revise{\MSS}((3, 1, 3, 3, \ldots), 1) = 1$ and $d_{\revise{\DTW}}(s', q^\dagger) = \revise{\MSS}((3, 2, 3, 2, 3, \ldots), 1) = 2$.
Thus, $q^\dagger$ can distinguish $s$ and $s'$.

Next we show, for any query $q$ such that \#\textsc{runs}($q$) $\geq 2c+7$ or $\leq 2c-1$, $d_{\revise{\DTW}}(s, q) = d_{\revise{\DTW}}(s', q)$.
Note that, to compute the DTW distances, according to Theorem \ref{theorem:obv_4}, we may remove the first/last blocks of $s$ (and $s'$) or $q$ to reduce to the case of Theorem \ref{theorem:dtw_mss}. 
Since $s$ and $s'$ have the same first and last blocks, 
the offsets while reducing to the case of Theorem \ref{theorem:dtw_mss} are the same. To prove that $d_{\revise{\DTW}}(s, q) = d_{\revise{\DTW}}(s', q)$, we only need to prove that for each possible reduction, the corresponding reduced MSS instances have the same sum of solutions (see Example \ref{example:mss_reduction_1} for illustration). 

\begin{example}
\label{example:mss_reduction_1}
To illustrate, take $c=2$ and $q=101$. In this case, we would have $s = 01^301^30^31^30^31^30$, $s' = 01^30^21^30^21^30^31^30$.
According to Theorem \ref{theorem:obv_4}, we have
 $d_{\revise{\DTW}}(s, q) = \min (1 + d_{\revise{\DTW}}(s[2, \revise{\len}(s)], q), 1 + d_{\revise{\DTW}}(s, q[2, \revise{\len}(q)])) = \min (1 + d_{\revise{\DTW}}(s[2, 21], q), 1 + d_{\revise{\DTW}}(s, q[2, 3]))$.
 Then $d_{\revise{\DTW}}(s[2, 21], q) = \min (1 + d_{\revise{\DTW}}(s[2, 20], q), 1 + d_{\revise{\DTW}}(s[2, 21], q[1, 2]))$.
 Also, we note that $d_{\revise{\DTW}}(s, q[2, 3])) = \min (1 + d_{\revise{\DTW}}(s[1, 20], q[2, 3]), 1 + d_{\revise{\DTW}}(s, q[2]))$.
 
 Therefore, to show that $d_{\revise{\DTW}}(s, q) = d_{\revise{\DTW}}(s', q)$, we only need to prove that 
 $$\begin{cases}d_{\revise{\DTW}}(s[2, 21], q) = d_{\revise{\DTW}}(s'[2, 21], q); \\ d_{\revise{\DTW}}(s[2, 20], q[1, 2]) = d_{\revise{\DTW}}(s'[2, 21], q[1, 2]); \\ d_{\revise{\DTW}}(s[1, 20], q[2, 3]) = d_{\revise{\DTW}}(s'[1, 20], q[2, 3]);\\
 d_{\revise{\DTW}}(s, q[2, 3] ) = d_{\revise{\DTW}}(s', q[2, 3] ), \end{cases}$$
where each of the 4 cases corresponds to an MSS instance.
\end{example}

Suppose after applying Theorem \ref{theorem:obv_4}, $s$ and $q$ are reduced to sub-sequences $s^*$ and $q^*$ (where $s^*$ and $q^*$ have the same beginning and ending characters), while $s'$ and $q$ are reduced to $s'^*$ and $q^*$. Now we calculate $d_{\revise{\DTW}}(s^*, q^*)$ and $d_{\revise{\DTW}}(s'^*, q^*)$ according to Theorem \ref{theorem:dtw_mss}. Suppose $s^*$ and $s'^*$ have $k^*$ runs and $q^*$ have $l^*$ runs.

Case 1. If \#\textsc{runs}($q$)$\geq 2c+7 = $\#\textsc{runs}($s$)$+2$, then $k^* \leq l^*$, by Theorem \ref{theorem:dtw_mss} the generated MSS instance only depends on $q^*$ and $k^*$. Thus, $d_{\revise{\DTW}}(s^*, q^*) = d_{\revise{\DTW}}(s'^*, q^*)$.

Case 2. If \#\textsc{runs}($q$)$\leq 2c-1$, then $k^* > l^*$ and we have the MSS instances $\revise{\MSS}(s^*, (k^* - l^*)/2)$ and $\revise{\MSS}(s'^*, (k^* - l^*)/2)$. Note that, $(k^* - l^*)/2 \geq (\#\textsc{runs}(s)-\#\textsc{runs}(q)-2)/2 \geq ((2c+5)-(2c-1)-2)/2 = 2$.
By Claim \ref{claim:MSS_133232}, we have $\revise{\MSS}(s^*, (k^* - l^*)/2) = \revise{\MSS}(s'^*, (k^* - l^*)/2)$. Thus, $d_{\revise{\DTW}}(s^*, q^*) = d_{\revise{\DTW}}(s'^*, q^*)$.

Combining case 1 and case 2, we know $d_{\revise{\DTW}}(s, q) = d_{\revise{\DTW}}(s', q)$ when \#\textsc{runs}($q$) $\geq 2c+7$ or $\leq 2c-1$.
Since there always exists $q\in \mc{Q}$ that can distinguish $s$ and $s'$, we know that \#\textsc{runs}($q$) $\in [2c, 2c+6]$, which proves the claim.

\bigskip

Let $c' \in \mathbb{N}_{+}$ satisfy $24 c' -9 \leq n$.
Let $c = 4c' -3$. We have $6c+9 \leq n$.
By Claim \ref{claim:mss_unique_diff}, $\exists q \in \mc{Q}$ such that \#\textsc{runs}($q$) $\in [2c, 2c+6]$, i.e., \#\textsc{runs}($q$) $\in [8c'-6, 8c']$.
For $c' = 1, 2, \ldots$, intervals  $[8c'-6, 8c']$ are disjoint.
Therefore, there should be at least $\lfloor (n+9)/24 \rfloor = \Omega(n)$ queries in the set $\mc{Q}$.
\end{proof}

With these useful results at hand, now we prove the following results for recovering sequences using the DTW distance oracle with only binary queries.

\begin{theorem}[Non-adaptive Strategy for DTW Equivalence Class Recovery]
\label{theorem:oz_equivalence}
There exists a set $\mc{Q}$ of \revise{$2n \in \mc{O}(n)$} queries, each of which has $\mc{O}(n)$ length, such that for any two different input sequences $s$ and $s'$, $s$ and $s'$ are distinguishable $\iff$ $s$ and $s'$ can be distinguished by $\mc{Q}$.
\end{theorem}

\begin{proof}
First ($\Leftarrow$), for any given query set $\mc{Q}$ and two different input sequences $s$ and $s'$, if $s$ and $s'$ can be distinguished by $\mc{Q}$ then $s$ and $s'$ are distinguishable.
Then we need to prove the opposite side ($\Rightarrow$).
To see this, we construct the following query set $\mc{Q}$ \revise{of size $2n \in \mc{O}(n)$} and prove the contrapositive: if $s$ and $s'$ cannot be distinguished by $\mc{Q}$, then $s$ and $s'$ are not distinguishable.

Let $$z_i = \begin{cases} 0^{n},\ i = 1; \\ 0^{n}1(01)^{m-1}0^{n} ,\ i = 2m+1; \\ 0^{n}(10)^{m-1}1^{n},\ i = 2m,\end{cases}$$ and $$o_i = \begin{cases} 1^{n},\ i = 1; \\ 1^{n}0(10)^{m-1}1^{n} ,\ i = 2m+1; \\ 1^{n}(01)^{m-1}0^{n},\ i = 2m,\end{cases}$$
where $1\leq i \leq n$ and $m$ is an positive integer. It is clear that $o_i$'s and $z_i$'s are of $\mc{O}(n)$ length. Let $\mc{Q} = \{o_i|1\leq i\leq n\}\bigcup\{z_i|1\leq i\leq n\}$. We show that given any two different input sequences $s$ and $s'$, if $s$ and $s'$ cannot be distinguished by $\mc{Q}$ then $s$ and $s'$ are not distinguishable.

\begin{claim}
\label{claim:condensed_exp}
Given two different input sequences $s$ and $s'$, if the condensed expressions of $s$ and $s'$ are different, then $s$ and $s'$ can be distinguished by $\mc{Q}$. 
\end{claim}
\noindent\textit{Proof of claim \ref{claim:condensed_exp}.} 
We note that the condensed expressions of $o_i$ and $z_i$ for $1 \leq i \leq n$ cover all possible condensed expressions for a sequence with length at most $n$. Therefore, for input sequence $s$, we can find a query sequence $q\in\mc{Q}$ such that $s$ and $q$ have the same condensed expression, and we would have $d_{\revise{\DTW}}(s, q) = 0$.
Since $s$ and $s'$ have different condensed expressions, we would have $d_{\revise{\DTW}}(s', q) \neq 0$. Thus, $q$ distinguishes $s$ and $s'$.

\leavevmode\newline
\noindent
Suppose $s$ and $s'$ cannot be distinguished by $\mc{Q}$. By Claim \ref{claim:condensed_exp}, we know that the condensed expression of $s$ and $s'$ are the same. Let the number of runs in $s$ and $s'$ be $k = \#\textsc{runs}(s) = \#\textsc{runs}(s')$.

\begin{claim}
\label{claim:runs_geq_3}
If $s$ and $s'$ cannot be distinguished by $\mc{Q}$, then $k\geq 3$.
\end{claim}

\noindent\textit{Proof of claim \ref{claim:runs_geq_3}}.
Consider the query sequence $z_1 = 0^n$ and $o_1 = 1^n$.
By querying $z_1$ and $o_1$, we can obtain the number of 1's and 0's in the input sequence. If $k \leq 2$, then $s$ (and $s'$) would contain at most a single $0$-run and a $1$-run. With queries $z_1$ and $o_1$ we can determine the length of the $0$-run and the $1$-run in $s$ and $s'$, and therefore distinguish them.

\begin{claim}
\label{claim:first_last_blocks}
If $s$ and $s'$ cannot be distinguished by $\mc{Q}$, then $\textsc{lor}(s, 1) = \textsc{lor}(s', 1)$ and $\textsc{lor}(s, k) = \textsc{lor}(s', k)$.
\end{claim}

\noindent\textit{Proof of claim \ref{claim:first_last_blocks}}. If $s$ and $s'$ start with $0$, we show that $d_{\revise{\DTW}}(s, o_{k-1}) = \textsc{lor}(s, 1) $.
Since $k\geq 3$, by Theorem \ref{theorem:obv_4}, 

\ifarxiv
\begin{align*}
d_{\revise{\DTW}}(s, o_{k-1}) = \min \left(\textsc{lor}(s, 1) +d_{\revise{\DTW}}(s[\textsc{lor}(s, 1) +1, \revise{\len}(s)], o_{k-1}), \right. \\
    \left. n+d_{\revise{\DTW}}(s, o_{k-1}[n+1, \revise{\len}(o_{k-1})])\right)   \nonumber
\end{align*}
\fi

Note that $d_{\revise{\DTW}}(s[\textsc{lor}(s, 1) +1, \revise{\len}(s)], o_{k-1}) = 0$, and $\textsc{lor}(s, 1)  \leq n \leq n+d_{\revise{\DTW}}(s, o_{k-1}[n+1, \revise{\len}(o_{k-1})])$, we know that $d_{\revise{\DTW}}(s, o_{k-1}) = \textsc{lor}(s, 1) $. Similarly $d_{\revise{\DTW}}(s', o_{k-1}) = \textsc{lor}(s', 1)$. 
Since $s$ and $s'$ cannot be distinguished by $\mc{Q}$, we have $\textsc{lor}(s, 1)  = d_{\revise{\DTW}}(s, o_{k-1}) = d_{\revise{\DTW}}(s', o_{k-1}) = \textsc{lor}(s', 1) $. Similarly, we would have $\textsc{lor}(s, k)  = d_{\revise{\DTW}}(s, z_{k-1}) = d_{\revise{\DTW}}(s', z_{k-1}) = \textsc{lor}(s', k) $.

By symmetry, if $s$ and $s'$ starts with $1$, we would have  $\textsc{lor}(s, 1)  = d_{\revise{\DTW}}(s, z_{k-1}) = d_{\revise{\DTW}}(s', z_{k-1}) = \textsc{lor}(s', 1) $ and  $\textsc{lor}(s, k)  = d_{\revise{\DTW}}(s, o_{k-1}) = d_{\revise{\DTW}}(s', o_{k-1}) = \textsc{lor}(s', k) $. This finishes the proof for Claim \ref{claim:first_last_blocks}.

\leavevmode\newline
\noindent
Next, we show that $s$ and $s'$ cannot be distinguished by any binary query $r$. Let the number of runs in $r$ be $l = \#\textsc{runs}(r)$.
Given $s$ and $r$, we can calculate $d_{\revise{\DTW}}(s, r)$ with Theorem \ref{theorem:obv_4} and Theorem \ref{theorem:dtw_mss}. Note that in Theorem \ref{theorem:obv_4}, we may remove the first/last blocks of $s$ (and $s'$) or $r$ to reduce to the case of Theorem \ref{theorem:dtw_mss}. By Claim \ref{claim:first_last_blocks} we have $\textsc{lor}(s, 1)  = \textsc{lor}(s', 1) $ and $\textsc{lor}(s, k)  = \textsc{lor}(s', k) $, while $\textsc{lor}(r, 1) $ and $\textsc{lor}(r, l) $ are only related to $r$ but not $s$ and $s'$. Therefore, 
the offsets while reducing to the case of Theorem \ref{theorem:dtw_mss} are the same. To prove that $d_{\revise{\DTW}}(s, r) = d_{\revise{\DTW}}(s', r)$, we only need to prove that for each possible reduction, the corresponding reduced MSS instances have the same sum of solutions (see Example \ref{example:mss_reduction} for illustration). 

\begin{example}
\label{example:mss_reduction}
To illustrate, take $s = 010110$, $s' = 011010$ and $r = 1001011$ as an example. In this case, we would have $\textsc{lor}(s, 1)  = \textsc{lor}(s', 1)  = 1$, $\textsc{lor}(s, k)  = \textsc{lor}(s', k)  = 1$, $\textsc{lor}(r, 1)  = 1$ and $\textsc{lor}(r, l)  = 2$.
According to Theorem \ref{theorem:obv_4}, we have
 $d_{\revise{\DTW}}(s, r) = \min (\textsc{lor}(s, 1)  + d_{\revise{\DTW}}(s[\textsc{lor}(s, 1) +1, \revise{\len}(s)], r), \textsc{lor}(r, 1)  + d_{\revise{\DTW}}(s, r[\textsc{lor}(r, 1) +1, \revise{\len}(r)])) = \min (1 + d_{\revise{\DTW}}(s[2, 6], r), 1 + d_{\revise{\DTW}}(s, \\r[2, 7]))$.
 Then $d_{\revise{\DTW}}(s[2, 6], r) = \min (\textsc{lor}(s, k)  + d_{\revise{\DTW}}(s[2, 6-\textsc{lor}(s, k) ], r), \textsc{lor}(r, l)  + d_{\revise{\DTW}}(s[2, 6], \\r[1, \revise{\len}(r)-\textsc{lor}(r, l) ])) = \min (1 + d_{\revise{\DTW}}(s[2, 5], r), 2 + d_{\revise{\DTW}}(s[2, 6], r[1, 5]))$.
 Also, $d_{\revise{\DTW}}(s, r[2, 7])) = \min (\textsc{lor}(s, k)  + d_{\revise{\DTW}}(s[1, \revise{\len}(s)-\textsc{lor}(s, k) ], r[2, 7]), \textsc{lor}(r, k)  + d_{\revise{\DTW}}(s, r[2, 7-\textsc{lor}(r, k) ] )) = \min (1 + d_{\revise{\DTW}}(s[1, 5], r[2, 7]), 2 + d_{\revise{\DTW}}(s, r[2, 5] ))$.
 
 Therefore, to show that $d_{\revise{\DTW}}(s, r) = d_{\revise{\DTW}}(s', r)$, we only need to prove that 
 $$\begin{cases}d_{\revise{\DTW}}(s[2, 5], r) = d_{\revise{\DTW}}(s'[2, 5], r); \\ d_{\revise{\DTW}}(s[2, 6], r[1, 5]) = d_{\revise{\DTW}}(s'[2, 6], r[1, 5]); \\ d_{\revise{\DTW}}(s[1, 5], r[2, 7]) = d_{\revise{\DTW}}(s'[1, 5], r[2, 7]);\\
 d_{\revise{\DTW}}(s, r[2, 5] ) = d_{\revise{\DTW}}(s', r[2, 5] ), \end{cases}$$
where each of the 4 cases corresponds to an MSS instance. \qed
\end{example}

Suppose after applying Theorem \ref{theorem:obv_4}, $s$ and $r$ are reduced to sub-sequences $s^*$ and $r^*$ (where $s^*$ and $r^*$ have the same beginning and ending characters), while $s'$ and $r$ are reduced to $s'^*$ and $r^*$. Now we calculate $d_{\revise{\DTW}}(s^*, r^*)$ and $d_{\revise{\DTW}}(s'^*, r^*)$ according to Theorem \ref{theorem:dtw_mss}. Suppose $s^*$ and $s'^*$ have $k^*$ runs and $r^*$ have $l^*$ runs.

Case 1. If $k^* = l^*$, then $d_{\revise{\DTW}}(s^*, r^*) = d_{\revise{\DTW}}(s'^*, r^*) = 0$.

Case 2. If $k^* < l^*$, by Theorem \ref{theorem:dtw_mss} the generated MSS instance only depends on $r^*$ and $k^*$. Thus, $d_{\revise{\DTW}}(s^*, r^*) = d_{\revise{\DTW}}(s'^*, r^*)$.

Case 3. If $k^* > l^*$, we have the MSS instances $\revise{\MSS}(s^*, (k^* - l^*)/2)$ and $\revise{\MSS}(s'^*, (k^* - l^*)/2)$. Note that, we can always find a query $q \in \mc{Q}$ which has $l^*$ runs and has the same starting and ending characters as $s^*$ and $s'^*$. Consider $d_{\revise{\DTW}}(s, q)$ and $d_{\revise{\DTW}}(s', q)$. Note that the first and last runs of $q$ are both of length $n$ and removing them would yield at least cost $n$, the only possible reduction would be $d_{\revise{\DTW}}(s^*, q)$ and $d_{\revise{\DTW}}(s'^*, q)$. Since $q$ cannot distinguish $s$ and $s'$, we have $d_{\revise{\DTW}}(s, q) = d_{\revise{\DTW}}(s', q)$, so $d_{\revise{\DTW}}(s^*, q) = d_{\revise{\DTW}}(s'^*, q)$, implying that $\revise{\MSS}(q^*, (k^* - l^*)/2)$ and $\revise{\MSS}(q'^*, (k^* - l^*)/2)$ have the same sum of solution. Therefore, $d_{\revise{\DTW}}(s^*, r^*) = d_{\revise{\DTW}}(s'^*, r^*)$.

Combining the 3 cases above, we always have $d_{\revise{\DTW}}(s^*, r^*) = d_{\revise{\DTW}}(s'^*, r^*)$, so $d_{\revise{\DTW}}(s, r) = d_{\revise{\DTW}}(s', r)$, implying that $s$ and $s'$ cannot be distinguished by $r$. This finishes the proof for Theorem \ref{theorem:oz_equivalence}.
\end{proof}

\subsection{Exact Recovery with Extra Character(s)}

\begin{theorem}[\revise{Lower Bound for DTW Exact Recovery}]
\label{thm:dtw_lower_bound}
    For a binary alphabet $\bits$, any algorithm to recover arbitrary input sequence $s \in \bits^\ell$ where $0 \leq \ell \leq n$ by querying DTW distance to a set of sequences of length $\mc{O}(n)$ \revise{\linelabel{constant-sized-alphabet}from a constant-sized extended alphabet $\Sigma$} would require a query complexity of $\Omega(n/\log n)$.
\end{theorem}

\noindent
\revise{\linelabel{upper-or-lower}Theorem \ref{thm:dtw_lower_bound} shows the lower bound of the query complexity for DTW exact recovery.
The proof of Theorem \ref{thm:dtw_lower_bound} is given by an information-theoretic lower bound, which refers back to the proof of Theorem \ref{thm:edit_lower_bound}.
\bigskip

With this lower bound, now we would like to show that if one is allowed to construct queries from a slightly larger alphabet beyond $\bits$, there exists a non-adaptive query strategy such that this lower bound is attainable as per the order of magnitude.
}

\bigskip

\subsubsection{With One Extra Character}

\begin{theorem}[Non-adaptive Strategy for DTW Exact Recovery with 1 Extra Character]
\label{theorem:dtw_1e}
For a binary alphabet $\bits$ and an input sequence $s := \bits^{\ell}$ where $0 \leq \ell \leq n$, there exists an algorithm to recover the input sequence $s$, given \revise{$n^2+n \in \mc{O}(n^2)$} query sequences $\revise{\mc{Q}}$ and the $d_{\revise{\DTW}}(s, q)$ to each query sequence $q \in \revise{\mc{Q}}$, where the query sequences are allowed to use only one extra character.
\end{theorem}

\stitle{Proof of Theorem \ref{theorem:dtw_1e}}
We give our proof by constructing $n^2+n$ query sequences of length $\leq n$ and presenting an algorithm to recover an input sequence $s$ from its DTW distance to these \revise{\linelabel{1-extra-char}$n^2+n$} query sequences.

Let $z_{i,k} = \begin{cases}0(10)^{m}(\frac{1}{2})^{k}, i = 2m+1; \\ (01)^{m}(\frac{1}{2})^{k}, i = 2m,\end{cases}$ and $o_{i,k} = \begin{cases}1(01)^{m}(\frac{1}{2})^{k}, i = 2m+1; \\ (10)^{m}(\frac{1}{2})^{k}, i = 2m,\end{cases}$ where $m$ is a non-negative integer, $1 \leq i \leq n$ and $0\leq k \leq n-i$. Let $\mc{Q} = \{o_{i,k}|1\leq i\leq n, 0\leq k \leq n-i\}\bigcup\{z_{i, k}|1\leq i\leq n, 0\leq k \leq n-i\} $. We have $|\mc{Q}| = 2\sum_{i = 1}^{n}i = n^2 + n$.

\revise{Without loss of generality}, we can assume that $s$ starts with a $0$ and has $t$ runs, where the $i$-th run of $s$ is $\textsc{lor}(s, i) $, $l = \sum_{i = 1}^{t}\textsc{lor}(s, i) $. Then $d_{\revise{\DTW}}(z_{t, 0}, s) = 0$.  Consider $d_{\revise{\DTW}}(z_{i, k}, s)$ where $1\leq i \leq t$. The first $i$ runs of $s$ have a total length of $\sum_{j=1}^{i}\textsc{lor}(s, j) $, and the last $t-i$ runs of $s$ have a total length of $\sum_{j=i+1}^{t}\textsc{lor}(s, j) $.

\begin{claim}
\label{claim: dtw_1e}
For $1\leq i\leq t$, $d_{\revise{\DTW}}(z_{i, k}, s) = \frac{k}{2} \iff k\geq \sum_{j=i+1}^{t}\textsc{lor}(s, j) $.
\end{claim}

\noindent\textit{Proof of Claim \ref{claim: dtw_1e}.} Since each $\frac{1}{2}$ in $z_{i, k}$ corresponds to at least $\frac{1}{2}$ cost, we have $d_{\revise{\DTW}}(z_{i, k}, s) \geq \frac{k}{2}$. Note that $s[1, \sum_{j=1}^{i}\textsc{lor}(s, j) ]$ and $z_{i, k}[1, i]$ can be perfectly matched. If $k\geq \sum_{j=i+1}^{t}\textsc{lor}(s, j) $, then $s[(\sum_{j=1}^{i}\textsc{lor}(s, j) )+1, l]$ and $z_{i, k}[i+1, i+k] = (\frac{1}{2})^{k}$ can be matched with exactly $\frac{k}{2}$ cost, so $d_{\revise{\DTW}}(z_{i, k}, s) = \frac{k}{2}$. Otherwise, if $k< \sum_{j=i+1}^{t}\textsc{lor}(s, j) $, we show that $d_{\revise{\DTW}}(z_{i, k}, s) > \frac{k}{2}$. In fact, if any of the $\frac{1}{2}$ in $z_{i, k}$ is matched to more than one character in $s$, we would already have $d_{\revise{\DTW}}(z_{i, k}, s) > \frac{k}{2}$. If all $\frac{1}{2}$'s in $z_{i, k}$ have degree 1, then $z_{i, k}[1..i]$ must be matched with $s[1, l-k]$. Since $l-k >  \sum_{j=1}^{i}\textsc{lor}(s, j) $, $z_{i, k}[1..i]$ and $s[1, l-k]$ cannot be perfectly matched, yielding a non-zero cost. This finishes the proof of Claim \ref{claim: dtw_1e}.

\leavevmode\newline
By Claim \ref{claim: dtw_1e}, we know that for $1\leq i \leq t$, $\sum_{j=i+1}^{t}\textsc{lor}(s, j)  = min_{d_{\revise{\DTW}}(z_{i, k}, s) = \frac{k}{2}}k$. In this way, we can recover the length of each run in $s$, and therefore recover $s$. A similar analysis can be performed for the cases where $s$ starts with a single $1$. \hfill$\blacksquare$

\bigskip

\subsubsection{With Two Extra Characters}

\begin{theorem}[Non-adaptive Strategy for DTW Exact Recovery with 2 Extra Characters]
\label{theorem:dtw_o1}
For a binary alphabet $\bits$ and an input sequence $s := \bits^{\ell}$ where $0 \leq \ell \leq n$, there exists an algorithm to recover the input sequence $s$, given \revise{$n + 2 \in \mc{O}(n)$} query sequences $\revise{\mc{Q}}$ of length $\leq n$ and the $d_{\revise{\DTW}}(s, q)$ to each query sequence $q \in \revise{\mc{Q}}$, where the query sequences are allowed to use only $\mc{O}(1)$ extra characters.
\end{theorem}

\begin{figure*}[t]
    \centering
    \includegraphics[width=\textwidth]{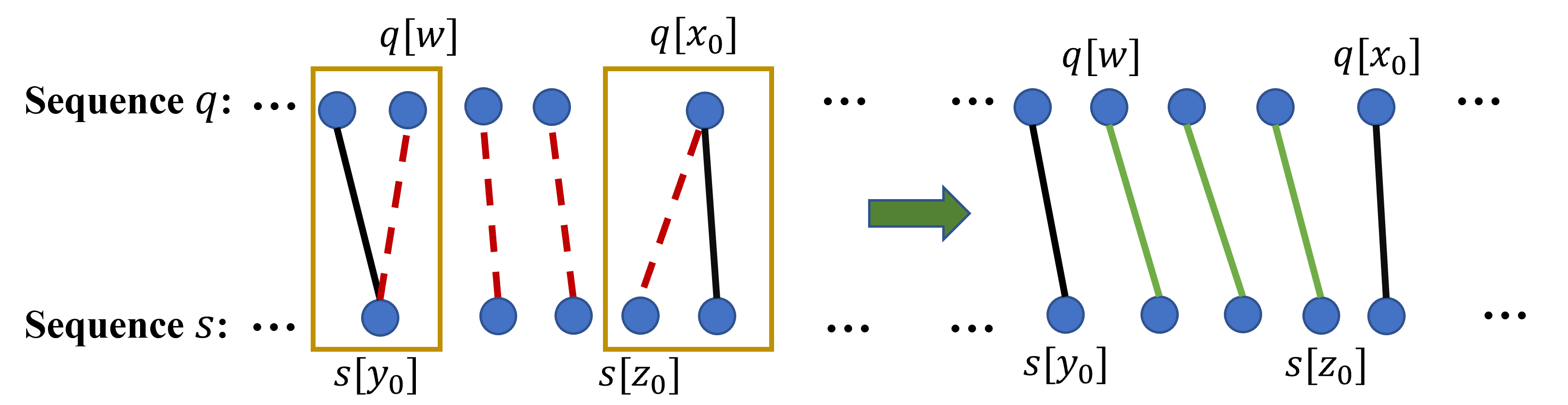}
    \caption{\revise{Optimal re-matching (\emph{hybrid stitching}) when both query and input sequences have a vertex with degree greater than $1$ (c.f. Claim \ref{claim:2_o1}). }}
    \label{fig:rematching_o1}
\end{figure*}

\stitle{Proof of Theorem \ref{theorem:dtw_o1}.}
We give our proof by constructing \revise{$n+2 \in \mc{O}(n)$} query sequences of length $\leq n$ and presenting an algorithm to recover an input sequence $s$ from its DTW distance to these $n+2$ query sequences.

Note that for any sequence $s$, we have $d_{\revise{\DTW}}(s, 0) = 0 \iff s \text{ consists of only 0's}$ and $d_{\revise{\DTW}}(s, 1) = 0 \iff s \text{ consists of only 1's}$. We can also derive that $d_{\revise{\DTW}}(0^m, 1) = m$ and $d_{\revise{\DTW}}(1^m,0) = m$. Thus, any input sequence consisting of only $0$s or $1$s can be exactly recovered
by the two query sequences $0$ and $1$. 
For simplicity, \textbf{we assume in the rest of the proof that the input sequence $s$ contains both $0$ and $1$} and let $s = s[1]s[2] \ldots s[\ell]$.

\medskip
\noindent\textit{Query Sequences Construction.}
Let $a, b$ be two fractional characters that satisfy $0 < b-a < a < b < \frac{1}{2}$ and the denominators of $a, b$ are co-prime.
We will use $a, b$ as the extra characters to construct the query sequences.
In particular, the rest of the query sequences (other than the $0$ query and the $1$ query) consist of queries $\mc{Q}$ is in the form of $a^{n-i}b^{i}$, where $i = 1, \dots, n$.
It is not hard to see this set of queries are monotonic sequences, for which we show the following property holds in the distance query to DTW.
\revise{\linelabel{wlog-1}We will use $a = \frac{1}{3}$ and $b = \frac{2}{5}$ as a running example for better explanation when necessary, but the proof works for all $a, b$ satisfying the condition.}

\begin{lemma}
\label{lemma:single-direction_o1}
Given a monotonic sequence $q$ of length $n$ where 
\begin{equation}
    \min_{i\in [n]}\max\{|q[i]-0|, |q[i]-1|\} > \max_{i,j\in [n]}|q[i]-q[j]|,\label{equ:single_direction_o1}
\end{equation}
for any input sequence $s$ with length $\ell \leq n$, given a DTW matching $M$ for $(q,s)$, we have $\revise{\deg}(q[i]) = 1$ for all elements $q[i]$ in $q$. 
\end{lemma}

\begin{proof}
Suppose $\exists i\in[\ell]$ such that $\revise{\deg}(q[i]) > 1$.
We first prove the following two claims.

\begin{claim}
\label{claim:1_o1}
\revise{
For any edge $e$ in $M$, the two vertices corresponding to $e$ in $M$ cannot have degree $>1$ at the same time.}
\end{claim}

\noindent\textit{Proof of Claim \ref{claim:1_o1}}. This is trivial, since otherwise by deleting $e$ we would obtain a better matching.

\begin{claim}
\label{claim:2_o1}
There does not exist $i\in[n]$ and $j\in[\ell]$ such that $\revise{\deg}(q[i])>1$ and $\revise{\deg}(s[j])>1$.
\end{claim}

\noindent \textit{Proof of Claim \ref{claim:2_o1}}. We prove this by contradiction. Suppose $\exists(i,j)$ where $i\in[n]$ and $j\in[\ell]$ such that  $\revise{\deg}(q[i])>1$ and $\revise{\deg}(s[j])>1$. 
\revise{\linelabel{index_set}Consider the following index sets $\mb{X}, \mb{Y}, \mb{Z}$.}
Let $\mb{X} = \{x \in [n] \mid \revise{\deg}(q[x]) > 1\}$, $\mb{Y} = \{y \in [\ell] \mid \revise{\deg}(s[y]) > 1 \}$ and $\mb{Z} = \{z \in [\ell] \mid \exists x \in \mb{X} \text{ such that edge }(q[x], s[z])\in M\}$. According to Claim \ref{claim:1_o1}, we know that $\mb{Y}\bigcap \mb{Z} = \emptyset$. Let $d = \min_{y\in \mb{Y}, z\in \mb{Z}}|y-z|$, we would have $d>0$. Suppose we have $x_0\in \mb{X}, y_0\in \mb{Y}, z_0\in \mb{Z}$ such that edge $(q[x_0], s[z_0])\in M$ and $|y_0-z_0| = d$.
And this leaves two cases to discuss.

1) if $y_0 < z_0$, then we know that (i) $\forall k\in (y_0,z_0)$, $\revise{\deg}(s[k]) = 1$; otherwise we would have $k\in \mb{Y}$ and $|k-z_0|<d$, causing a contradiction. (ii) $\forall k_1, k_2$ within range $(y_0, z_0)$, $k_1\neq k_2$, we would have $s[k_1]$ and $s[k_2]$ matched to different vertices in $q$; otherwise, suppose edges $(q[t], s[k_1])\in M$ and $(q[t], s[k_2])\in M$. We would have $t\in \mb{X}$, $k_1\in \mb{Z}$ and $|y_0-k_1| < d$, causing a contradiction. 

Now, since $d$ is minimal, we can suppose that $s[y_0]$ is matched to $$\{q[w-\revise{\deg}(s[y_0])+1], q[w-\revise{\deg}(s[y_0])+2], \ldots, q[w]\},$$ and $q[x_0]$ is matched to $$\{s[z_0], s[z_0+1], \ldots, s[z_0+\revise{\deg}(q[x_0])-1]\}.$$ 
Since $y_0 < z_0$, by the monotonic property of the matching, we know that $w < x_0$. 
With (i) and (ii), we know that for $w < l < x_0$ and $y_0 < r < z_0$, the vertices $q[l]$'s and $s[r]$'s are perfectly matched one-to-one. Fig \ref{fig:rematching_o1} is an illustration of such an example.

We now claim, by re-matching edges between vertices $q[w], q[w+1], \ldots, q[x_0]$ and $s[y_0], s[y_0+1], \ldots, s[z_0]$, we can construct another matching $M'$ which is better than $M$, contradicting that $M$ is a DTW matching. We remove the $d+1$ edges $E = \{(q[w], s[y_0]), (q[w+1], s[y_0+1]), \ldots, (q[x_0], s[z_0])\}$ from $M$ and add $d$ new edges $E'=\{(q[w], s[y_0+1]), (q[w+1], s[y_0+2]), \ldots , (q[x_0-1], s[z_0])\}$ to obtain a new matching $M'$. Since $\revise{\deg}(s[y_0])>1$ and $\revise{\deg}(q[x_0])>1$ in $M$, $M'$ would still be a valid matching. Computing the sum of two sets of edges $E$ and $E'$, respectively, would yield 
\ifarxiv
the following.
\fi

\ifarxiv
\begin{align*}
    \Cost(E)  &   = |s[y_0]-q[w]|+\sum_{i = 1}^{d}|s[y_0+i]-q[w+i]| \label{eqn_dbl_1}\\
    & > |q[w]-q[x_0]|+ \sum_{i = 1}^{d}|s[y_0+i]-q[w+i]| \hfill \tag{Equation.~\ref{equ:single_direction_o1}} \\
    & = \sum_{i = 1}^{d}|q[w+i-1]-q[w+i]|+ \sum_{i = 1}^{d}|s[y_0+i]-q[w+i]|  \hfill \tag{Monotonicity of $q$}  \\
    & =  \sum_{i = 1}^{d}(|q[w+i-1]-q[w+i]|+|s[y_0+i]-q[w+i]|) \nonumber \\
    & \geq \sum_{i = 1}^{d}|q[w+i-1]-s[y_0+i]| \hfill \tag{Triangle Inequality}\\
    & = \revise{\Cost}(E'). \nonumber
\end{align*}
\fi

So $M'$ would be a better matching than $M$, causing a contradiction. 

2) if $y_0 > z_0$, this case is symmetric to 1) and we can use a similar method to complete the proof by contradiction. We give a detailed proof in the appendix. 

Combining the two cases finishes the proof for Claim \ref{claim:2_o1}.

\leavevmode\newline
\noindent
Suppose $\exists i\in[\ell]$ such that $\revise{\deg}(q[i]) > 1$. With Claim 2, we know that $\forall j\in[\ell]$, $\revise{\deg}(s[j]) = 1$. Thus, we would have $\sum_{j=1}^{\ell}\revise{\deg}(s[j]) = \sum_{i=1}^{n}\revise{\deg}(q[i]) > n \geq \ell = \sum_{j=1}^{\ell}\revise{\deg}(s[j]),$
which causes a contradiction and finishes the proof of Lemma \ref{lemma:single-direction_o1}.
\end{proof}

\begin{figure*}[t]
    \centering
    \includegraphics[width=\textwidth]{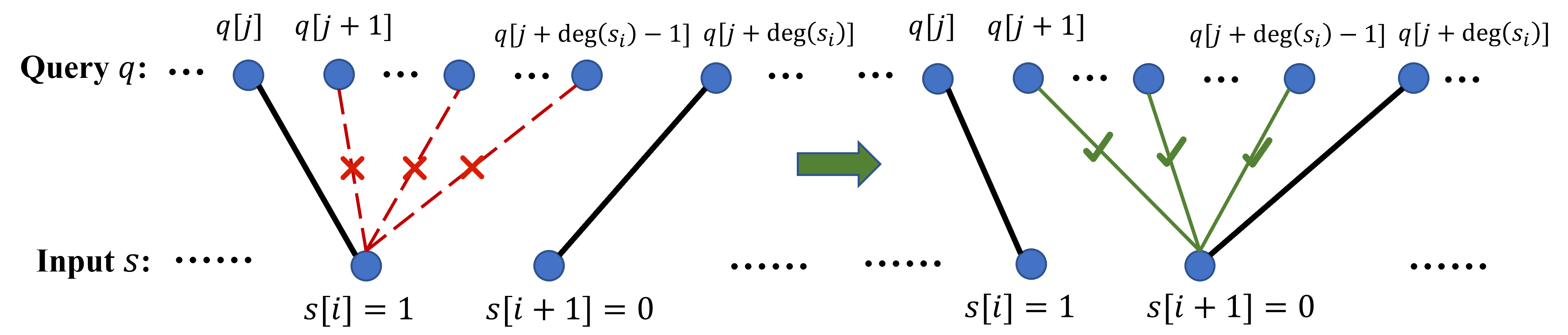}
    \caption{\revise{Obtaining a lower cost matching by shifting matched edges (c.f. Lemma \ref{lemma:3_o1}).}}
    \label{fig:shift_o1}
\end{figure*}

\noindent
Furthermore, we have the following lemma for the DTW matching for our query sequences.

\begin{lemma}
\label{lemma:3_o1}
For any given input sequence $s$ and query $q \in \mc{Q}$, the DTW matching $M$ for $(q, s)$ has $\revise{\deg}(s[i]) = 1$ in $M$ if $s[i]=1$. 
\end{lemma}
\begin{proof}
We give proof by contradiction. Given an optimal DTW matching $M$ for $(q,s)$, suppose $\exists 1\leq i\leq \ell$ such that $s[i] = 1$ and $\revise{\deg}(s[i]) > 1$. Suppose $s[i]$ is matched to $q[j], q[j+1], \ldots, q[j+\revise{\deg}(s[i])-1]$. 

First, we show that we can ``swap'' $s[i]$ with its neighboring element while maintaining the optimality of the matching. If one of the neighboring elements of $s[i]$ is $1$, w.l.o.g, suppose $s[i+1] = 1$, then we can construct an alternate optimal matching $M^{*}$ where $\revise{\deg}(s[i]) = 1$ and $\revise{\deg}(s[i+1])>1$. According to Lemma \ref{lemma:single-direction_o1}, $s[i+1]$ cannot be matched with any of $q[j], q[j+1], \ldots, q[j+\revise{\deg}(s[i])-1]$ in $M$, otherwise there would exist $j+1\leq k\leq j+\revise{\deg}(s[i])-1$ such that $\revise{\deg}(q[k]) = 2$. Thus, by matching $q[j+1], \ldots, q[j+\revise{\deg}(s[i])-1]$ to $s[i+1]$ instead of $s[i]$, we would obtain a new optimal matching $M^{*}$ where $\revise{\deg}(s[i]) = 1$ and $\revise{\deg}(s[i+1])>1$. 

As there exists at least one $0$ in $s$, we know that there exists an optimal DTW matching $M_0^{*}$ for $(q, s)$ where $\exists s[i]$ such that $s[i] = 1$, $\revise{\deg}(s[i]) > 1$ and one of the neighboring element of $s[i]$ is $0$. \revise{Without loss of generality}, suppose $s[i+1] = 0$. Similarly, according to Lemma \ref{lemma:single-direction_o1}, $s[i+1]$ cannot be matched with any of $q[j], q[j+1], \ldots, q[j+\revise{\deg}(s[i])-1]$ in $M_0$. Here we construct a new matching $M_0'$ by matching $q[j+1], \ldots , q[j+\revise{\deg}(s[i])-1]$ to $s[i+1]$ instead of $s[i]$. Fig \ref{fig:shift_o1} illustrates an example of such a construction. Considering the total cost of differing edges in both matchings, we have 
$\sum_{k=j+1}^{j+\revise{\deg}(s[i])-1}|s[i]-q[k]|
   > \sum_{k=j+1}^{j+\revise{\deg}(s[i])-1}\frac{1}{2}
   > \sum_{k=j+1}^{j+\revise{\deg}(s[i])-1}|s[i+1]-q[k]|.$
Thus $M_0^*$ would be a better matching than $M_0$, causing a  contradiction and thus finishing the proof. 
\end{proof}

\noindent
\textbf{Notation clarification.}
For the rest of the proof, we will use $q^{(i)}$ to denote the $i$-th query in the query set $\mc{Q}$ and $q^{(i)}[j]$ the $j$-th character in $q^{(i)}$.

\begin{lemma}
\label{lemma:4_o1}
For any input sequence $s$, there exists \revise{\linelabel{isomorphic_set-0}a set of isomorphic matchings} $\mc{M}^*$, where $\revise{M_i^*(q^{(i)}, s)} \in \mc{M}^*$ is optimal for query $q^{(i)} \in \mc{Q}$.
\end{lemma}
\begin{proof}
According to previous assumptions, we know that the input sequence $s$ contains at least one $0$. Suppose $s[u]$ is the first $0$ in $s$. 
\linelabel{star_unstar_1} 
We construct the following matching $M_i^*$ for each $q^{(i)} \in \mc{Q}$:

1) For $1\leq j < u$, $q^{(i)}[j]$ is matched to $s[j]$ in $M_i^*$;

2) For $u\leq j \leq u+n-\ell$, $q^{(i)}[j]$ is matched to $s[u]$ in $M_i^*$;

3) For $ u+n-\ell < j \leq n$, $q^{(i)}[j]$ is matched to $s[j-(n-\ell)]$ in $M_{i}^*$.

The constructed $M_i^*$'s form \revise{\linelabel{isomorphic_set-1}a set of isomorphic matchings}, and we will show that each $M_i^*$ is an optimal matching between $s$ and $q^{(i)}$. To prove this, we first define the ``shifting'' operation.

\begin{definition}[Shifting Operation for Queries in $\mc{Q}$]
\label{def:shifting_operation}
Given a matching $M$ between input sequence $s$ of length $\ell$ and query sequence $q\in\mc{Q}$ of length $n$. Suppose $\exists 1\leq x < y \leq \ell$ s.t. $s[x] = s[y] = 0$, $\revise{\deg}(s[x]) > 1$, and $\forall x < j < y$, $\revise{\deg}(s[j]) = 1$. We now construct a new matching $M'$ based on $M$:

Suppose $q[z]$ is the last character matched to $s[x]$ and $q[w]$ is the first character matched to $s[y]$, we know that $z-x = w-y$ (cf. lemma \ref{lemma:single-direction_o1}). For $x \leq j < y$, we remove the edge $(q[j-x+z], s[j])$ from $M$ and add the edge $(q[j-x+z], s[j+1])$. As $\revise{\deg}(s[x]) > 1$. This will give us a valid matching.
We call this process a \emph{shifting operation}.
\end{definition}

\begin{figure*}[t]
    \centering
    \includegraphics[width=\textwidth]{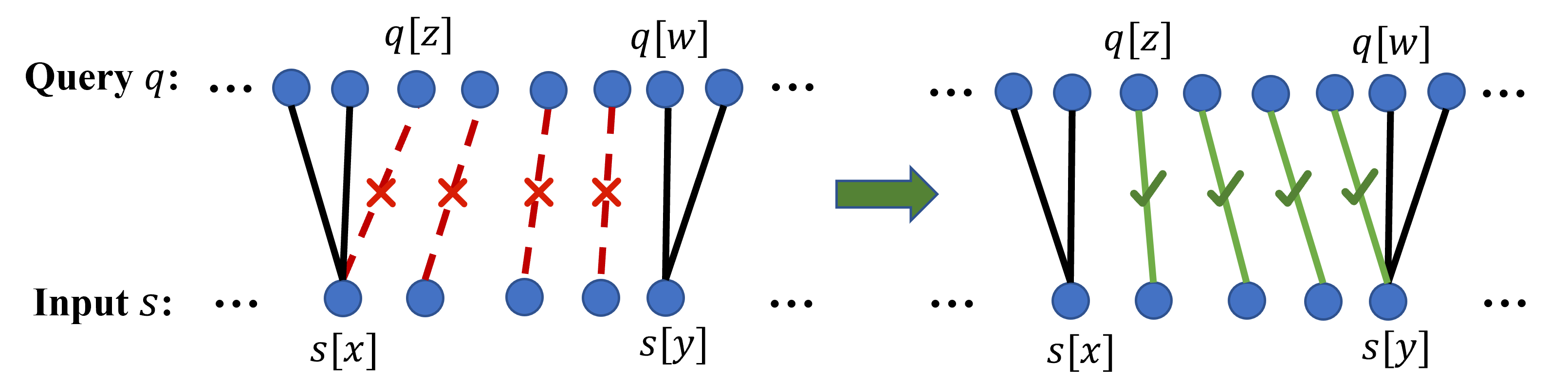}
    \caption{An illustration of the shifting operation (c.f. Definition \ref{def:shifting_operation}).}
    \label{fig:shifting_operation}
\end{figure*}

An illustration of the shifting operation is shown in Fig \ref{fig:shifting_operation}.
The shifting operation reduces $\revise{\deg}(s[x])$ by 1 and increases $\revise{\deg}(s[y])$ by 1, while preserving the degree of all other vertices in $s$. 
Now we give the following claims for shifting operations.

\begin{claim}
\label{claim:3_o1}
A shifting operation does not reduce the total cost of the matching.
\end{claim}

\noindent\textit{Proof of claim.}
As one can observe, the shifting operation will not increase the total number of edges -- the number of removed edges is equal to the number of newly added edges.
Then we only need to consider the cost of those changed edges.
\revise{\linelabel{wlog-2}Recall that our monotonic query sequences are in the form of $a^{n-k}b^{k}$ for $k = 1, \dots, n$.}
To calculate the change of cost in the shifting operation, we have two cases to analyze.

\noindent
Case 1. All characters between $q[z]$ and $q[w]$ (including $q[z]$ and $q[w]$) in the query sequence are the same, either $a$ or $b$.
In this case, the total cost does not change after the shifting operation.
This is because, $\forall s[j] \in s$ s.t. $x < j < y$, the edge changes from $(q[j-x+z], s[j])$ to $(q[j-x+z-1], s[j])$ and the cost $\revise{\Cost}(q[j-x+z], s[j]) = \revise{\Cost}(q[j-x+z-1], s[j])$ since $q[j-x+z] = q[j-x+z-1]$.
Notice in the matching before shifting, we have the edge $(q[z], s[x])$ while in the matching after shifting this edge is removed but the edge $(q[w], s[y])$ is added.
These two edges have equal cost $\revise{\Cost}(q[z], s[x]) = \revise{\Cost}(q[w], s[y])$ because $s_x = s_y = 0$.

\noindent
Case 2. The characters between $q[z]$ and $q[w]$ (including $q[z]$ and $q[w]$) contain both $a$ and $b$.
\revise{Without loss of generality}, we can assume there exists \revise{\linelabel{qi_i}index $i$,} s.t. for $j < i, q[j] = a$ while for $j \geq i, q[j] = b$.
Applying a similar analysis as we did in case 1, the cost of edges containing characters $q[j]$ such that $z < j < i - 1$ \revise{\linelabel{and_or}or} $i< j < w$ remains the same after the shifting operation.
Suppose $s[t]$ gets matched to $q[i]$ before the shifting operation.
We only need to analyze the cost of the (removed and added) edges corresponding to characters $s[x], s[t]$ and $s[y]$.
Before the shifting operation, these three characters get matched in edges $(q[z], s[x]), (q[i], s[t])$, respectively, while in the matching after shifting, they are involved in edges $(q[i-1], s[t]), (q[w-1], s[y])$.
We can compute the total cost of these three edges before shifting $\revise{\Cost_{\text{before}}} = |a - 0| + |b - s[t]|$ and the total cost after shifting $\revise{\Cost_{\text{after}}} = |a - s[t]| + |b - 0|$.
If $s[t] = 0$, then $\revise{\Cost_{\text{before}}} = a+b = \revise{\Cost_{\text{after}}}$; otherwise if $s[t] = 1$, then $\revise{\Cost_{\text{before}}} = a + 1 - b$ and $\revise{\Cost_{\text{after}}} = 1 - a  + b$.
Since $0 < a < b < 1$, $\revise{\Cost_{\text{before}}} < \revise{\Cost_{\text{after}}}$ when $s[t] = 1$.
Therefore in this case, $\revise{\Cost_{\text{before}}} \leq \revise{\Cost_{\text{after}}}$.

Combining both cases, the total cost of the matching before the shifting operation could be only less than or equal to the cost after shifting, which proves the claim.

\begin{claim}
\label{claim:4_o1}
Given input sequence $s$, query $q^{(i)}\in \mc{Q}$ and any matching $M_i$ between $s$ and $q^{(i)}$. 
If $M_i$ satisfies the properties that (i) $\forall 1\leq j \leq \ell$, $\revise{\deg}(s[j]) > 1\Rightarrow s_j = 0$, (ii) $\forall 1 \leq k \leq n$, $\revise{\deg}(q^{(i)}[k])=1$, then we can obtain $M_i$ by applying a series of shifting operations to $M_i^*$. 
\end{claim}

\noindent\textit{Proof of Claim.} 
If the input sequence $s$ contains only a single 0, then this claim is trivial since any matching $M_i = M_i^*$.
For cases that the input sequence $s$ contains more than one 0, \revise{without loss of generality}, we can assume $s$ has $k$ 0's and in the matching $M_i$, for each 0 in $s$ (denoted by $s_{0_m}$, $m \in [k]$), the degree $\revise{\deg}(s_{0_m}) = t_m \geq 1$.
Note that, as we defined, the shifting operation can be performed between $s[x]$ and $s[y]$, if $s[x] = s[y] = 0$, $\revise{\deg}(s[x]) > 1$, and $\forall x < t < y$, $\revise{\deg}(s[t]) = 1$.
This condition obviously holds for the matching $M_i^*$ if $s[x]$ and $s[y]$ are the nearest neighboring 0's in the input sequence $s$, because all characters between $s[x]$ and $s[y]$ are 1's and in $M_i^*$ all characters $s[j]$ s.t. $s[j] = 1$ we have $\revise{\deg}(s[j]) = 1$ (indicated by Lemma \ref{lemma:3_o1}).
{Property (ii) indicates that both $M_i^*$ and $M_i$ have the same number of edges $n$, and property (i) indicates $\forall 1 \leq j \leq \ell, s[j] = 1 \Rightarrow  \revise{\deg}(s[j])=1$.} 
For $M_i^*$ and $M_i$, we have $\sum_{s[j]=0} \revise{\deg}(s[j]) = n - (\# 1$s in $s) = \sum_{m=1}^k t_m$. 
For matching $M_i^*$, the degree of all 0's is 1 except for the first 0 and therefore the degree of the first 0 is $\sum_{m=1}^k t_m - (k - 1)$.
Therefore, we can perform the shifting operation $\sum_{m=1}^k t_m - (k - 1) - t_1$ times to move $\sum_{m=1}^k t_m - (k - 1) - t_1$ edges from the first 0 to the second 0.
Similarly, we continue doing shifting operations to move $\sum_{m=j}^{k} t_m - (k - 1) - t_j$ edges from the $j$-th 0 to the $(j+1)$-th 0.
We can hence obtain $M_i$ after all shifting operations are finished and this shows the correctness of this claim.

\leavevmode\newline
\noindent
Suppose $M_{i0}$ is an optimal DTW matching between $s$ and $q^{(i)}$. 
By Lemma \ref{lemma:single-direction_o1}, in DTW matching $M_{i0}$, $\forall 1 \leq k \leq n, \revise{\deg}(q^{(i)}[k]) = 1$.
By Lemma \ref{lemma:3_o1} we know that $M_{i0}$ has $\revise{\deg}(s[i]) = 1$ in $M_{i0}$ if $s[i] = 1$, so $\forall 1\leq j \leq \ell$, $\revise{\deg}(s[j]) > 1\Rightarrow s_j = 0$ in $M_{i0}$. By Claim \ref{claim:4_o1} we know that we can obtain $M_{i0}$ by applying a series of shifting operations to $M_i^*$, and according to Claim \ref{claim:3_o1} we would have $\revise{\Cost}(M_i^*)\leq \revise{\Cost}(M_{i0})$. Thus, $M_i^*$ is an optimal matching between $q^{(i)}$ and $s$. 
\end{proof}

\begin{proposition}
\label{prop:1}
Let $x_j \in\{0, 1\}$ be the value of the character matched to $q^{(i)}[j]$ in all isomorphic DTW matchings $M_i^*$, where $j \in [n]$.
We denote the sequence $x \coloneqq x[1]\dots x[n]$, where $x[j] = x_j$ for $j \in [n]$.
\revise{\linelabel{amplify_leftmost_0}The sequence $x$ can be obtained by amplifying the leftmost 0 in s.}
\end{proposition}

\begin{proof}
We see the proposition is naturally true based on the construction of $M_i^*$ in the proof of Lemma~\ref{lemma:4_o1}.
\end{proof}

\noindent\textit{Algorithm to recover DTW matching $M_i^*$.}
\linelabel{star_unstar_2} We now give the algorithm to recover the \emph{isomorphic} DTW matchings $M_i^*$ with the query set $\mc{Q}$ (Algorithm \ref{alg:dtw_on_recover}: line 8-15).
The query result $d_i$ of $q^{(i)} = a^{n-i}b^i$ would be $d_i = \sum_{j=1}^{n-i} |x[j]-a|+\sum_{j=n-i+1}^{n}|x[j]-b|$.
Recall that $a = \frac{1}{3}$ and $b = \frac{2}{5}$.
Consider $q^{1} = a^{n-1} b$, where $d_1 = \sum_{j=1}^{n-1} |x[j]- 1/3|+|x[n]-2/5|$.
By computing $(d_1 * 15) \mod 5$, we can know whether $x[n]$ is 0 or 1.
For $i>1$, we have $d_i - d_{i-1} = (\sum_{j=1}^{n-i} |x[j]-a|+\sum_{j=n-i+1}^{n}|x[j]-b|) - (\sum_{j=1}^{n-i+1} |x[j]-a|+\sum_{j=n-i}^{n}|x[j]-b|) = |x[n-i+1]-b|-|x[n-i+1]-a|$. By computing $((d_i-d_{i-1}) * 15) \mod 5$, we can know whether $x[n-i+1]$ is 0 or 1.
Then we can recover all $x[j]$'s using this procedure.

\begin{algorithm2e*}[t]
\LinesNumbered
\caption{Exact Recovery Algorithm via Queries to DTW Distance Oracle ($\mc{O}(1)$ Extra Chars)}
\label{alg:dtw_on_recover}
\SetKwInput{KwInput}{Input}                %
\SetKwInput{KwOutput}{Output}              %
\DontPrintSemicolon
  \KwInput{Non-adaptive query sequences $\revise{\mc{Q}} = \{ q^{(1)}, q^{(2)}, \dots, q^{(n+2)} \}$, where $q^{(n+1)} = 0$, $q^{(n+2)} = 1$ and the rest of the queries follows our construction; \newline
  The DTW distance \emph{query results} $\revise{\mc{R}} = \{ d_1, d_2, \dots, d_{n+2} \}$ aligned from each query sequence in $\mc{Q}$ to the input sequence to be recovered.\;}
\KwOutput{The sequence $s$ to be recovered.}
\SetKwFunction{FMain}{{\sc RecoveryDTW}}
  \SetKwProg{Fn}{Function}{:}{}
    \Fn{\FMain{$\revise{\mc{Q}}, \revise{\mc{R}}$}}{
    \uIf{$d_{n+1}$ = 0}{
        \Return s $\coloneqq$ $0^{d_{n+2}}$
    }
    \uIf{$d_{n+2}$ = 0}{
        \Return s $\coloneqq$ $1^{d_{n+1}}$
    }
    positions$\coloneqq$ [] \\
    coef\_1 $\coloneqq$ 0 \\
    \For({{\color{blue}\Comment{Corresponding queries $q^{(i)} = a^{n-i}b^i$}}}){$i \in [1, n]$}{
        $ \text{coef} \coloneqq d_i * 15 * 2 \mod 5$  \\
        \uIf{$(\text{\em coef} - \text{\em coef}\_1 + 5) \mod 5  = 2$}{
            positions.append(0)
        }
        \uElseIf{$(\text{\em coef} - \text{\em coef}\_1 + 5) \mod 5 = 3$ }{
            positions.append(1)
        }
        coef\_1 $\coloneqq$ coef
    }
    positions.reverse() \\
    sequence $\coloneqq$ [], $i \coloneqq 0$\\
    $n\_0$ $\coloneqq$ $d_{n+2}$, $n\_1$ $\coloneqq$ $d_{n+1}$\\
    \While{$\text{\em positions}[i] = 1$} {
        sequence.append($1$)\\
        $i$ += 1 \\
    }
    $i$ += $n - n\_0 - n\_1$ \\
    \While{$i<n$} {
        sequence.append(positions[$i$])\\
        $i$ += 1 \\
    }
    \Return s $\coloneqq$ sequence
    }
\end{algorithm2e*}

\bigskip
\noindent\textit{Algorithm to recover input sequence $s$.}  We now give an overall algorithm (as shown in Algorithm \ref{alg:dtw_on_recover}) that recovers $s$ using the matching recovery algorithm and claims. 
For the all 0 and all 1 input sequences, we can use $q^{(n+1)} = 0$, $q^{(n+2)} = 1$ to directly recover them (Algorithm  \ref{alg:dtw_on_recover}: line 2-5).
For the rest of the cases, we first recover the optimal isomorphic matching using the described algorithm (Algorithm  \ref{alg:dtw_on_recover}: line 6-15).
Let the recovered matching for $\mc{Q}$ be $m = x[1] \ldots x[n]$, ($x[i] \in \{0, 1\}$). 
Denote the position of the \revise{leftmost} 0 in $x$ to be $u$ ($1 \leq u \leq \ell$).
Then we know $s[1], \ldots, s[u-1] = 1$ by Proposition~\ref{prop:1}. 
By using the sequence $q^{(n+2)} = 1$ to query $s$, we get the total number of 0's $(n_0)$ in $s$. Consider the substring $x[u, n]$, and delete the leading zeros in $x[u, n]$ until it has $n_0$ zeros. Suppose we obtain string $s_d$ after the deletion. We know that $s = 1^{u-1}s_d$  (Algorithm  \ref{alg:dtw_on_recover}: line 16-24).
\hfill$\blacksquare$

\stitle{Remark.}
If we are allowed to use $\mc{O}(n)$ extra characters in our queries, we have non-adaptive solutions with $\mc{O}(1)$ query complexity for DTW distance.
This assumption is stronger than the problem setting (where only $\mc{O}(1)$ extra characters are considered) throughout the paper.
For details of this complementary result, see Appendix \ref{app:dtw_on_extra}.

\section{Recovery with Non-Adaptive \frechet Distance Oracle Queries}
\label{sec:frechet}

Consider two sequences $x$ and $y$ ($x \neq y$) defined on the binary alphabet $\bits$.
\revise{The query result from a \frechet distance oracle only gives very limited information, viz. 0 or 1 (which is more limited than the query from DTW oracle).
This 1-bit binary information restricts the power of sequence recovery with \frechet oracle.
Note that it is not possible to distinguish any sequences $x$ and $y$ under \frechet distance.
\linelabel{trivial}To see this and to see why the recovery problem is interesting for \frechet distance, we first define the concept of equivalent sequences \emph{under \frechet distance} and revisit the problem from the perspective of equivalent sequences.
}

\begin{definition}[Equivalent Sequences under \frechet Distance]
Given two sequences $x$ and $y$, we say $x$ and $y$ are equivalent if $y$ is obtained by taking any bit in $x$ and copying this bit contiguously any number of times.
For any pair of equivalent sequences, the \frechet distance between them is $0$.
\end{definition}

A simple example of equivalent sequences under \frechet distance is two sequences, $1$ and $11$.
$11$ can be seen as copying the bit $1$ in the first sequence and the \frechet distance between $1$ and $11$ is $0$.
In addition, these two sequences cannot be distinguished by \emph{any} query sequence.
This is because for the second sequence, the double $1$ characters can be matched to the same character in the query sequence as the single $1$ sequence.
This will not change the \frechet distance because the $l_\infty$ norm of the cost of matching edges is not changed.

\revise{From the perspective of equivalent sequences, for any two sequences $x$ and $y$, they are either in the same equivalence class (the \frechet distance is 0) or in different equivalence classes (the \frechet distance is 1).
Thus the \frechet distance between two sequences reflects whether or not they are equivalent.
Any equivalent sequences, therefore as suggested by its name, are not distinguishable, because all queries from the same equivalence class return 0 and all queries from different equivalence classes return 1.
Further, we can categorize all the equivalence classes under \frechet distance and then derive the lower bound of query complexity of recovering  \emph{non-equivalent sequences under \frechet distance}, which is shown in the following theorem.
}

\begin{theorem}[\revise{Lower Bound of Recovery from \frechet Distance}]
\label{thm:frechet_lower_bound}
For a binary alphabet $\bits$, any algorithm to recover an arbitrary input sequence $s \in \bits^i$ \revise{\linelabel{revise_thm:frechet_lower_bound}up to equivalence}, where $0 \leq i \leq n$, by querying its \frechet distance to a non-adaptive set of sequences requires a query complexity of $\Omega(n)$.

\end{theorem}
\begin{proof}
\revise{\linelabel{classify_frechet} We begin this proof of query complexity lower bound with a classification of all equivalence classes under the \frechet distance.
For each length $1 \leq i \leq n$, there exists two non-equivalent sequences under  \frechet distance, which are  $\underbrace{010101 \ldots}_{\text{of length} ~i}$ and $\underbrace{101010 \ldots}_{\text{of length} ~i}$, yielding $2n$ mutually non-equivalent sequences in total.} As the \frechet distance oracle returns 0 when the input sequence and the query sequence are equivalent and 1 otherwise, we would need at least $2n-1$ queries to exactly recover the input sequence. If the number of queries is less than $2n-1$, we can always select 2 sequences from the $2n$ mutually non-equivalent sequences which are not covered by the queries, and these two sequences cannot be distinguished by the query sequences. This yields an $\Omega(n)$ lower bound on the query complexity. 
\end{proof}

\revise{In the analysis of non-adaptive strategies for DTW distance, we have shown that, with extra characters, we can obtain stronger results in recovering the exact sequence.
However, using queries from the extended alphabet (no matter how many extra characters are allowed) does not help increase the power of recovery under \frechet distance, proved in the following theorem.}

\begin{theorem}[Extra Characters Are Not Helpful]
    \label{thm:frechet_extra_not_useful}
    Given two sequences $s$ and $s'$, if the \frechet distance $d_{F}(s, s') = 0$, then any query $q$ with extra characters cannot distinguish $s$ and $s'$.
\end{theorem}

\begin{proof}
\revise{\linelabel{not-triangle}Given sequences $s, s'$ (where $d_{F}(s, s') = 0$) and query $q$ with extra characters, our goal is to show $d_{F}(s, q) = d_{F}(s', q)$ for \emph{every possible $q$}. 
The technique of our proof is, for an optimal matching between $s$ and any query $q$, we can construct a matching between $s'$ and $q$ with the same cost, and vice versa.}
In this way, we know that $d_{F}(s', q) \leq d_{F}(s, q)$ and $d_{F}(s, q) \leq d_{F}(s', q)$, so $d_{F}(s, q) = d_{F}(s', q)$ and $q$ cannot distinguish $s$ and $s'$.  

Since $d_{F}(s, s') = 0$, $s$ and $s'$ have the same condensed expression. Suppose $s$ and $s'$ has $k$ runs. 
\revise{In the optimal matching between $q$ and $s$}, let \revise{$q_{s^{(i)}}$} denote the substring in $q$ which is matched to the $i$-th run of $s$ for every $i\in[k]$.  
We can always match all \revise{$q_{s^{(i)}}$}'s to the $i$-th run of $s'$ instead.
\revise{Note that the $i$-th runs of $s$ and $s'$ (denoted by $s^{(i)}$ and $s^{\prime(i)}$, resp.) are of the same character with maybe various length.
The \frechet distance between $q_{s^{(i)}}$ and $s^{(i)}$ only depends on the characters in $q_{s^{(i)}}$ and thus $d_F(q_{s^{(i)}}, s^{(i)}) = d_F(q_{s^{(i)}}, s^{\prime(i)})$.
Therefore we obtain a matching between $s'$ and $q$ with a cost of $d_{F}(s, q)$.
This matching between $q$ and $s'$ may be not optimal but is valid, and therefore we can conclude $d_{F}(s', q)\leq d_{F}(s, q)$.
Due to the symmetry of the statement, we can similarly obtain $d_{F}(s, q) \leq d_{F}(s', q)$.}
This finishes the proof of this theorem.
\end{proof}

\revise{Since extra characters are not helpful in recovering from \frechet distance queries, we conclude the analysis with a \emph{trivially interesting} approach to recover sequences up to equivalence.
The approach uses up to $2n-1$ queries, which exactly matches our query complexity lower bound, as shown in the following theorem.}

\begin{theorem}[Non-adaptive Strategy for \frechet Equivalence Class Recovery]
\label{theorem:frechet}
For a binary alphabet $\bits$ and two input sequences $s, s' \in \bits^i$ where $0 \leq i \leq n$ and $s$ and $s'$ are non-equivalent sequences under \frechet distance, there exists an algorithm to distinguish the input sequences $s$ and $s'$, given \revise{$2n-1 \in \mc{O}(n)$} query sequences $\revise{\mc{Q}}$ and the \frechet distance of $s$ and $s'$ to each query sequence $q \in \revise{\mc{Q}}$.
\end{theorem}
\begin{proof}
We first show that, for each length $0 \leq i \leq n$, there are only two non-equivalent sequences under \frechet distance, which are  $010101...$ and $101010...$ sequences, viz., we can identify two non-equivalent sequences by specifying the sequence length $i$ and the starting bit.
Therefore, for the maximum sequence length $n$, there are only $2n$ mutually non-equivalent sequences.

Given any two different sequences from this $2n$-sized collection of non-equivalent sequences under \frechet distance, we can use $\mc{O}(n)$ query sequences to distinguish them.
That is, we can utilize the exact set of $2n$ non-equivalent sequences as the query sequences.
If the query sequence $p$ is exactly the input sequence $q$, the \frechet distance between $p$ and $q$ is $d_F(p, q) = 0$.
If the query sequence $p$ is not equivalent to the input sequence $q$, then the \frechet distance between $p$ and $q$ is $d_F(p, q) = 1$ because it is impossible to skip over a bit without paying cost $1$.
\revise{Note that any one of the $2n$ queries can be skipped since we know the fact that there would be exactly one 0 among the $2n$ query results.}
Therefore, $2n-1 \in \mc{O}(n)$ query sequences suffice to distinguish any two sequences from the non-equivalent sequence set and this finishes the proof.
\end{proof}

This theorem shows that, if an input is in the collection of non-equivalent sequences under \frechet distance, we can use $\mc{O}(n)$ queries to exactly recover this sequence given the query results under the \frechet distance.

\medskip
\noindent\textbf{Remark: Extension to non-binary alphabets.}
Our results are presented for input sequences from binary alphabet $\bits$.
These results can be extended to any non-binary alphabet $\Sigma$ by encoding the non-binary alphabet in a binary domain.
This will increase the query complexity by a constant factor from $|\Sigma|$ (one-hot encoding) to $\log(|\Sigma|)$ (binary encoding).
This extension works for the results for all distance metrics shown in this paper.
However, we note that this extension may not be optimal if one considers a large alphabet (e.g., larger than $n$).
In fact, calculating some of the distances themselves on a general alphabet is under SETH \citep{abboud2016simulating,bringmann2015quadratic}, which is a much hard problem than on the binary case \citep{kuszmaul2021binary}.
Obtaining optimal results on the extension of the non-decomposable distance recovery problem leaves room for future research.

\section{Related Work}
\label{sec:related_work}

A distance embedding \citep{cormode2003sequence} embeds  sequences from the original distance metric space to other distance measures (usually $l_p$ norms), such that the distance measurements in the original space can be preserved up to a factor of $D$, namely \emph{the distortion rate}.
The sequence distance embedding problem is related to our problem in the sense that, in our problem, we intend to recover the input sequence from a list of query results that are in the $l_p$ space, which can be regarded as finding a special distance embedding.
Existing works on the sequence distance embedding problem mainly focus on constructing such an embedding which can have a close approximation (viz., \emph{low distortion rate}) and reduce the computational complexity (i.e., cost) on the new distance space.
\cite{andoni2003lower} shows a lower bound of $3/2$ on the distortion rate of embedding edit distance into $\ell_p$ norm spaces.
An improvement of $(\log n)^{\frac{1}{2} - o(1)}$ on this lower bound \citep{khot2005nonembeddability} has been further simplified and improved into $\Omega(\log n)$ by \cite{krauthgamer2009improved}.

Distance embeddings can be used to estimate the distance on the complex metric space because the evaluation and computations on the new (simpler metric) space can be significantly faster \citep{cormode2003sequence}.
Under the asymmetric query model (when estimating the edit distance between $x$ and $y$, the algorithm has unrestricted power accessing $x$ but limited power accessing $y$), \cite{andoni2010polylogarithmic} proposes a $(\log n)^{\mc{O}(1/\epsilon)}$ approximation algorithm that runs in $n^{1 + \epsilon}$ time.
\cite{charikar2018estimating} considers the alignment problem when estimating the edit distance (finding the sequence of edits between the estimated sequences) and presents an alignment with $(\log n)^{\mc{O}(1/\epsilon^2)}$ approximation in time $\Tilde{\mc{O}}(n^{1+\epsilon})$.
The sequence distance embedding problem has been investigated on other distance metrics as well, for example, the block edit distance \citep{cormode2003sequence} and the Ulam distance \citep{charikar2006embedding}.
Existing work also shows embeddings from edit distance to the Hamming space \citep{belazzougui2016edit,chakraborty2016streaming}.
However, to the best of our knowledge, there is no prior work considering the embedding problem of the \emph{DTW distance} and the \emph{exact} recovery problem based on distance oracle query results.

\section{Open Problems}
\label{sec:open_problems}

We initiate an exact recovery problem of sequences using queries to a non-decomposable distance oracle.
We show recovery algorithms for edit distance, DTW distance, and \frechet distance, as well as a general adaptive algorithm for a wide class of distance oracles.
We envision the following directions for future work.

First, for the edit distance, there is still a quadratic gap between the non-adaptive query complexity upper and lower bounds without extra characters. Closing this gap requires a deeper understanding about the properties of edit distance.

Second, for the DTW distance, it remains unclear whether 1 extra character suffices for an $\mc{O}(n)$ non-adaptive upper bound, or we can have an $\Omega(n^2)$ non-adaptive lower bound with $1$ extra character (our proof uses $2$ extra characters).

Furthermore, as the initial work on non-decomposable distance recovery problem, we consider a simpler setting where input sequences are drawn from binary alphabet $\bits$.
While our results can be naturally extended to a non-binary alphabet, as stated in the paper, with a compensation of increasing the query complexity up to a constant factor, we notice that for some distances (e.g., DTW), the calculation on the general alphabet is much harder than on the binary case.
This spawns the open question for follow-up work to consider: Would there exist a strategy specifically designed for the non-binary alphabet with lower query complexity (than using encoding extensions to our results on the binary alphabet)?

Lastly, it would be interesting to consider the exact sequence recovery problem using the properties of specific distance metrics.
For example, the Edit distance with Real
Penalty (ERP) distance \citep{chen2004marriage} which supports local time shifting in time series by the marriage of the $\ell_1$ norm and edit distance, would be of interest.
One can also consider other variants of our problem in terms of adaptive queries or the approximate recovery problem in the presence of noise.

\newpage

\section*{Acknowledgement}
All authors thank the anonymous reviewers of ITCS 2023 and IEEE Transactions of Information Theory for their detailed comments which helped to improve the paper during the revision process.
David P. Woodruff would like to thank support from ONR grant N00014-18-1-2562 and a Simons Investigator Award.
Hongyang Zhang would like to thank support from NSERC Discovery Grant RGPIN-2022-03215, DGECR-2022-00357. 

\bibliography{ref}

\begin{thebibliography}{48}
\providecommand{\natexlab}[1]{#1}
\providecommand{\url}[1]{\texttt{#1}}
\expandafter\ifx\csname urlstyle\endcsname\relax
  \providecommand{\doi}[1]{doi: #1}\else
  \providecommand{\doi}{doi: \begingroup \urlstyle{rm}\Url}\fi

\bibitem[Abboud et~al.(2015)Abboud, Backurs, and Williams]{abboud2015tight}
Amir Abboud, Arturs Backurs, and Virginia~Vassilevska Williams.
\newblock Tight hardness results for {LCS} and other sequence similarity
  measures.
\newblock In Venkatesan Guruswami, editor, \emph{{IEEE} 56th Annual Symposium
  on Foundations of Computer Science, {FOCS} 2015, Berkeley, CA, USA, 17-20
  October, 2015}, pages 59--78. {IEEE} Computer Society, 2015.
\newblock \doi{10.1109/FOCS.2015.14}.
\newblock URL \url{https://doi.org/10.1109/FOCS.2015.14}.

\bibitem[Abboud et~al.(2016)Abboud, Hansen, Williams, and
  Williams]{abboud2016simulating}
Amir Abboud, Thomas~Dueholm Hansen, Virginia~Vassilevska Williams, and Ryan
  Williams.
\newblock Simulating branching programs with edit distance and friends: or: a
  polylog shaved is a lower bound made.
\newblock In Daniel Wichs and Yishay Mansour, editors, \emph{Proceedings of the
  48th Annual {ACM} {SIGACT} Symposium on Theory of Computing, {STOC} 2016,
  Cambridge, MA, USA, June 18-21, 2016}, pages 375--388. {ACM}, 2016.
\newblock \doi{10.1145/2897518.2897653}.
\newblock URL \url{https://doi.org/10.1145/2897518.2897653}.

\bibitem[Afshani et~al.(2019)Afshani, Agrawal, Doerr, Doerr, Larsen, and
  Mehlhorn]{afshani2019query}
Peyman Afshani, Manindra Agrawal, Benjamin Doerr, Carola Doerr, Kasper~Green
  Larsen, and Kurt Mehlhorn.
\newblock The query complexity of a permutation-based variant of mastermind.
\newblock \emph{Discret. Appl. Math.}, 260:\penalty0 28--50, 2019.
\newblock \doi{10.1016/j.dam.2019.01.007}.
\newblock URL \url{https://doi.org/10.1016/j.dam.2019.01.007}.

\bibitem[Aldridge et~al.(2019)Aldridge, Johnson, and
  Scarlett]{aldridge2019group}
Matthew Aldridge, Oliver Johnson, and Jonathan Scarlett.
\newblock Group testing: An information theory perspective.
\newblock \emph{Found. Trends Commun. Inf. Theory}, 15\penalty0 (3-4):\penalty0
  196--392, 2019.
\newblock \doi{10.1561/0100000099}.
\newblock URL \url{https://doi.org/10.1561/0100000099}.

\bibitem[Amir et~al.(2018)Amir, Amit, Landau, and Sokol]{amir2018period}
Amihood Amir, Mika Amit, Gad~M. Landau, and Dina Sokol.
\newblock Period recovery of strings over the hamming and edit distances.
\newblock \emph{Theor. Comput. Sci.}, 710:\penalty0 2--18, 2018.
\newblock \doi{10.1016/j.tcs.2017.10.026}.
\newblock URL \url{https://doi.org/10.1016/j.tcs.2017.10.026}.

\bibitem[Andoni et~al.(2003)Andoni, Deza, Gupta, Indyk, and
  Raskhodnikova]{andoni2003lower}
Alexandr Andoni, Michel Deza, Anupam Gupta, Piotr Indyk, and Sofya
  Raskhodnikova.
\newblock Lower bounds for embedding edit distance into normed spaces.
\newblock In \emph{Proceedings of the Fourteenth Annual {ACM-SIAM} Symposium on
  Discrete Algorithms, January 12-14, 2003, Baltimore, Maryland, {USA}}, pages
  523--526. {ACM/SIAM}, 2003.
\newblock URL \url{http://dl.acm.org/citation.cfm?id=644108.644196}.

\bibitem[Andoni et~al.(2010)Andoni, Krauthgamer, and
  Onak]{andoni2010polylogarithmic}
Alexandr Andoni, Robert Krauthgamer, and Krzysztof Onak.
\newblock Polylogarithmic approximation for edit distance and the asymmetric
  query complexity.
\newblock In \emph{51th Annual {IEEE} Symposium on Foundations of Computer
  Science, {FOCS} 2010, October 23-26, 2010, Las Vegas, Nevada, {USA}}, pages
  377--386. {IEEE} Computer Society, 2010.
\newblock \doi{10.1109/FOCS.2010.43}.
\newblock URL \url{https://doi.org/10.1109/FOCS.2010.43}.

\bibitem[Aronov et~al.(2006)Aronov, Har{-}Peled, Knauer, Wang, and
  Wenk]{aronov2006frechet}
Boris Aronov, Sariel Har{-}Peled, Christian Knauer, Yusu Wang, and Carola Wenk.
\newblock Fr{\'{e}}chet distance for curves, revisited.
\newblock In Yossi Azar and Thomas Erlebach, editors, \emph{Algorithms - {ESA}
  2006, 14th Annual European Symposium, Zurich, Switzerland, September 11-13,
  2006, Proceedings}, volume 4168 of \emph{Lecture Notes in Computer Science},
  pages 52--63. Springer, 2006.
\newblock \doi{10.1007/11841036\_8}.
\newblock URL \url{https://doi.org/10.1007/11841036\_8}.

\bibitem[Belazzougui and Zhang(2016)]{belazzougui2016edit}
Djamal Belazzougui and Qin Zhang.
\newblock Edit distance: Sketching, streaming, and document exchange.
\newblock In Irit Dinur, editor, \emph{{IEEE} 57th Annual Symposium on
  Foundations of Computer Science, {FOCS} 2016, 9-11 October 2016, Hyatt
  Regency, New Brunswick, New Jersey, {USA}}, pages 51--60. {IEEE} Computer
  Society, 2016.
\newblock \doi{10.1109/FOCS.2016.15}.
\newblock URL \url{https://doi.org/10.1109/FOCS.2016.15}.

\bibitem[Braverman et~al.(2019)Braverman, Charikar, Kuszmaul, Woodruff, and
  Yang]{DBLP:conf/compgeom/BravermanCKWY19}
Vladimir Braverman, Moses Charikar, William Kuszmaul, David~P. Woodruff, and
  Lin~F. Yang.
\newblock The one-way communication complexity of dynamic time warping
  distance.
\newblock In Gill Barequet and Yusu Wang, editors, \emph{35th International
  Symposium on Computational Geometry, SoCG 2019, June 18-21, 2019, Portland,
  Oregon, {USA}}, volume 129 of \emph{LIPIcs}, pages 16:1--16:15. Schloss
  Dagstuhl - Leibniz-Zentrum f{\"{u}}r Informatik, 2019.
\newblock \doi{10.4230/LIPIcs.SoCG.2019.16}.
\newblock URL \url{https://doi.org/10.4230/LIPIcs.SoCG.2019.16}.

\bibitem[Bressan et~al.(2021)Bressan, Cesa{-}Bianchi, Lattanzi, and
  Paudice]{bressan2021exact}
Marco Bressan, Nicol{\`{o}} Cesa{-}Bianchi, Silvio Lattanzi, and Andrea
  Paudice.
\newblock Exact recovery of clusters in finite metric spaces using oracle
  queries.
\newblock In Mikhail Belkin and Samory Kpotufe, editors, \emph{Conference on
  Learning Theory, {COLT} 2021, 15-19 August 2021, Boulder, Colorado, {USA}},
  volume 134 of \emph{Proceedings of Machine Learning Research}, pages
  775--803. {PMLR}, 2021.
\newblock URL \url{http://proceedings.mlr.press/v134/bressan21a.html}.

\bibitem[Bringmann and K{\"{u}}nnemann(2015)]{bringmann2015quadratic}
Karl Bringmann and Marvin K{\"{u}}nnemann.
\newblock Quadratic conditional lower bounds for string problems and dynamic
  time warping.
\newblock In Venkatesan Guruswami, editor, \emph{{IEEE} 56th Annual Symposium
  on Foundations of Computer Science, {FOCS} 2015, Berkeley, CA, USA, 17-20
  October, 2015}, pages 79--97. {IEEE} Computer Society, 2015.
\newblock \doi{10.1109/FOCS.2015.15}.
\newblock URL \url{https://doi.org/10.1109/FOCS.2015.15}.

\bibitem[Bshouty(2009)]{bshouty2009optimal}
Nader~H. Bshouty.
\newblock Optimal algorithms for the coin weighing problem with a spring scale.
\newblock In \emph{{COLT} 2009 - The 22nd Conference on Learning Theory,
  Montreal, Quebec, Canada, June 18-21, 2009}, 2009.
\newblock URL
  \url{http://www.cs.mcgill.ca/\%7Ecolt2009/papers/004.pdf\#page=1}.

\bibitem[Buchin et~al.(2022)Buchin, Driemel, van Greevenbroek, Psarros, and
  Rohde]{DBLP:conf/waoa/BuchinDGPR22}
Maike Buchin, Anne Driemel, Koen van Greevenbroek, Ioannis Psarros, and Dennis
  Rohde.
\newblock Approximating length-restricted means under dynamic time warping.
\newblock In Parinya Chalermsook and Bundit Laekhanukit, editors,
  \emph{Approximation and Online Algorithms - 20th International Workshop,
  {WAOA} 2022, Potsdam, Germany, September 8-9, 2022, Proceedings}, volume
  13538 of \emph{Lecture Notes in Computer Science}, pages 225--253. Springer,
  2022.
\newblock \doi{10.1007/978-3-031-18367-6\_12}.
\newblock URL \url{https://doi.org/10.1007/978-3-031-18367-6\_12}.

\bibitem[Cai et~al.(2019)Cai, Xu, Yi, Huang, and Rajasekaran]{cai2019dtwnet}
Xingyu Cai, Tingyang Xu, Jinfeng Yi, Junzhou Huang, and Sanguthevar
  Rajasekaran.
\newblock {DTWNet}: a dynamic time warping network.
\newblock In Hanna~M. Wallach, Hugo Larochelle, Alina Beygelzimer, Florence
  d'Alch{\'{e}}{-}Buc, Emily~B. Fox, and Roman Garnett, editors, \emph{Advances
  in Neural Information Processing Systems 32: Annual Conference on Neural
  Information Processing Systems 2019, NeurIPS 2019, December 8-14, 2019,
  Vancouver, BC, Canada}, pages 11636--11646, 2019.
\newblock URL
  \url{https://proceedings.neurips.cc/paper/2019/hash/02f063c236c7eef66324b432b748d15d-Abstract.html}.

\bibitem[Cantor and Mills(1966)]{cantor1966determination}
David~G. Cantor and W.~H. Mills.
\newblock Determination of a subset from certain combinatorial properties.
\newblock \emph{Canadian Journal of Mathematics}, 18:\penalty0 42–48, 1966.
\newblock \doi{10.4153/CJM-1966-007-2}.

\bibitem[Chakraborty et~al.(2016)Chakraborty, Goldenberg, and
  Kouck{\'{y}}]{chakraborty2016streaming}
Diptarka Chakraborty, Elazar Goldenberg, and Michal Kouck{\'{y}}.
\newblock Streaming algorithms for embedding and computing edit distance in the
  low distance regime.
\newblock In Daniel Wichs and Yishay Mansour, editors, \emph{Proceedings of the
  48th Annual {ACM} {SIGACT} Symposium on Theory of Computing, {STOC} 2016,
  Cambridge, MA, USA, June 18-21, 2016}, pages 712--725. {ACM}, 2016.
\newblock \doi{10.1145/2897518.2897577}.
\newblock URL \url{https://doi.org/10.1145/2897518.2897577}.

\bibitem[Charikar and Krauthgamer(2006)]{charikar2006embedding}
Moses Charikar and Robert Krauthgamer.
\newblock Embedding the ulam metric into \emph{l}\({}_{\mbox{1}}\).
\newblock \emph{Theory Comput.}, 2\penalty0 (11):\penalty0 207--224, 2006.
\newblock \doi{10.4086/toc.2006.v002a011}.
\newblock URL \url{https://doi.org/10.4086/toc.2006.v002a011}.

\bibitem[Charikar et~al.(2018)Charikar, Geri, Kim, and
  Kuszmaul]{charikar2018estimating}
Moses Charikar, Ofir Geri, Michael~P. Kim, and William Kuszmaul.
\newblock On estimating edit distance: Alignment, dimension reduction, and
  embeddings.
\newblock In Ioannis Chatzigiannakis, Christos Kaklamanis, D{\'{a}}niel Marx,
  and Donald Sannella, editors, \emph{45th International Colloquium on
  Automata, Languages, and Programming, {ICALP} 2018, July 9-13, 2018, Prague,
  Czech Republic}, volume 107 of \emph{LIPIcs}, pages 34:1--34:14. Schloss
  Dagstuhl - Leibniz-Zentrum f{\"{u}}r Informatik, 2018.
\newblock \doi{10.4230/LIPIcs.ICALP.2018.34}.
\newblock URL \url{https://doi.org/10.4230/LIPIcs.ICALP.2018.34}.

\bibitem[Chen and Ng(2004)]{chen2004marriage}
Lei Chen and Raymond~T. Ng.
\newblock On the marriage of $l_p$-norms and edit distance.
\newblock In Mario~A. Nascimento, M.~Tamer {\"{O}}zsu, Donald Kossmann,
  Ren{\'{e}}e~J. Miller, Jos{\'{e}}~A. Blakeley, and K.~Bernhard Schiefer,
  editors, \emph{(e)Proceedings of the Thirtieth International Conference on
  Very Large Data Bases, {VLDB} 2004, Toronto, Canada, August 31 - September 3
  2004}, pages 792--803. Morgan Kaufmann, 2004.
\newblock \doi{10.1016/B978-012088469-8.50070-X}.
\newblock URL \url{http://www.vldb.org/conf/2004/RS21P2.PDF}.

\bibitem[Cohen et~al.(2019)Cohen, Rosenfeld, and
  Kolter]{DBLP:conf/icml/CohenRK19}
Jeremy~M. Cohen, Elan Rosenfeld, and J.~Zico Kolter.
\newblock Certified adversarial robustness via randomized smoothing.
\newblock In Kamalika Chaudhuri and Ruslan Salakhutdinov, editors,
  \emph{Proceedings of the 36th International Conference on Machine Learning,
  {ICML} 2019, 9-15 June 2019, Long Beach, California, {USA}}, volume~97 of
  \emph{Proceedings of Machine Learning Research}, pages 1310--1320. {PMLR},
  2019.
\newblock URL \url{http://proceedings.mlr.press/v97/cohen19c.html}.

\bibitem[Coja{-}Oghlan et~al.(2020)Coja{-}Oghlan, Gebhard, Hahn{-}Klimroth, and
  Loick]{coja2020optimal}
Amin Coja{-}Oghlan, Oliver Gebhard, Max Hahn{-}Klimroth, and Philipp Loick.
\newblock Optimal group testing.
\newblock In Jacob~D. Abernethy and Shivani Agarwal, editors, \emph{Conference
  on Learning Theory, {COLT} 2020, 9-12 July 2020, Virtual Event [Graz,
  Austria]}, volume 125 of \emph{Proceedings of Machine Learning Research},
  pages 1374--1388. {PMLR}, 2020.
\newblock URL \url{http://proceedings.mlr.press/v125/coja-oghlan20a.html}.

\bibitem[Cormode(2003)]{cormode2003sequence}
Graham Cormode.
\newblock \emph{Sequence distance embeddings}.
\newblock PhD thesis, University of Warwick, Coventry, {UK}, 2003.
\newblock URL \url{http://wrap.warwick.ac.uk/61310/}.

\bibitem[Dorfman(1943)]{dorfman1943detection}
Robert Dorfman.
\newblock The detection of defective members of large populations.
\newblock \emph{The Annals of Mathematical Statistics}, 14\penalty0
  (4):\penalty0 436--440, 1943.

\bibitem[Eiter and Mannila(1994)]{eiter1994computing}
Thomas Eiter and Heikki Mannila.
\newblock Computing discrete fr{\'e}chet distance.
\newblock Technical Report CD-TR 94/64, Christian Doppler Laboratory for Expert
  Systems, TU Vienna, Austria, April 1994.
\newblock URL
  \url{http://www.kr.tuwien.ac.at/staff/eiter/et-archive/cdtr9464.pdf}.

\bibitem[Fernandez et~al.(2019)Fernandez, Woodruff, and
  Yasuda]{fernandez2019query}
Manuel Fernandez, David~P. Woodruff, and Taisuke Yasuda.
\newblock The query complexity of mastermind with l\({}_{\mbox{p}}\) distances.
\newblock In Dimitris Achlioptas and L{\'{a}}szl{\'{o}}~A. V{\'{e}}gh, editors,
  \emph{Approximation, Randomization, and Combinatorial Optimization.
  Algorithms and Techniques, {APPROX/RANDOM} 2019, September 20-22, 2019,
  Massachusetts Institute of Technology, Cambridge, MA, {USA}}, volume 145 of
  \emph{LIPIcs}, pages 1:1--1:11. Schloss Dagstuhl - Leibniz-Zentrum f{\"{u}}r
  Informatik, 2019.
\newblock \doi{10.4230/LIPIcs.APPROX-RANDOM.2019.1}.
\newblock URL \url{https://doi.org/10.4230/LIPIcs.APPROX-RANDOM.2019.1}.

\bibitem[Fredman and Willard(1990)]{fredman1990blasting}
Michael~L. Fredman and Dan~E. Willard.
\newblock {BLASTING} through the information theoretic barrier with {FUSION}
  {TREES}.
\newblock In Harriet Ortiz, editor, \emph{Proceedings of the 22nd Annual {ACM}
  Symposium on Theory of Computing, May 13-17, 1990, Baltimore, Maryland,
  {USA}}, pages 1--7. {ACM}, 1990.
\newblock \doi{10.1145/100216.100217}.
\newblock URL \url{https://doi.org/10.1145/100216.100217}.

\bibitem[Hu et~al.(2023{\natexlab{a}})Hu, Li, Woodruff, Zhang, and
  Zhang]{hu2023recoveryTIT}
Zhuangfei Hu, Xinda Li, David~P. Woodruff, Hongyang Zhang, and Shufan Zhang.
\newblock Recovery from non-decomposable distance oracles.
\newblock \emph{IEEE Transactions on Information Theory}, pages 1--1,
  2023{\natexlab{a}}.
\newblock \doi{10.1109/TIT.2023.3289981}.

\bibitem[Hu et~al.(2023{\natexlab{b}})Hu, Li, Woodruff, Zhang, and
  Zhang]{hu_et_al:LIPIcs.ITCS.2023.73}
Zhuangfei Hu, Xinda Li, David~P. Woodruff, Hongyang Zhang, and Shufan Zhang.
\newblock {Recovery from Non-Decomposable Distance Oracles}.
\newblock In Yael Tauman~Kalai, editor, \emph{14th Innovations in Theoretical
  Computer Science Conference (ITCS 2023)}, volume 251 of \emph{Leibniz
  International Proceedings in Informatics (LIPIcs)}, pages 73:1--73:22,
  Dagstuhl, Germany, 2023{\natexlab{b}}. Schloss Dagstuhl -- Leibniz-Zentrum
  f{\"u}r Informatik.
\newblock ISBN 978-3-95977-263-1.
\newblock \doi{10.4230/LIPIcs.ITCS.2023.73}.
\newblock URL \url{https://drops.dagstuhl.de/opus/volltexte/2023/17576}.

\bibitem[Jiang and Polyanskii(2019)]{jiang2019metric}
Zilin Jiang and Nikita Polyanskii.
\newblock On the metric dimension of cartesian powers of a graph.
\newblock \emph{J. Comb. Theory, Ser. {A}}, 165:\penalty0 1--14, 2019.
\newblock \doi{10.1016/j.jcta.2019.01.002}.
\newblock URL \url{https://doi.org/10.1016/j.jcta.2019.01.002}.

\bibitem[Khot and Naor(2005)]{khot2005nonembeddability}
Subhash Khot and Assaf Naor.
\newblock Nonembeddability theorems via fourier analysis.
\newblock In \emph{46th Annual {IEEE} Symposium on Foundations of Computer
  Science {(FOCS} 2005), 23-25 October 2005, Pittsburgh, PA, USA, Proceedings},
  pages 101--112. {IEEE} Computer Society, 2005.
\newblock \doi{10.1109/SFCS.2005.54}.
\newblock URL \url{https://doi.org/10.1109/SFCS.2005.54}.

\bibitem[Knuth(1976)]{knuth1976computer}
Donald~E Knuth.
\newblock The computer as {Master Mind}.
\newblock \emph{Journal of Recreational Mathematics}, 9\penalty0 (1):\penalty0
  1--6, 1976.

\bibitem[Krauthgamer and Rabani(2009)]{krauthgamer2009improved}
Robert Krauthgamer and Yuval Rabani.
\newblock Improved lower bounds for embeddings
  intol\({}_{\mbox{1}}\){\textdollar}.
\newblock \emph{{SIAM} J. Comput.}, 38\penalty0 (6):\penalty0 2487--2498, 2009.
\newblock \doi{10.1137/060660126}.
\newblock URL \url{https://doi.org/10.1137/060660126}.

\bibitem[Kremer et~al.(1995)Kremer, Nisan, and Ron]{indexGame}
Ilan Kremer, Noam Nisan, and Dana Ron.
\newblock On randomized one-round communication complexity.
\newblock In Frank~Thomson Leighton and Allan Borodin, editors,
  \emph{Proceedings of the Twenty-Seventh Annual {ACM} Symposium on Theory of
  Computing, 29 May-1 June 1995, Las Vegas, Nevada, {USA}}, pages 596--605.
  {ACM}, 1995.
\newblock \doi{10.1145/225058.225277}.
\newblock URL \url{https://doi.org/10.1145/225058.225277}.

\bibitem[Kuszmaul(2021)]{kuszmaul2021binary}
William Kuszmaul.
\newblock Binary dynamic time warping in linear time.
\newblock \emph{CoRR}, abs/2101.01108, 2021.
\newblock URL \url{https://arxiv.org/abs/2101.01108}.

\bibitem[L{\'{e}}cuyer et~al.(2019)L{\'{e}}cuyer, Atlidakis, Geambasu, Hsu, and
  Jana]{lecuyer2019certified}
Mathias L{\'{e}}cuyer, Vaggelis Atlidakis, Roxana Geambasu, Daniel Hsu, and
  Suman Jana.
\newblock Certified robustness to adversarial examples with differential
  privacy.
\newblock In \emph{2019 {IEEE} Symposium on Security and Privacy, {SP} 2019,
  San Francisco, CA, USA, May 19-23, 2019}, pages 656--672. {IEEE}, 2019.
\newblock \doi{10.1109/SP.2019.00044}.
\newblock URL \url{https://doi.org/10.1109/SP.2019.00044}.

\bibitem[Levenshtein(1966)]{levenshtein1966binary}
Vladimir~I Levenshtein.
\newblock Binary codes capable of correcting deletions, insertions, and
  reversals.
\newblock In \emph{Soviet physics doklady}, volume~10, pages 707--710, 1966.

\bibitem[Li and Vit{\'{a}}nyi(1991)]{li1991combinatorics}
Ming Li and Paul M.~B. Vit{\'{a}}nyi.
\newblock Combinatorics and kolmogorov complexity.
\newblock In \emph{Proceedings of the Sixth Annual Structure in Complexity
  Theory Conference, Chicago, Illinois, USA, June 30 - July 3, 1991}, pages
  154--163. {IEEE} Computer Society, 1991.
\newblock \doi{10.1109/SCT.1991.160256}.
\newblock URL \url{https://doi.org/10.1109/SCT.1991.160256}.

\bibitem[Rodr{\'{\i}}guez{-}Vel{\'{a}}zquez
  et~al.(2014)Rodr{\'{\i}}guez{-}Vel{\'{a}}zquez, Yero, Kuziak, and
  Oellermann]{rodriguez2014strong}
Juan~Alberto Rodr{\'{\i}}guez{-}Vel{\'{a}}zquez, Ismael~Gonz{\'{a}}lez Yero,
  Dorota Kuziak, and Ortrud~R. Oellermann.
\newblock On the strong metric dimension of {Cartesian} and direct products of
  graphs.
\newblock \emph{Discret. Math.}, 335:\penalty0 8--19, 2014.
\newblock \doi{10.1016/j.disc.2014.06.023}.
\newblock URL \url{https://doi.org/10.1016/j.disc.2014.06.023}.

\bibitem[Schaar et~al.(2020)Schaar, Froese, and
  Niedermeier]{DBLP:conf/cpm/SchaarFN20}
Nathan Schaar, Vincent Froese, and Rolf Niedermeier.
\newblock Faster binary mean computation under dynamic time warping.
\newblock In Inge~Li G{\o}rtz and Oren Weimann, editors, \emph{31st Annual
  Symposium on Combinatorial Pattern Matching, {CPM} 2020, June 17-19, 2020,
  Copenhagen, Denmark}, volume 161 of \emph{LIPIcs}, pages 28:1--28:13. Schloss
  Dagstuhl - Leibniz-Zentrum f{\"{u}}r Informatik, 2020.
\newblock \doi{10.4230/LIPIcs.CPM.2020.28}.
\newblock URL \url{https://doi.org/10.4230/LIPIcs.CPM.2020.28}.

\bibitem[Selberg(1949)]{selberg1949elementary}
Atle Selberg.
\newblock An elementary proof of the prime-number theorem.
\newblock \emph{Annals of Mathematics}, 50\penalty0 (2):\penalty0 305--313,
  1949.
\newblock ISSN 0003486X.
\newblock URL \url{http://www.jstor.org/stable/1969455}.

\bibitem[Shapiro and Fine(1960)]{shapiro1960e1399}
Harold~S Shapiro and NJ~Fine.
\newblock E1399.
\newblock \emph{The American Mathematical Monthly}, 67\penalty0 (7):\penalty0
  697--698, 1960.

\bibitem[Sima and Bruck(2021)]{sima2021trace}
Jin Sima and Jehoshua Bruck.
\newblock Trace reconstruction with bounded edit distance.
\newblock In \emph{{IEEE} International Symposium on Information Theory, {ISIT}
  2021, Melbourne, Australia, July 12-20, 2021}, pages 2519--2524. {IEEE},
  2021.
\newblock \doi{10.1109/ISIT45174.2021.9518244}.
\newblock URL \url{https://doi.org/10.1109/ISIT45174.2021.9518244}.

\bibitem[Soderberg and Shapiro(1963)]{soderberg1963combinatory}
Staffan Soderberg and H.~S. Shapiro.
\newblock A combinatory detection problem.
\newblock \emph{The American Mathematical Monthly}, 70\penalty0 (10):\penalty0
  1066--1070, 1963.
\newblock ISSN 00029890, 19300972.
\newblock URL \url{http://www.jstor.org/stable/2312835}.

\bibitem[Sunjaya and Sunjaya(2020)]{sunjaya2020pooled}
Angela~Felicia Sunjaya and Anthony~Paulo Sunjaya.
\newblock Pooled testing for expanding covid-19 mass surveillance.
\newblock \emph{Disaster Medicine and Public Health Preparedness}, 14\penalty0
  (3):\penalty0 e42--e43, 2020.

\bibitem[Vershynin(2011)]{vershynin2011lectures}
Roman Vershynin.
\newblock Lectures in geometric functional analysis.
\newblock \emph{Unpublished manuscript. Available at http://www-personal.
  umich. edu/romanv/papers/GFA-book/GFA-book. pdf}, 3\penalty0 (3):\penalty0
  3--3, 2011.

\bibitem[Wang et~al.(2018)Wang, Zhao, and Chuah]{wang2017optimal}
Chao Wang, Qing Zhao, and Chen{-}Nee Chuah.
\newblock Optimal nested test plan for combinatorial quantitative group
  testing.
\newblock \emph{{IEEE} Trans. Signal Process.}, 66\penalty0 (4):\penalty0
  992--1006, 2018.
\newblock \doi{10.1109/TSP.2017.2780053}.
\newblock URL \url{https://doi.org/10.1109/TSP.2017.2780053}.

\bibitem[Yelin et~al.(2020)Yelin, Aharony, Tamar, Argoetti, Messer, Berenbaum,
  Shafran, Kuzli, Gandali, Shkedi, Hashimshony, Mandel-Gutfreund, Halberthal,
  Geffen, Szwarcwort-Cohen, and Kishony]{yelin2020evaluation}
Idan Yelin, Noga Aharony, Einat~Shaer Tamar, Amir Argoetti, Esther Messer, Dina
  Berenbaum, Einat Shafran, Areen Kuzli, Nagham Gandali, Omer Shkedi, Tamar
  Hashimshony, Yael Mandel-Gutfreund, Michael Halberthal, Yuval Geffen, Moran
  Szwarcwort-Cohen, and Roy Kishony.
\newblock {Evaluation of COVID-19 RT-qPCR Test in Multi sample Pools}.
\newblock \emph{Clinical Infectious Diseases}, 71\penalty0 (16):\penalty0
  2073--2078, 05 2020.
\newblock ISSN 1058-4838.
\newblock \doi{10.1093/cid/ciaa531}.
\newblock URL \url{https://doi.org/10.1093/cid/ciaa531}.

\end{thebibliography}

\newpage
\appendix

\section{Other Related Work}

\vspace{0.5em}

\noindent
\textbf{Recovery problems in metric spaces.}
Our problem is related to the recovery or reconstruction problems over metric spaces.
\citep{amir2018period} study the \emph{period recovery problem} on strings, which is to find the primitive periods between two strings such that the periodic distance is below a threshold. 
They present an $\mc{O}(n \log n)$-time algorithm for Hamming distance and an $\mc{O}(n^{4/3})$-time algorithm for edit distance.
\citep{sima2021trace} investigate the approximate recovery problem over bounded edit distance spaces in the presence of noise and show $n^{\mc{O}(k)}$ noisy samples suffice for (approximate) reconstruction.
Interestingly, \citep{bressan2021exact} consider the \emph{exact} recovery problem using oracle queries but the objective of their work is to exactly recover the clusters in Euclidean space, which is similar but orthogonal to our problem.

\vspace{0.5em}

\noindent
\textbf{Learning problems: Coin-weighing and group-testing problems.}
The related ``decomposable'' instance to our problem of querying a Hamming distance oracle is equivalent to the coin-weighing problem \citep{bshouty2009optimal} and the quantitative group testing problem \citep{wang2017optimal}.
Both the coin-weighing problem and group-testing problems are well-studied learning problems in the literature and have many real-world applications \citep{soderberg1963combinatory,cantor1966determination,li1991combinatorics,sunjaya2020pooled,yelin2020evaluation}.
The coin-weighing problem is to determine the weight of each coin (of two distinct weights $w_1$ and $w_2$) by using a minimal number of weighings of a subset of $n$ total coins each time.
\citep{cantor1966determination} and \citep{bshouty2009optimal} respectively present $2n/\log n$ weighing solutions which are optimal non-adaptive solutions to this problem.
Assuming the number of $w_1$ weight coins is known to be $d$, this $d$-coin weighing problem can be solved by an adaptive algorithm in time $2d\log \frac{n}{d}/\log d + \mc{O}(d/\log d + d(\log \log d) \log \frac{n}{d}/(\log d)^2)$ \citep{bshouty2009optimal}.
The major difference between our problem and these well-studied problems is that we consider distance metrics which cannot be aligned and represented as $\sum_i^n f(x_i - y_i)$ (i.e., the edit distance, DTW distance, and \frechet distances).

\section{Coordinate Descent Algorithm Instantiation}
\label{sec: Local Search Algorithm Instantiation}

Now we briefly discuss how we apply the Coordinate Descent algorithm to all three distances we consider in this paper by justifying the two conditions hold.

\stitle{Edit distance.}
For condition 2, we know that \revise{$\forall s, q, \dist(s, q) \leq n$} since the maximum length of $s$ or $q$ is $n$.
For condition 1, in each iteration, we consider a set $Q$ that contains all sequences that can be transformed from \revise{$q$} by inserting, deleting or substituting one character in \revise{$q$} (edit operations).
Note that $|Q|$ cannot exceed $(n+1) + n +n= 3n+1$.
We claim that there exists a \revise{$q'$} in $Q$ such that $\revise{\dist(s, q) > \revise{\dist}(s, q')}$.
Let $\revise{\dist(s, q) = d}$. 
By the definition of edit distance, there exists a chain of edit operations of length $d$ that transforms \revise{$s$} to $q$, resulting in a list of intermediate sequences $q_1, ..., q_{d-1}$.
Note that $\revise{\dist(s, q) \geq \dist(q_1, q) + 1}$, otherwise we have $\revise{\dist(q_1, q) > d - 1}$.
However, the chain implies we can transform $q_1$ to $q$ in $d-1$ edit operations, which leads to a contradiction.
Since $q_1 \in Q$, we can find $q_1$ satisfying the condition in $3n+1$ searches.
Therefore, the algorithm is guaranteed to recover the input in $\mc{O}(n^2)$ steps.

\stitle{DTW distance.}
For DTW distance, condition 2 holds since \revise{$\forall s, q, \dist(s, q) \leq n$}.
For condition 1, consider the $\#\textsc{runs}(x)$ in $s$ and $q$. If 
$\#\textsc{runs}(x)$ of $q$ $<$ $s$, then either adding an (arbitrary length) run to the start or the end of $q$ will decrease the DTW distance from $s$. On the other hand, if 
$\#\textsc{runs}(x)$ of $q$ $>$ $s$, then either deleting a run from the start or the end of $q$ will decrease the DTW distance from $s$.
If $\#\textsc{runs}(x)$ of $q$ $=$ $s$ and $\revise{\dist}(s, q) \neq 0$, we can still decrease the distance from $q$ by either adding/deleting a run to the start/end of the sequence.
Therefore, the algorithm is guaranteed to recover the input in $\mc{O}(n^2)$ steps.

\stitle{\frechet distance.}
Condition 2 holds since \revise{$\forall s, q, \dist(s, q) \leq 1$}.
For condition 1, enumerating $2n$ non-equivalent sequences, (i.e., 010101... and 101010...) guarantees to find \revise{$q'$} such that $\revise{\dist(s, q) > \revise{\dist}(s, q') = 0}$.
Therefore, the algorithm terminates in $\mc{O}(n)$ steps.

\section{Proofs of Claim \ref{claim:2_o1} in Lemma \ref{lemma:single-direction_o1} }

\noindent \textit{Proof of Claim \ref{claim:2_o1}}. We prove this by contradiction. Suppose $\exists(i,j)$ where $i\in[n]$ and $j\in[\ell]$ such that  $\revise{\deg}(q[i])>1$ and $\revise{\deg}(s[j])>1$. 
Let $\mb{X} = \{x \in [n] \mid \revise{\deg}(q[x]) > 1\}$, $\mb{Y} = \{y \in [\ell] \mid \revise{\deg}(s[y]) > 1 \}$ and $\mb{Z} = \{z \in [\ell] \mid \exists x \in \mb{X} \text{ such that edge }(q[x], s[z])\in M\}$. 
According to Claim \ref{claim:1_o1}, we know that $\mb{Y}\bigcap \mb{Z} = \emptyset$. Let $d = \min_{y\in \mb{Y}, z\in \mb{Z}}|y-z|$. We would have $d>0$. Suppose we have $x_0\in \mb{X}, y_0\in \mb{Y}, z_0\in \mb{Z}$ such that edge $(q[x_0], s[z_0])\in M$ and $|y_0-z_0| = d$.
There are two cases to discuss. As we have already solved the case $y_0<z_0$ in the proofs of Claim \ref{claim:2_o1} in Lemma \ref{lemma:single-direction_o1}, here we only discuss the case $y_0 > z_0$.

\ifarxiv
In this case, we can assume that $s[y_0]$ is matched to 
$\{q[w], q[w+1], ..., q[w+\revise{\deg}(s[y_0])-1]\}$ and $q[x_0]$ is matched to $\{s[z_0-\revise{\deg}(q[x_0])+1], s[z_0-\revise{\deg}(q[x_0])+2],..., s[z_0]\}.$ 
We remove $d+1$ edges $E = \{(q[x_0], s[z_0]), (q[x_0+1], s[z_0+1]), ..., (q[w], s[y_0])\}$and add $d$ new edges 
$$E'=\{(q[x_0+1], s[z_0]), (q[x_0+2], s[z_0+1]), ..., (q[w], s[y_0-1])\}$$ to construct a new matching $M'$. Since $\revise{\deg}(s[y_0])>1$ and $\revise{\deg}(q[x_0])>1$ in $M$, $M'$ would still be a valid matching. Computing the sum of two sets of edges $E$ and $E'$, respectively, would yield the following.
\begin{align}
   \Cost(E) &  = |s[y_0]-q[w]|+\sum_{i = 1}^{d}|s[z_0+i-1]-q[x_0+i-1]| \label{eqn_dbl_2} \\
    & > |q[w]-q[x_0]|+ \sum_{i = 1}^{d}|s[z_0+i-1]-q[x_0+i-1]|  \hfill \tag{Equation.~\ref{equ:single_direction_o1}} \\
    & = \sum_{i = 1}^{d}|q[x_0+i-1]-q[x_0+i]|+ \sum_{i = 1}^{d}|s[z_0+i-1]-q[x_0+i-1]|  \hfill \tag{Monotonicity of $q$}  \\
    & =  \sum_{i = 1}^{d}(|q[x_0+i-1]-q[x_0+i]|+|s[z_0+i-1]-q[x_0+i-1]|) \nonumber \\
    & \geq \sum_{i = 1}^{d}|q[x_0+i]-s[z_0+i-1]| \hfill \tag{Triangle Inequality}\\
    & = \revise{\Cost}(E'). \nonumber
\end{align}
\fi

Hence, $M'$ would be a better matching than $M$, a contradiction. Combining 1) and 2) completes the proof of Claim \ref{claim:2_o1}.

\section{Recovery Using Non-Adaptive DTW Distance Oracle with \texorpdfstring{$\mc{O}(n)$}{On} Extra Characters}
\label{app:dtw_on_extra}

\begin{theorem}[Non-adaptive Strategy for DTW Exact Recovery with $\mc{O}(n)$ Extra Characters]
\label{theorem:dtw_on}
Define a sequence of $n$ elements, each of which has $\mc{O}(\log n)$ bit complexity, as a query sequence.
For a binary alphabet $\bits$ and an input sequence $s := \bits^{\ell}$ where $0 \leq \ell \leq n$, there exists an algorithm to recover the input sequence $s$, given \revise{$4 \in \mc{O}(1)$} query sequences $\revise{\mc{Q}}$ and the $d_{\revise{\DTW}}(s, q)$ to each query sequence $q \in \revise{\mc{Q}}$.
\end{theorem}

We note that, if we remove the constraint of $\mc{O}(\log n)$ bit complexity, we can give a straightforward solution by leveraging the query string $q = \{ \frac{1}{(n+1)}, \frac{1}{(n+1)^2}, \dots, \frac{1}{(n+1)^n} \}$ to encode much more information in a single query. 
With the word RAM bit complexity requirement \citep{fredman1990blasting} on the queries though, namely that each entry fits into a single $O(\log n)$-bit word, such solutions are not allowed. 

\stitle{Proof of Theorem \ref{theorem:dtw_on}.}
Note that we can still use query sequences $0$ and $1$ to recover input sequences consisting of only $0$s or $1$s. For simplicity, we assume in the rest of the proof that the input sequence $s$ contains both $0$ and $1$ and let $s = s[1]s[2]...s[\ell]$.

We give our proof by constructing $2$ query sequences $q$ and $q'$ and presenting an algorithm to recover an input sequence $s$ from its DTW distance to these $2$ query sequences.

\noindent\textit{Query Sequences Construction.}
Let $P_{prime} = \{p_1, p_2, ..., p_n\}$ be the first $n$ primes not including $2$. By the prime number theorem \citep{selberg1949elementary}, we have that $p_n = \mc{O}(n\log n)$.  Note that for any prime number $p_i > 2$, $\exists 1\leq x_i < p_i$ such that $\frac{1}{4} < \frac{x_i}{p_i} < \frac{1}{2}$. We obtain $q$ by selecting such a $\frac{x_i}{p_i}$ for each $p_i\in P_{prime}$ and rearranging them in increasing order. Then we construct $q$ as $q = q[1]q[2]...q[n]$ where $\frac{1}{4} < q[1] < q[2] < ... < q[n] < \frac{1}{2}$. Let $q'[i] = 1-q[i], 1\leq i \leq n$, and let $q' = q'[1]q'[2]...q'[n]$. We would have $\frac{3}{4} > q'[1] > q'[2] > ... > q'[n] > \frac{1}{2}$. Since $p_n = \mc{O}(n\log n)$, it is easy to verify that each $q[i]$ and $q'[i]$ does have bit complexity  $\mc{O}(\log n)$.

According to Lemma \ref{lemma:single-direction_o1}, each element in $q$ would be involved exactly once in $d_{\revise{\DTW}}(q, s)$, and a similar argument would hold for $q'$. We hereby present an algorithm to determine the value of the matched element for each element in $q$, and the same algorithm can also be applied to $q'$.

\noindent\textit{Algorithm to determine matched elements for a query sequence.} Suppose $q[i] = \frac{x_{t_i}}{p_{t_i}}$ where $\{t_j\}$ is a permutation of $[n]$. Letting $m_i$ be the value matched to $q[i]$ in the optimal DTW matching for $(q, s)$ (different $m_i$s could correspond to the same element in $s$), $m_i \in\{0, 1\}$ , we would have
\begin{align*}
    d_{\revise{\DTW}}(q, s) = \sum_{i=1}^{n}|m_i-q[i]| =
\sum_{i=1}^{n}\left|\frac{m_{i}p_{t_i}-x_{t_i}}{p_{t_i}}\right| \\
=\frac{\sum_{i=1}^{n} \left( |m_{i}p_{t_i}-x_{t_i}| \cdot \Pi_{j\neq t_i}p_j \right) }{{\Pi_{i=1}^{n}p_{i}}}.
\end{align*}

Let $d_{\revise{\DTW}}(q, s) = \frac{u}{v}$, where $u$ and $v$ are co-primes.  
We have $u = \sum_{i=1}^{n}\left(|m_{i}p_{t_i}-x_{t_i}|\cdot\Pi_{j\neq t_i}p_j \right)$ and $v = \Pi_{i=1}^{n}p_{i}$. 
Consider $u \bmod p_{t_k}$ for a specific $k$. 
As each term in the summation has a factor $p_{t_k}$ except $|m_{k}p_{t_k}-x_{t_k}|\cdot\Pi_{j\neq t_k}p_j$, we have
$a \equiv |m_{k}p_{t_k}-x_{t_k}|\cdot\Pi_{j\neq t_k}p_j \mod{p_{t_k}}.$
Note that $p_{t_k}-x_{t_k}\not\equiv x_{t_k}\mod{p_{t_k}}$, so $(p_{t_k}-x_{t_k})\cdot \Pi_{j\neq t_k}p_j\not\equiv x_{t_k}\cdot \Pi_{j\neq t_k}p_j\bmod p_{t_k}$. Thus, we can determine $m_k$ by checking whether $(p_{t_k}-x_{t_k})\cdot \Pi_{j\neq t_k}p_j\equiv u\mod p_{t_k}$ or $x_{t_k}\cdot \Pi_{j\neq t_k}p_j\equiv u\mod p_{t_k}$.

Furthermore, we have the following claim for the optimal DTW matching between $q$, $q'$ and $s$.

\begin{claim}
\label{claim:3}
For any given input sequence $s$ and optimal DTW matching $M$ and $M'$ for $(q,s)$ and $(q', s)$ respectively, we have $\revise{\deg}(s_[i]) = 1$ in $M$ if $s[i]=1$ and $\revise{\deg}(s[i]) = 1$ in $M'$ if $s[i]=0$ . 
\end{claim}

\noindent\textit{Proof of claim.} We give a proof by contradiction. Given an optimal DTW matching $M$ for $(q,s)$, suppose $\exists 1\leq i\leq \ell$ such that $s[i] = 1$ and $\revise{\deg}(s[i]) > 1$. Suppose $s[i]$ is matched to $q[j], q[j+1], \ldots, q[j+\revise{\deg}(s[i])-1]$. 

First, we show that we can ``swap'' $s[i]$ with its neighboring element while maintaining optimality of the matching. If one of the neighboring elements of $s[i]$ is $1$, \revise{without loss of generality}, suppose $s[i+1] = 1$, then we can construct an alternate optimal matching $M^*$ where $\revise{\deg}(s[i]) = 1$ and $\revise{\deg}(s[i+1])>1$. According to Lemma \ref{lemma:single-direction_o1}, $s[i+1]$ cannot be matched with any of $q[j], q[j+1], \ldots, q[j+\revise{\deg}(s[i])-1]$ in $M$, otherwise there would exist $j+1\leq k\leq j+\revise{\deg}(s[i])-1$ such that $\revise{\deg}(q[k]) = 2$. Thus by matching $q[j+1], \ldots, q[j+\revise{\deg}(s[i])-1]$ to $s[i+1]$ instead of $s[i]$, we would obtain a new optimal matching $M^*$ where $\revise{\deg}(s[i]) = 1$ and $\revise{\deg}(s[i+1])>1$. 

As there exists at least one $0$ in $s$, we know that there exists an optimal DTW matching $M_0^{*}$ for $(q, s)$ where $\exists s[i]$ such that $s[i] = 1$, $\revise{\deg}(s[i]) > 1$ and one of the neighboring elements of $s[i]$ is $0$. \revise{Without loss of generality}, suppose $s[i+1] = 0$. Similarly, according to Lemma \ref{lemma:single-direction_o1}, $s[i+1]$ cannot be matched with any of $q[j], q[j+1], \ldots, q[j+\revise{\deg}(s[i])-1]$ in $M_0$. Here we construct a new matching $M_0'$ by matching $q[j+1], \ldots , q[j+\revise{\deg}(s[i])-1]$ to $s[i+1]$ instead of $s[i]$. Fig \ref{fig:shift_o1} illustrates an example of such a construction. Considering the total cost of differing edges in both matchings, we have 
$\sum_{k=j+1}^{j+\revise{\deg}(s[i])-1}|s[i]-q[k]|
   > \sum_{k=j+1}^{j+\revise{\deg}(s[i])-1}\frac{1}{2}
   > \sum_{k=j+1}^{j+\revise{\deg}(s[i])-1}|s[i+1]-q[k]|.$
Thus $M_0^*$ would be a better matching than $M_0$, causing a contradiction and thus finishing the proof. 
A similar proof can be derived for query sequence $q'$ and the case $s[i]=0$.\\

\noindent\textit{Algorithm to recover $s$.}  We now give an overall algorithm that recovers $s$ using the above algorithm and claim. 
Applying the above algorithm gives the matched elements of $q$ and $q'$. 
Let the matching result for $q$ and $q'$ be $m = m_1...m_n$ and $m' = m_1'...m_n'$, ($m_i, m_i' \in \{0, 1\}$) respectively. 
We break $m$ and $m'$ into blocks such that each block is the longest substring that contains either $0$ or $1$. By also breaking $s$ into such blocks, we know that $m$ has the same number of blocks as $s$ according to Lemma \ref{lemma:single-direction_o1}. Similarly $m'$ has the same number of blocks as $s$. Let $l$ be the number of blocks that $m$ and $m'$ have. Then we can represent $m$ and $m'$ as $m=A_{1}...A_{l}$ and $m'=B_{1}...B_{l}$. Note that if $A_i$ contains only $1$, then $s$ must have the same number of $1$'s in the \revise{\linelabel{i-th-2}$i$-th} block, otherwise there will be some $s[k]$ for which  $\revise{\deg}(s[k])>1$, which contradicts Claim \ref{claim:3}. Similarly, if $B_j$ contains only $0$, then $s$ has the same number of $0$'s in the \revise{\linelabel{i-th-3}$j$-th} block.
Then we can fully recover $s$ as $h(A_1)...h(A_l)$ where $h(X_i) = A_i$ if $X_i$ contains $1$ and $h(X_i) = B_i$ if $X_i$ contains $0$.
\hfill$\blacksquare$

\end{document}